\documentclass{lmcs}
\pdfoutput=1

\usepackage{lastpage}
\lmcsdoi{16}{4}{12}
\lmcsheading{}{\pageref{LastPage}}{}{}%
{Mar.~28,~2018}{Dec.~01,~2020}{}

\usepackage[utf8]{inputenc}


\usepackage{amsmath, amsthm, amssymb}%, amsfonts}
\usepackage{tikz-cd}
\usepackage{mathpartir}

\def\eg{\emph{e.g.}}

\newcommand{\dcalc}{$\mathtt{d}$}

\newcommand{\prim}{\tau}
\newcommand{\prdef}[4]{[#1\!\doteq\!#2,#3:#4]}
\newcommand{\mydef}{\!:=\!}
\newcommand{\smydef}{:=}
\newcommand{\fun}{\!\Rightarrow\!}
\newcommand{\sfun}{\Rightarrow}
\newcommand{\injl}[2]{[#1,:\!#2]}
\newcommand{\injr}[2]{[:\!#1,#2]}
\newcommand{\pleft}[1]{#1.1}
\newcommand{\pright}[1]{#1.2}
\newcommand{\case}[2]{[#1\,?\,#2]}
\newcommand{\myneg}{\neg}

\newcommand{\binop}[2]{\oplus(#1,#2)}
\newcommand{\prsumop}[2]{[#1\oplus#2]}
\newcommand{\prsumopd}[2]{[#1\oplus'#2]}

\newcommand{\binbop}[3]{\ensuremath{\oplus_{#1}(#2,#3)}}
\newcommand{\binbopd}[3]{\ensuremath{\oplus'_{#1}(#2,#3)}}

\newcommand{\gv}{\vdash\,}
\newcommand{\sgv}{\!\vdash\!}
\newcommand{\sngv}{\!\models\!}
\newcommand{\rd}{\to^{*}}
\newcommand{\srd}{\to}
\newcommand{\rdn}[1]{\to^{#1}}
\newcommand{\rdr}{\nabla}
\newcommand{\eqv}{=_{\lambda}}
\newcommand{\neqv}{\neq_{\lambda}}

\newcommand{\mrd}{\to_{:=}^*}
\newcommand{\smrd}{\to_{:=}}
\newcommand{\mrdn}[1]{\to_{:=}^{#1}}
\newcommand{\mrdr}{\nabla_{:=}}

\newcommand{\snur}{\to_{\myneg}}
\newcommand{\nur}{\to_{\myneg}^*}
\newcommand{\nurn}[1]{\to_{\myneg}^{#1}}
\newcommand{\dmrd}{\to_{=}}

\newcommand{\ax}{\emph{ax}}
\newcommand{\axm}{{\mathit(ax)}}
\newcommand{\mystart}{\emph{start}}
\newcommand{\mystartm}{\mathit{(start)}}
\newcommand{\weak}{\emph{weak}}
\newcommand{\weakm}{\mathit{(weak)}}
\newcommand{\myconv}{\emph{conv}}
\newcommand{\convm}{\mathit{(conv)}}
\newcommand{\absu}{\emph{abs$_U$}}
\newcommand{\absum}{\mathit{(abs_U)}}
\newcommand{\appl}{\emph{appl}}
\newcommand{\applm}{\mathit{(appl)}}
\newcommand{\abse}{\emph{abs$_E$}}
\newcommand{\absem}{\mathit{(abs_E)}}
\newcommand{\chin}{\emph{ch$_I$}}
\newcommand{\chinm}{\mathit{(ch_I)}}
\newcommand{\chba}{\emph{ch$_B$}}
\newcommand{\chbam}{\mathit{(ch_B)}}
\newcommand{\pdef}{\emph{def}}
\newcommand{\pdefm}{\mathit{(def)}}
\newcommand{\bprod}{\emph{prd}}
\newcommand{\bprodm}{\mathit{(prd)}}
\newcommand{\prl}{\emph{pr$_L$}}
\newcommand{\prlm}{\mathit{(pr_L)}}
\newcommand{\prr}{\emph{pr$_R$}}
\newcommand{\prrm}{\mathit{(pr_R)}}

\newcommand{\bsumm}{\mathit{(sum)}}

\newcommand{\injllm}{\mathit{(inj_L)}}

\newcommand{\injlrm}{\mathit{(inj_R)}}
\newcommand{\cased}{\emph{case}}
\newcommand{\casedm}{\mathit{(case)}}
\newcommand{\negate}{\emph{neg}}
\newcommand{\negatem}{\mathit{(neg)}}

\newcommand{\dexp}{\mathcal{E}}
\newcommand{\dnrm}{\mathcal{\bar{E}}}
\newcommand{\dexps}{\mathcal{\dot{E}}}
\newcommand{\dvar}{\mathcal{V}}
\newcommand{\nf}{\mathcal{N}}
\newcommand{\de}{\mathcal{D}}
\newcommand{\sn}[1]{{\mathcal{S}_{#1}}}
\newcommand{\ce}{\mathit{C}}
\newcommand{\dnf}[1]{\mathcal{N}^{=}_{#1}}

\newcommand{\free}{\mathit{FV}}
\newcommand{\dom}{\mathit{dom}}
\newcommand{\ran}{\mathit{ran}}
\newcommand{\nrm}[2]{\,\parallel\!\!#2\!\!\parallel_{#1}}

\newcommand{\gsub}[2]{\lbrack#2/#1\rbrack}
\newcommand{\sth}{\mid}
\newcommand{\union}{\cup}
\newcommand{\sz}{\mathit{S}}

\newcommand{\ff}{\mathbf{ff}}
\newcommand{\tr}{\mathbf{tt}}

\newcommand{\noarg}{(.)}
\newcommand{\nats}{N}
\newcommand{\pnats}{N^{>0}}

\newcommand{\negax}{\myneg}
\newcommand{\cast}{()}
\newcommand{\castin}[1]{()_{#1}^+}
\newcommand{\castout}[1]{()_{#1}^-}

\begin{document}

%\bibliographystyle{plainurl}

% Author macros::begin %%%%%%%%%%%%%%%%%%%%%%%%%%%%%%%%%%%%%%%%%%%%%%%%

\title[An extended type system with lambda-typed lambda-expressions]{An extended type system \texorpdfstring{\\}{} with lambda-typed lambda-expressions}
%\titlecomment{{\lsuper*}OPTIONAL comment concerning the title, \eg,
%  if a variant or an extended abstract of the paper has appeared elsewhere.}

\author[M.~Weber]{Matthias Weber}	%required
\address{Faculty of Computer Science, Technical University of Berlin}	%required
\email{mattnweber@gmail.com}  %optional
%\thanks{}	%optional
\begin{abstract}
We present the system \dcalc, an extended type system with lambda-typed lambda-expressions.
It is related to type systems originating from the Automath project.
\dcalc\ extends existing lambda-typed systems by an existential abstraction operator as well as propositional operators.
$\beta$-reduction is extended to also normalize negated expressions using a subset of the laws of classical negation,
hence \dcalc\ is normalizing both proofs and formulas which are handled uniformly as functional expressions.
\dcalc\ is using a reflexive type axiom for a constant $\prim$ to which no function can be typed.
Some properties are shown including confluence, subject reduction, uniqueness of types, strong normalization, and consistency.
We illustrate how, when using \dcalc, due to its limited logical strength, additional axioms must be added both for negation and for the mathematical structures whose deductions are to be formalized.
\end{abstract}

\maketitle
\section{Introduction and Overview}%
\label{overview}
In this paper, we will present an extended type system with lambda-typed lambda-expressions. Since such systems have received little attention we begin with some general remarks to provide some context.

Most type systems contain subsystems that can be classified as instances of \emph{pure type systems} (PTS) (\eg~\cite{Bar:93}).
As one of their properties, these systems use distinct operators to form dependent products and functional (i.e.~$\lambda$)-abstractions.
This reflects their underlying semantic distinction between the domains of functions and types.
In contrast, in the semantics of lambda-typed systems there is a single domain of (partial) functions and typing is a binary relation between total functions.
As a consequence a function participating in the type relation may play the role of an element or of a type.
Moreover, since in general the domain and range of the type relation are not disjoint a function can (and will usually) have a double role both as element and as type.
Therefore in lambda-typed calculi one has to separate three aspects of an operator: its functional interpretation, i.e.~the equivalence class of entities it represents, its role as a type, i.e.~the entities on the element side related to its entities in the type relation, and
its role as an element, i.e.~the entities on the type side related to its entities in the type relation.
From this point of view, the semantic distinction between dependent products and $\lambda$-abstractions can be reduced to the distinction between the type-role or the element-role of a single underlying function and therefore, when defining a calculus, a single operator is sufficient.

Type systems outside of PTS using a single operator for both dependent products and $\lambda$-abstractions (i.e.~using \emph{$\lambda$-structured types}) have been investigated in early type systems such as $\Lambda$~\cite{deBruijn94,Nederpelt73}, as well as more recent approaches~\cite{Guidi09,JFP05}. In particular~\cite{JFP05} introduces the single-binder based $\flat$-cube, a variant of the $\beta$-cube which does not keep uniqueness of types and studies $\lambda$-structured type variants of well-known systems within this framework.

Before we present the basic elements of \dcalc, we would like to point out its additional underlying semantic assumptions:
\begin{itemize}
\item The range of the type relation is a subset of its domain, or in other words, every type has a type.
\item The type relation includes entities typing to themselves, which allows for generating type hierarchies of arbitrary size from a finite set of distinct base entities. We discuss below why in \dcalc\ this assumption does not lead to well-known paradoxes.
\item There are no inequivalent types of an element, or in other words, the type relation is restricted to be a function.
This means that \dcalc\ needs to satisfy the type uniqueness property (see Section~\ref{type.confl}).
We will review this assumption in Section~\ref{related.unique}.
\end{itemize}
It is well known that types can be interpreted as propositions and elements as proofs~\cite{Howard69}.
In our semantic setting this analogy is valid and we will make use of it throughout the paper, for example when motivating and describing the roles of operators we will frequently use the viewpoint of propositions and proofs.

Note that in our setting the interest is in formalization of structured mathematical reasoning rather than in computational proof content.

Finally we would like to make a notational remark related to \dcalc\, but also to lambda-typed calculi in general.
Since an operator in \dcalc\ can be interpreted both as element and as type (or proposition), there is, in our view, a notational and naming dilemma.
For example, one and the same entity would be written appropriately as a typed lambda-abstraction $\lambda x\!:\!a.b$ in its role as an element and as a universal quantification $\forall x\!:\!a.b$ in its role as a proposition
and neither notation would adequately cover both roles.
Therefore, in this case we use a more neutral notation $[x:a]b$ called \emph{universal abstraction}.
As will be seen this notational issue also applies to other operators of \dcalc.

After these remarks we can now turn to the core of \dcalc\, which is the system $\lambda^{\lambda}$~\cite{PdG93}, a reconstruction of a variation~\cite{Nederpelt73} of $\Lambda$, modified with a reflexive type axiom with a constant $\prim$, see Table~\ref{lamlammod} for its type rules.
\begin{table}[!htb]
\fbox{
\begin{minipage}{0.96\textwidth}
\begin{align*}
\\[-8mm]
\mathit{(ax)}\;&\frac{}{\gv\prim:\prim}&
\mathit{(start)}\;\frac{\Gamma\gv a:b}{\Gamma,x:a\gv x:a}\qquad\qquad\qquad\qquad\\[1mm]
\mathit{(weak)}\;&\frac{\Gamma\gv a:b\qquad\Gamma\gv c:d}{\Gamma,x:c\gv a:b}&
\mathit{(conv)}\;\frac{\Gamma\gv a:b\qquad b\eqv c\qquad\Gamma\gv c:d}{\Gamma\gv a:c}\\[1mm]
\absum\;&\frac{\Gamma,x:a\gv b:c}{\Gamma\gv[x:a]b:[x:a]c}&
\applm\;\frac{\Gamma\gv a:[x:b]c\qquad\Gamma\gv d:b}{\Gamma\gv (a\,d):c\gsub{x}{d}}\qquad\;\\[-5mm]
\end{align*}
\end{minipage}
}
\caption{The kernel of~\dcalc: The system $\lambda^{\lambda}$~\cite{PdG93} modified with $\prim:\prim$\label{lamlammod}}
\end{table}
As usual, we use contexts $\Gamma=(x_1:a_1,\ldots,x_n:a_n)$ declaring types of distinct variables and a $\beta$-conversion induced congruence $\eqv$ on expressions.
$\Gamma,x:a$ denotes the extension of $\Gamma$ on the right by a binding $x:a$ where $x$ is a variable not yet declared in $\Gamma$. Furthermore if $n=0$ we omit writing contexts.
We use the notation $a\gsub{x}{b}$ to denote the substitution of all free occurrences of $x$ in $a$ by $b$.

Element-roles and type-roles can now be illustrated by two simple examples:
Let $C=[x:\prim]\prim$ be the constant function delivering $\prim$, for any given argument of type $\prim$.
In its role as a type $C$ corresponds to an implication (``$\prim$ implies $\prim$'') and its
proofs include the identity $I=[x:\prim]x$ over $\prim$, i.e.~$\gv I:C$ using essentially the rules \mystart~and \absu.
In its role as an element $C$ corresponds to a constant function and it will type to itself, i.e.~$\gv C:C$ using essentially the rules \ax~and \absu.
$I$, which as an element is obviously the identity function, in its role as a type corresponds to the proposition ``everything of type $\prim$ is true'' which should not have any elements (this is shown in Section~\ref{regular}).

More generally, in the rule \appl~the type-role of a universal abstraction ($[x:b]c$) can be intuitively understood as an infinite conjunction of instances ($c\gsub{x}{b_1}\wedge c\gsub{x}{b_2}\ldots$ where $b_i:b$ for all $i$).
Individual instances ($c\gsub{x}{b_i}$) can be projected by means of the application operator.
Intuitively, a proof of a universal abstraction ($[x:a]c$) must provide an infinite list of instances ($b\gsub{x}{a_1},b\gsub{x}{a_2},\ldots$ where $b\gsub{x}{a_i}:c\gsub{x}{a_i}$ for all $i$).
Consequently, in the rule \absu~a universal abstraction ($[x:a]b$ where $x:a\sgv b:c$) has the element-role of being a proof of another universal abstraction ($[x:a]c$).

As another example, the introduction and elimination rules for universal quantification can be derived almost trivially.
Note that in this and the following examples we write $[a\fun b]$ to denote $[x:a]b$ if $x$ does not occur free in $b$:
\begin{eqnarray*}
P:[a\fun b]\;\;\gv\;\;[x:[y:a](P\,y)]x&:&[[y:a](P\,y)\fun[y:a](P\,y)]\\{}
P:[a\fun b]\;\;\gv\;\;[x:a][z:[y:a](P\,y)](z\,x)&:&[x:a][[y:a](P\,y)\fun(P\,x)]
\end{eqnarray*}
When substituting $P$ by a constant function $[a\fun c]$ (where $\sgv c:b$) the types simplify to two variants of the modus ponens rule:
\begin{eqnarray*}
\gv\;\;[x:[a\fun c]]x&:&[[a\fun c]\fun[a\fun c]]\\{}
\gv\;\;[x:a][z:[a\fun c]](z\,x)&:&[a\fun[[a\fun c]\fun c]]
\end{eqnarray*}
In contrast to other type systems with $\lambda$-structured types, \dcalc\ is using a reflexive axiom (\ax).
This might seem very strange as the use of a reflexive axiom in combination with basic rules of PTS leads to paradoxes, \eg~\cite{Coquand86,Hurkens1995}, see also~\cite[Section~5.5]{Bar:93}.
However, as will be seen, in our setting of $\lambda$-structured types, the axiom \ax\ leads to a consistent system.
This is actually not very surprising since $\lambda^{\lambda}$ does not have an equivalent to the product rule used in PTS:\@
\begin{align*}
\mathit{(product)}\;\;&\frac{\Gamma\gv a:s_1\quad\Gamma,x:a\gv b:s_2}{\Gamma\gv(\Pi x:a.b):s_3}\quad\text{where}\;s_i\in S,\;\text{$S$ is a set of sorts}
\end{align*}
Adding such a rule, appropriately adapted, to the kernel of \dcalc\ would violate uniqueness of types~\cite{JFP05} and allow for reconstructing well-known paradoxes~\cite{dcalculus} (see also Section~\ref{related.unique}).
Without such a rule, functions from $\prim$ such as $I$ do not accept functional arguments such as $C$\footnote{In fact, $(I\;C)$ cannot be typed since $\sgv I:[x:\prim]\prim$, hence $I$ is expecting arguments of type $\prim$, but $\sgv C:C$ and $\prim\neq_{\lambda}C$}.
The lack of functions of type $\prim$ is the reason for achieving consistency when assuming $\prim:\prim$.

Unlike instances of PTS~(\eg~\cite{luo1989extended}), systems with $\lambda$-structured types have never been extended by existential or classical propositional operators.
While the kernel of \dcalc\ is sufficiently expressive to axiomatize basic mathematical structures (Sections~\ref{equality},~\ref{examples.nats})
the expressive and structuring  means of deductions can be enhanced by additional operators.
We begin with an operator that effectively provides for a \emph{deduction interface}, i.e.~a mechanism to hide details of interdependent deductions.
For this purpose, in analogy to a universal abstraction ($[x:a]b$), \dcalc\ introduces an \emph{existential abstraction} ($[x!a]b$) (see Table~\ref{existential} for its type rules).
\begin{table}[!htb]
\fbox{
\begin{minipage}{0.96\textwidth}
\begin{align*}
\\[-8mm]
\pdefm\;\;&\frac{\Gamma\sgv a: b\quad\Gamma\sgv c:d\gsub{x}{a}\quad\Gamma,x:b\sgv d:e}{\Gamma\sgv\prdef{x}{a}{c}{d}:[x!b]d}&
\absem\;\;&\frac{\Gamma,x:a\sgv b:c}{\Gamma\sgv[x!a]b:[x:a]c}\\[-5mm]
\end{align*}
\begin{align*}
\chinm\;\;&\frac{\Gamma\sgv a:[x!b]c}{\Gamma\sgv\pleft{a}:b}&
\chbam\;\;&\frac{\Gamma\sgv a:[x!b]c}{\Gamma\sgv\pright{a}:c\gsub{x}{\pleft{a}}}\\[-5mm]
\end{align*}
\end{minipage}
}
\caption{Type rules for existential abstractions\label{existential}}
\end{table}
The notation is intended to maximise coherence with universal abstraction.
The type-role of an existential abstraction ($[x!b]d$) can be intuitively understood as an infinite disjunction ($d\gsub{x}{b_1}\vee d\gsub{x}{b_2}\ldots$ where $b_i:b$ for all $i$).
A proof of an existential abstraction ($[x!b]d$) must prove one of the instances (say $d\gsub{x}{b_j}$), i.e.~it must provide an instance of the quantification domain ($b_j:b$) and an element proving the instantiated formula ($c:d\gsub{x}{b_j}$).
This is formalized in rule \pdef\ with a new operator called \emph{protected definition} combining the two elements and a tag for the abstraction type ($\prdef{x}{a}{c}{d}$)\footnote{The type tag $d$ in $\prdef{x}{a}{c}{d}$ is necessary to ensure uniqueness of types}.
This means that two deductions ($a$ and $c$), where one ($c$) is using the other one ($a$) in its type, are simultaneously abstracted into an existential abstraction type\footnote{In Section~\ref{related.definition}, we discuss the restriction that $x$ may not occur free in $c$ and how it could be relaxed.}.
The element-role of an existential abstraction ($[x!a]b$) is not the one of a logical operator but of an entity providing an infinite list of instances ($b\gsub{x}{a_1}$, $b\gsub{x}{a_2}$, $\ldots$ where $a_i:a$ for all $i$).
But, as we just discussed above with respect to typing of universal abstractions (\absu), this is sufficient to type it to a universal abstraction (\abse).
Hence the element-roles of universal and existential abstraction are equivalent.
Pragmatically the element-role of existential abstraction is less frequently used and of less importance.
This is the reason why the notation for existential abstractions is more ``type-oriented'' than that of universal abstraction.

In analogy to the type-elimination of universal abstraction, \emph{projections} ($\pleft{a}$ and $\pright{a}$
\footnote{A postfix notation has been chosen since it seems more intuitive for projection sequences.})
are introduced as type-eliminators for existential abstraction with obvious equivalence laws.
\[
\pleft{\prdef{x}{a}{b}{c}}\eqv a\qquad
\pright{\prdef{x}{a}{b}{c}}\eqv b
\]
The type rules \chin~and \chba~for projections are similar to common rules for $\Sigma$ types, \eg~\cite{luo1989extended}.

As an example, the introduction and elimination rules for existential quantification can be derived as follows (with $\Gamma=(P,Q:[a\fun b])$):
\begin{eqnarray*}
\Gamma\;\gv\;[x\!:\!a][z\!:\!(P\,x)]\prdef{y}{x}{z}{(P\,y)}&:&[x:a][(P\,x)\fun[y!a](P\,y)]\\[1mm]
\Gamma\;\gv\;[x\!:\![y_1!a](P y_1)][z\!:\![y_2\!:\!a][(P y_2)\fun(Q y_2)]]\\
\prdef{y_3}{\pleft{x}}{((z\:\pleft{x})\pright{x})}{(Q y_3)}&:&[[y_1!a](P y_1)\\
&&\;\;\fun[[y_2\!:\!a][(P y_2)\fun(Q y_2)]\fun[y_3!a](Q y_3)]]
\end{eqnarray*}
A more detailed explanation of these deductions is given in Sections~\ref{reduction} and~\ref{typing}.
Note how the second example illustrates in a nutshell the role of existential abstraction as a deduction interface.
It can be understood as the transformation of a deduction on the basis of a reference ($x$) to its interface ($[y_1!a](P y_1)$) followed by the creation of a new interface ($[y_3!a](Q y_3)$) hiding the transformation details (application of $z$ to the extracted interdependent elements $\pleft{x}$ and $\pright{x}$).
More applications will be shown in Section~\ref{examples}, in particular in Sections~\ref{example.partial} (Partial functions),~\ref{example.functions} (Defining functions from deductions), and~\ref{example.groups} (Proof structuring).

The type rule for existential abstractions (\abse) has the consequence that existential abstractions can now be instantiated as functions in the elimination rule for universal abstraction, i.e.~the rule \appl~can be instantiated as follows:
\[
(\mathit{appl})\gsub{a}{[x!a_1]a_2}\quad \frac{\Gamma\sgv [x!a_1]a_2:[x:b]c\quad \Gamma\sgv d:b}{\Gamma\sgv([x!a_1]a_2\,d):c\gsub{x}{d}}
\]
This motivates the extension of $\beta$-equality to existential abstractions, i.e.~we have
\[
([x:a]b\;c)\;\eqv\;b\gsub{x}{c}\;\eqv\;([x!a]b\;c)
\]
Note that these properties merely state that both abstractions have equivalent functional behaviour when applied to arguments.
However, they are not identical and this semantic distinction is represented by their role in the type relation.
For example, note that from $x:[y!a]b$ and $z:a$ one cannot conclude $(x\,z):([y!a]b\,z)$.
Note that this approach naturally precludes adding axioms of extensionality as they deny such distinctions.

%Note also that the extension of $\beta$-equality to existential abstraction precludes $\eta$-equality, i.e.~to uniquely determine a function by its value at each point, as this would directly lead to an inconsistency: When assuming $[x:a](b\, x)\eqv b$ for arbitrary $a$ and $b$ where $x$ not free in $b$ then
%\[
%[x:a](b\, x)\;\eqv\;[x:a]([x!a](b\, x)\, x)\;\eqv\;[x!a](b\, x)
%\]
Finally, one might wonder why not type an existential abstraction $[x!a]b$ to an existential abstraction $[x!a]c$ (assuming  $x:a\sgv b:c$)?
According to the intuitive understanding of the type role of an existential abstraction as an infinite disjunction this would be logically
invalid and indeed it is quite easy to see such a rule would lead to inconsistency\footnote{
First, with such a rule the type $P=[x:\prim][y!x]\prim$ would have itself as an element, i.e.~$\sgv P : P$.
However from a declaration of this type one could then extract a proof of $[z:\prim]z$ as follows
$y:P\sgv[z:\prim]\pleft{(y\,z)}:[z:\prim]z$.}.

The remaining part of \dcalc\ consists of some propositional operators (see Table~\ref{propositional} for their type rules).
\begin{table}[!htb]
\fbox{
\begin{minipage}{0.96\textwidth}
\begin{align*}
\\[-8mm]
\bprodm\;\;&\frac{\Gamma\sgv a:c\quad \Gamma\sgv b:d}{\Gamma\sgv[a,b]:[c,d]}\quad\;\;\;&
\bsumm\;\;&\frac{\Gamma\sgv a:c\quad \Gamma\sgv b:d}{\Gamma\sgv[a+b]:[c,d]}\\[0mm]
\prlm\;\;&\frac{\Gamma\sgv a:[b,c]}{\Gamma\sgv\pleft{a}:b}&
\prrm\;\;&\frac{\Gamma\sgv a:[b,c]}{\Gamma\sgv\pright{a}:c}\\[0mm]
\injllm\;\;&\frac{\Gamma\sgv a:b\quad\Gamma\sgv c:d}{\Gamma\sgv\injl{a}{c}:[b+c]}&
\injlrm\;\;&\frac{\Gamma\sgv a:b\quad\Gamma\sgv c:d}{\Gamma\sgv\injr{c}{a}:[c+b]}\\[-6mm]
\end{align*}
\begin{align*}
\casedm\;\;&\frac{\Gamma\sgv a:[x:c_1]d\quad\Gamma\sgv b:[y:c_2]d\quad\Gamma\sgv d:e}{\Gamma\sgv\case{a}{b}:[z:[c_1+c_2]]d}\qquad\qquad\qquad\qquad\quad\\[-7mm]
\end{align*}
\begin{align*}
\negatem\;\;&\frac{\Gamma\sgv a:b}{\Gamma\sgv\myneg a:b}&&\qquad\qquad\qquad\qquad\qquad\qquad\qquad\qquad\qquad\;\;\\[-5mm]
\end{align*}
\end{minipage}
}
\caption{Type rules for propositional operators\label{propositional}}
\end{table}
\dcalc\ adds a \emph{product} ($[a,b]$) as a binary variation of universal abstraction ($[x:a]b$), i.e.~with an element-role as a binary pair and an type-role as a binary conjunction.
Due to the same notation dilemma as for universal abstraction, i.e.~neither the notation for pairs $\langle a,b\rangle$ nor for conjunctions $a \wedge b$ would be satisfactory, we use the neutral notation $[a,b]$.
The type rules \bprod, \prr, and \prl~for products directly encode the introduction and elimination rules for conjunctions.
The equivalence laws for projection are extended in an obvious way.
\[
\pleft{[a,b]}\eqv a\qquad
\pright{[a,b]}\eqv b
\]
As a example we show that existential abstractions without dependencies (i.e.~$[x!a]b$, where $x$ does not occur free in $b$) are logically equivalent to products:
\begin{eqnarray*}
\gv\;\;[y:[x!a]b][\pleft{y},\pright{y}]&:&[[x!a]b\fun[a,b]]\\{}
\gv\;\;[y:[a,b]]\prdef{x}{\pleft{y}}{\pright{y}}{b}&:&[[a,b]\fun[x!a]b]
\end{eqnarray*}
%However, note that in comparison to existential abstractions, products have an intuitive symmetric type rule for this case (cf.~\abse\ vs.~\bprod).
%However, note that in comparison to existential abstractions, products have an intuitive symmetric type rule for this case (cf.~\abse\ vs.~\bprod).

Similarly \dcalc\ adds a sum ($[b_1+b_2]$) as a binary variation of the existential abstraction ($[x!a]b$), i.e.~with an element-role as a binary pair and an type-role as a binary disjunction.
Consequently, there is a type sequence analogous to existential abstraction: Expressions ($a_1:b_1$, $a_2:b_2$) may be used within \emph{injections} ($\injl{a_1}{b_2}$, $\injr{b_1}{a_2}$) (injections are the analogue to protected definitions), which type to a sum ($[b_1+b_2]$) which types to a product.
The bracket notation for products and abstractions is extended to injections and sums for notational coherence.
Similarly to existential abstraction we use a more type-oriented notation for sums.
The injection rules directly encode the introduction rules for disjunctions and rule \cased~introduces a case distinction operator.
Two equivalence laws describe its functional interpretation.
\[
(\case{a}{b}\,\injl{c}{d})\eqv(a\,c)\qquad
(\case{a}{b}\,\injr{c}{d})\eqv(b\,d)
\]
As an example, the introduction and elimination rules for disjunction can be derived as follows:
\begin{eqnarray*}
\gv\;\;[[x:a]\injl{x}{b},[x:a]\injr{b}{x}] &:&[[a\fun[a+b]],[a\fun[b+a]]]\\{}
\gv\;\;[x:[a+b]][y:[a\fun c]][z:[b\fun c]](\case{y}{z}\,x)&:&[[a +b]\fun[[a\fun c]\fun [[b\fun c]\fun c]]]
\end{eqnarray*}
In the second example, in analogy to existential abstractions, a reference ($x$) to an interface ($[a+b]$) is hiding the information on which particular proposition ($a$ or $b$) is proven.

\noindent
In analogy to the $\beta$-equality for existential abstraction, projection is extended to sums.
\[
\pleft{[a+b]}\eqv a\qquad\pright{[a+b]}\eqv a
\]
Finally, to support common mathematical reasoning practices,
\dcalc\ introduces a classical negation operator $\myneg a$ which has a neutral type rule \negate~and which defines an equivalence class of propositions w.r.t.~classical negation.
The central logical properties are double negation and De Morgan's laws:
\[
a\eqv\myneg\myneg a
\qquad
[a,b]\eqv\myneg[\myneg a +\myneg b]
\qquad
[x\!:\!a]b\eqv\myneg[x!a]\myneg b
\]
Furthermore, negation has no effect on operators to which no elimination operator can be applied when they are used as types
\[
\myneg\prim\eqv\prim
\qquad\myneg\prdef{x}{a}{b}{c}\eqv\prdef{x}{a}{b}{c}
\]
\[
\myneg\injl{a}{b}\eqv\injl{a}{b}
\qquad\myneg\injr{a}{b}\eqv\injr{a}{b}
\qquad\myneg\case{a}{b}\eqv\case{a}{b}
\]
The negation laws define many negated formulas as equivalent which helps to eliminate many routine applications of logical equivalences in deductions.
For example, the following laws can be derived for arbitrary well-typed expressions $a,b$.
\begin{eqnarray*}
\gv\;\;\;[x:a]x&:&[a\fun a]\;\eqv\;[\myneg\myneg a\fun a]\;\eqv\;[a\fun\myneg\myneg a]\\
\gv\;\;\;[x:\myneg[a,b]]x&:&[\myneg[a,b]\fun[\myneg a\!+\!\myneg b]]\;\eqv\;[[\myneg a\!+\!\myneg b]\fun\myneg[a,b]]\\
\gv\;\;\;[x:\myneg[x!a]b]x&:&[\myneg[x!a]b\,\fun\,[x\!:\!a]\myneg b]\;\eqv\;[[x\!:\!a]\myneg b\,\fun\,\myneg[x!a]b]
\end{eqnarray*}
Truth and falsehood can now be defined as follows:
\[
\ff := [x:\prim]x \qquad \tr := \myneg\ff
\]
Note that falsehood is just a new convenient notation for the type role of the identity ($\ff=I$).
The expected properties follow almost directly:
\begin{eqnarray*}
\gv\;\;[x:\prim][y:\ff](y\,x)&:&[x:\prim][\ff\fun x]\\
\gv\;\;\prdef{x}{\prim}{\prim}{\myneg x}&:&\tr
\end{eqnarray*}
In the proof of $\tr$ note that $\tr\eqv[x!\prim]\myneg x$ and $\prim :\prim\eqv\myneg\prim\eqv(\myneg x)\gsub{x}{\prim}$.

Equivalence rules for negation obviously do not yield all logical properties of negation, \eg~they are not sufficient to prove
$[a+\myneg a]$. Therefore one has to assume additional axioms.
Similarly, due to the limited strength of \dcalc, additional axioms must also be added for the mathematical structures whose deductions are to be formalized (for both see Section~\ref{examples}).

As this completes the overview of \dcalc, one may ask for its general advantages w.r.t.~PTS\@.
While essentially equivalent expressive means could probably also be defined in a more classical semantic setting, the purely functional approach of \dcalc\ can be considered as conceptually more simple.
Independently from the semantic setting the use of common logical quantifiers and propositional connectors including a classical negation with rich equivalence laws seems more suitable for modelling mathematical deductions than encoded operators or operators with constructive interpretation.

The remainder of this paper is structured as follows: A formal definition of \dcalc\ is presented in Section~\ref{definition} and several application examples are shown in Section~\ref{examples}. Readers with less focus on the theoretical results can well read Section~\ref{examples} before Section~\ref{definition}.
The main part of the paper is Section~\ref{properties} containing proofs of confluence, subject reduction, uniqueness of types, strong normalization, and consistency.
\section{Definition}%
\label{definition}
First we summarize the syntax (which has been motivated in Section~\ref{overview}) and define some basic notions such as free occurrences of variables.
We then define the congruence relation $a\eqv b$ as transitive and symmetric closure of a reduction relation $a\rd b$.
Finally we define contexts $\Gamma$ and the type relation $\Gamma\sgv a:b$.
\subsection{Basic definitions}%
\label{basic.definitions}
\begin{defi}[Expressions]%
\label{expression}
Let $\dvar=\{x,y,z,\ldots\}$ be an infinite set of variables.
The set of expressions $\dexp$ is generated by the following rules
\begin{eqnarray*}
\dexp&\!::=\!&\prim\,\mid\,\dvar\\
&&\,\mid\,[\dvar:\dexp]\dexp\,\mid\,(\dexp\,\dexp)\\
&&\,\mid\,[\dvar!\dexp]\dexp\,\mid\,\prdef{\dvar}{\dexp}{\dexp}{\dexp}\,\mid\,\pleft{\dexp}\,\mid\,\pright{\dexp}\\
&&\,\mid\,[\dexp,\dexp]\,\mid\,[\dexp+\dexp]\,\mid\,\injl{\dexp}{\dexp}\,\mid\,\injr{\dexp}{\dexp}\,\mid\,\case{\dexp}{\dexp}\,\mid\,\myneg\dexp
\end{eqnarray*}
Expressions will be denoted by $a,b,c,d,\ldots$,
\begin{itemize}
\item
$\prim$ is the \emph{primitive constant}, $x,y,z,\ldots\in\dvar$ are \emph{variables},
\item
$[x:a]b$ is a \emph{universal abstraction}, $(a\,b)$ is an \emph{application},
\item
$[x!a]b$ is an \emph{existential abstraction}, $\prdef{x}{a}{b}{c}$ is a \emph{protected definition}, $\pleft{a}$ and $\pright{a}$ are \emph{left-} and \emph{right-projection},
\item
$[a,b]$ is a \emph{product}, $[a+b]$ is a \emph{sum}, $\injl{a}{b}$ and $\injr{a}{b}$ are \emph{left-} and \emph{right-injection}, $\case{a}{b}$ is a \emph{case distinction}, and $\myneg a$ is a \emph{negation}.
\end{itemize}
%We use additional brackets to disambiguate expressions, \eg~$([x:\prim]x)(\prim)$.
For the sake of succinctness and homogenity in the following definitions we are using some additional notations for groups of operations:
\begin{eqnarray*}
\prsumop{a}{b}&\text{stands for}&[a,b]\;\text{or}\;[a+b]\\
\binop{a_1}{\ldots,a_n}&\text{stands for}&
\begin{cases}
\pleft{a_1},\pright{a_1},\text{or}\;\myneg a_1&\text{if}\;n=1\\
(a_1\,a_2),\prsumop{a_1}{a_2},\injl{a_1}{a_2},\injr{a_1}{a_2}\;\text{or}\;\case{a_1}{a_2}
&\text{if}\;n=2
\end{cases}\\
\binbop{x}{a_1}{\ldots,a_n}&\text{stands for}&
\begin{cases}
[x:a_1]a_2\;\text{or}\;[x!a_1]a_2&\text{if}\;n=2\\
\prdef{x}{a_1}{a_2}{a_3}&\text{if}\;n=3
\end{cases}
\end{eqnarray*}
If one of these notations is used more than once in an equation or inference rule it always denotes the same concrete operation.
\end{defi}
\noindent
\begin{defi}[Free variables, substitution]%
\label{free}
Variables occurring in an expression which do not occur in the range of a binding \emph{occur free} in the expression.
$\free(a)$ which denotes the set of \emph{free variables} of an expression $a$ is defined as follows:
\begin{eqnarray*}
\free(\prim)&=&\{\}\\
\free(x)&=&\{x\}\\
\free(\binop{a_1}{\ldots,a_n})&=&\free(a_1)\union\ldots\union\free(a_n)\\
\free(\binbop{x}{a_1}{\ldots,a_n}&=&\free(a_1)\union\ldots\union\free(a_{n-1})\union(\free(a_n)\!\setminus\!\{x\})
\end{eqnarray*}
Note that in a protected definition $\prdef{x}{a_1}{a_2}{a_3}$ the binding of $x$ is for $a_3$ only.
The \emph{substitution} $a\gsub{x}{b}$ of all free occurrences of variable $x$ in expression $a$ by expression $b$ is defined as follows:
\begin{eqnarray*}
\prim\gsub{x}{b}&=&\prim\\
y\gsub{x}{b}&=&\begin{cases}b&\text{if}\;x=y\\y&\text{otherwise}\end{cases}\\
\binop{a_1}{\ldots,a_n}\gsub{x}{b}&=&\binop{a_1\gsub{x}{b}}{\ldots,a_n\gsub{x}{b}}\\
\binbop{y}{a_1}{\ldots,a_n}\gsub{x}{b}&=&\begin{cases}\binbop{y}{a_1\gsub{x}{b}}{\ldots,a_{n-1}\gsub{x}{b},a_n}\\
           \qquad\text{if}\;x=y\\
           \binbop{y}{a_1\gsub{x}{b}}{\ldots,a_n\gsub{x}{b}}\\
           \qquad\text{otherwise}\end{cases}
\end{eqnarray*}
\end{defi}
\noindent
A substitution $a\gsub{x}{b}$ may lead to name clashes in case a variable $y$ occurring free in the inserted expression $b$ comes into the range of a binding of $y$ in the original expression. These name clashes can be avoided by renaming of variables.
\begin{defi}[$\alpha$-conversion, Implicit renaming, name-independent representations]%
\label{alpha}
The renaming relation on bound variables is usually called $\alpha$-conversion and induced by the following axiom (using the notation $=_{\alpha}$).
\[
\frac{y\notin\free(a_n)}{\binbop{x}{a_1}{\ldots,a_{n-1},a_n} \;=_{\alpha}\;\binbop{y}{a_1}{\ldots,a_{n-1},a_n\gsub{x}{y}}}
\]
In order not to clutter the presentation, we will write variables as strings but always assume appropriate renaming of bound variables in order to avoid name clashes. This assumption is justified because one could also use a less-readable but name-independent presentation of expressions using \eg~de Bruijn indexes~\cite{Bruijn72lambdacalculus} which would avoid the necessity of renaming all together.
\end{defi}
\subsection{Reduction and congruence}%
\label{reduction}
The functional interpretation is defined by the reduction of an expression into a more basic expression.
We will later show that reduction, if it terminates, always leads to a unique result.
\begin{defi}[Single-step reduction]%
\label{sred}
\emph{Single-step reduction} $a\srd b$ is the smallest relation satisfying the axioms and inference rules of Table~\ref{red.rules}.
\begin{table}[!htb]
\fbox{
\begin{minipage}{0.96\textwidth}
\begin{tabular}{@{$\;$}l@{$\;\;$}r@{$\;\;$}c@{$\;\;$}ll@{$\;\;$}r@{$\;\;$}c@{$\;\;$}l}
$\;$\\[-3mm]
$\mathit{(\beta_1)}$&$([x:a]b\;c)$&$\srd$&$b\gsub{x}{c}$&$\mathit{(\beta_2)}$&$([x!a]b\;c)$&$\srd$&$b\gsub{x}{c}$\\
$\mathit{(\beta_3)}$&$(\case{a}{b}\,\injl{c}{d})$&$\srd$&$(a\,c)$&$\mathit{(\beta_4)}$&$(\case{a}{b}\,\injr{c}{d})$&$\srd$&$(b\,d)$\\[2mm]
$\mathit{(\pi_1)}$&$\pleft{\prdef{x}{a}{b}{c}}$&$\srd$&$a$&$\mathit{(\pi_2)}$&$\pright{\prdef{x}{a}{b}{c}}$&$\srd$&$b$\\
$\mathit{(\pi_3)}$&$\pleft{[a,b]}$&$\srd$&$a$&$\mathit{(\pi_4)}$&$\pright{[a,b]}$&$\srd$&$b$\\
$\mathit{(\pi_5)}$&$\pleft{[a+b]}$&$\srd$&$a$&$\mathit{(\pi_6)}$&$\pright{[a+b]}$&$\srd$&$b$\\[2mm]
$\mathit{(\nu_1)}$&$\myneg\myneg a$&$\srd$&$a$\\
$\mathit{(\nu_2)}$&$\myneg[a,b]$&$\srd$&$[\myneg a+\myneg b]$&$\mathit{(\nu_3)}$&$\myneg[a+b]$&$\srd$&$[\myneg a,\myneg b]$\\
$\mathit{(\nu_4)}$&$\myneg[x:a]b$&$\srd$&$[x!a]\myneg b$&$\mathit{(\nu_5)}$&$\myneg[x!a]b$&$\srd$&$[x:a]\myneg b$\\
$\mathit{(\nu_6)}$&$\myneg\prim$&$\srd$&$\prim$&$\mathit{(\nu_7)}$&$\myneg\prdef{x}{a}{b}{c}$&$\srd$&$\prdef{x}{a}{b}{c}$\\
$\mathit{(\nu_8)}$&$\myneg\injl{a}{b}$&$\srd$&$\injl{a}{b}$&$\mathit{(\nu_9)}$&$\myneg\injr{a}{b}$&$\srd$&$\injr{a}{b}$\\
$\mathit{(\nu_{10})}$&$\myneg\case{a}{b}$&$\srd$&$\case{a}{b}$
\end{tabular}
\begin{align*}
\\[-6mm]
\mathit{(\oplus{\overbrace{(\_,\ldots,\_)}^{n}}_i)}\quad&\frac{a_i\srd b_i}{\binop{a_1,\ldots,a_i}{\ldots,a_n}\srd \binop{a_1,\ldots,b_i}{\ldots,a_n}}\\
\mathit{(\oplus_x{\overbrace{(\_,\ldots,\_)}^{n}}_i)}\quad&\frac{a_i\srd b_i}{\binbop{x}{a_1,\ldots,a_i}{\ldots,a_n}\srd \binbop{x}{a_1,\ldots,b_i}{\ldots,a_n}}\\[-4mm]
\end{align*}
\end{minipage}
}
\caption{Axioms and rules for single-step reduction.\label{red.rules}}
\end{table}
All reduction axioms have been motivated and explained in the introduction in form of equivalence laws.
\end{defi}
\begin{defi}[Reduction]
\emph{Reduction} $a\rd b$ of an expression $a$ to $b$ is defined as the reflexive and transitive closure of single-step reduction $a\srd b$.
The notation $a\rdn{n} b$ where $n\geq 0$ is used to indicate $n$ consecutive single-step reductions.
We use the notation $a\rd b\rd c\ldots$ to indicate reduction sequences.
If two expressions $a$ and $b$ reduce to a common expression we write $a\rdr b$.
To show arguments about equality in arguments about reduction we use the notation $a_1=\cdots=a_n\;\rd\;b_1=\cdots=b_m\:\rd\: c_1=\cdots=c_k\cdots$.
This will also be used for sequences of $n$-step reductions $\rdn{n}$ and accordingly for sequences containing both notations.
\end{defi}
\begin{defi}[Congruence]
\emph{Congruence} of expressions, denoted by $a\eqv b$, is defined as the symmetric and transitive closure of reduction.
The notations for reduction sequences are extended to contain congruences as well.
\end{defi}
\noindent
The following simple examples illustrate reduction and congruence:
\begin{eqnarray*}
(([y_2:a][y:(P\:y_2)](Q\:y_2)\:\pleft{x})\:\pright{x})
&\srd_{(\beta_1)}&([y:(P\:\pleft{x})](Q\:\pleft{x})\:\pright{x})\\
&\srd_{(\beta_1)}&(Q\:\pleft{x})\\{}
\pright{\pleft{[\prdef{x}{a}{b}{c},d]}}&\srd_{(\pi_3)}&\pright{\prdef{x}{a}{b}{c}}\\{}
&\srd_{(\pi_2)}&b\\{}
[x:\myneg[y:\prim]\prim]\myneg[a,b]&\srd_{(\nu_4)}&[x:[y!\prim]\myneg\prim]\myneg[a,b]\\{}
&\srd_{(\nu_6)}&[x:[y!\prim]\prim]\myneg[a,b]\\{}
&\srd_{(\nu_2)}&[x:[y!\prim]\prim][\myneg a+\myneg b]\\{}
[x:\myneg[a+b]][\myneg a,\myneg b]&\eqv&[x:[\myneg a,\myneg b]]\myneg[a+b]
\end{eqnarray*}
\subsection{Typing and Validity}%
\label{typing}
\begin{defi}[Context]
A \emph{context}, denoted by $\Gamma$, is a finite sequence of declarations $(x_1 :a_1,\ldots,x_n :a_n)$, where $x_i$ are variables with $x_i\neq x_j$ and $a_i$ are  expressions.
The assumption about name-free representation of bound variables justifies the uniqueness assumption.
The lookup of a variable in a context $\Gamma(x)$ is a partial function defined by $\Gamma(x_i)= a_i$.
$\dom(\Gamma)=\{x_1,\ldots,x_n\}$ and $\ran(\Gamma)=\{a_1,\ldots,a_n\}$ denote  the domain and range of a context $\Gamma$.
$\Gamma,x:a$ denotes the extension of $\Gamma$ on the right by a binding $x:a$ where $x$ is a variable not yet declared in $\Gamma$.
$\Gamma_1,\Gamma_2$ denotes the concatenation of two contexts declaring disjoint variables.
The empty context is written as $()$.
$[\Gamma]a= [x_1:a_1]\ldots[x_1:a_n]a$ denotes the abstraction of a context over an expression.
\end{defi}
\begin{defi}[Typing]
Typing $\Gamma\sgv a:b$ of $a$ to $b$, the \emph{type} of $a$, under a context $\Gamma$ is the smallest ternary relation on contexts and two expressions satisfying the inference rules of Table~\ref{typ.rules} all of which
have been motivated and explained in Section~\ref{overview}.
\begin{table}[!htb]
\fbox{
\begin{minipage}{0.96\textwidth}
\begin{align*}
\\[-9mm]
\axm\;\;&\sgv\prim:\prim&
\mystartm\;\;&\frac{\Gamma\sgv a:b}{\Gamma,x:a\sgv x:a}\\[0mm]
\weakm\;\;&\frac{\Gamma\sgv a:b\quad\Gamma\sgv c:d}{\Gamma,x:c\sgv a:b}\qquad&
\convm\;\;&\frac{\Gamma\sgv a:b\quad b\eqv c\quad\Gamma\sgv c:d}{\Gamma\sgv a:c}\\[0mm]
\absum\;\;&\frac{\Gamma,x:a\sgv b:c}{\Gamma\sgv[x:a]b:[x:a]c}&
\absem\;\;&\frac{\Gamma,x:a\sgv b:c}{\Gamma\sgv[x!a]b:[x:a]c}\\[0mm]
\applm\;\;&\frac{\Gamma\sgv a:[x:c]d\quad\Gamma\sgv b:c}{\Gamma\sgv(a\,b):d\gsub{x}{b}}&\\[-6mm]
\end{align*}
\begin{align*}
\pdefm\;\;&\frac{\Gamma\sgv a: b\quad\Gamma\sgv c:d\gsub{x}{a}\quad\Gamma,x:b\sgv d:e}{\Gamma\sgv\prdef{x}{a}{c}{d}:[x!b]d}\qquad\qquad\qquad\qquad\qquad\\[-7mm]
\end{align*}
\begin{align*}
\chinm\;\;&\frac{\Gamma\sgv a:[x!b]c}{\Gamma\sgv\pleft{a}:b}&
\chbam\;\;&\frac{\Gamma\sgv a:[x!b]c}{\Gamma\sgv\pright{a}:c\gsub{x}{\pleft{a}}}\qquad\\[1mm]
\bprodm\;\;&\frac{\Gamma\sgv a:c\quad \Gamma\sgv b:d}{\Gamma\sgv[a,b]:[c,d]}\quad\;\;\;&
\bsumm\;\;&\frac{\Gamma\sgv a:c\quad \Gamma\sgv b:d}{\Gamma\sgv[a+b]:[c,d]}\\[0mm]
\prlm\;\;&\frac{\Gamma\sgv a:[b,c]}{\Gamma\sgv\pleft{a}:b}&
\prrm\;\;&\frac{\Gamma\sgv a:[b,c]}{\Gamma\sgv\pright{a}:c}\\[0mm]
\injllm\;\;&\frac{\Gamma\sgv a:b\quad\Gamma\sgv c:d}{\Gamma\sgv\injl{a}{c}:[b+c]}&
\injlrm\;\;&\frac{\Gamma\sgv a:b\quad\Gamma\sgv c:d}{\Gamma\sgv\injr{c}{a}:[c+b]}\\[-6mm]
\end{align*}
\begin{align*}
\casedm\;\;&\frac{\Gamma\sgv a:[x:c_1]d\quad\Gamma\sgv b:[y:c_2]d\quad\Gamma\sgv d:e}{\Gamma\sgv\case{a}{b}:[z:[c_1+c_2]]d}\qquad\qquad\qquad\qquad\quad\\[-7mm]
\end{align*}
\begin{align*}
\negatem\;\;&\frac{\Gamma\sgv a:b}{\Gamma\sgv\myneg a:b}&&\qquad\qquad\qquad\qquad\qquad\qquad\qquad\qquad\qquad\;\;
\end{align*}
\end{minipage}
}
\caption{Axiom and rules for typing.\label{typ.rules}}
\end{table}
We often use the notation $\Gamma\sgv a_1=\cdots=a_n\::\: b_1=\cdots=b_m$ to indicate arguments about equality of expressions in proofs about the type relation.
Sometimes we also mix the use of $\eqv$ and $=$ in this notation.
\end{defi}
\noindent
As a simple example consider the introduction and elimination laws for existential quantification (see Section~\ref{overview}).
The type determination of the first law can be seen as follows: With $\Gamma_1=(P:[y:a]b,x:a,z:(P\,x))$) where $y\notin\free(b)$
and since $(P\,y)\gsub{y}{x}=P(\,x)$ by the rules \mystart, \weak, and \appl\ we obtain
\[
\Gamma_1\gv x:a\qquad\Gamma_1\gv z:(P\,y)\gsub{y}{x}\qquad\Gamma_1\gv (P\,y):b
\]
which due to the rule \pdef\ implies that
\begin{eqnarray*}
\Gamma_1&\gv&\prdef{y}{x}{z}{(P\,y)}:[y!a](P\,y)
\end{eqnarray*}
The type relation in Section~\ref{overview} follows from rule \absu.
The type determination of the second law can be seen as follows, let
\begin{eqnarray*}
\Gamma_2&=&(P:[y:a]b,Q:[y:a]b,x:[y_1!a](P\,y_1),z:[y_2\!:\!a][y:(P\,y_2)](Q\,y_2))
\end{eqnarray*}
By rules \mystart, \weak, and \appl\ we obtain:
\[
\Gamma_2\gv(z\,\pleft{x}):[u:(P\,\pleft{x})](Q\,\pleft{x})
\]
From this, by the same rules we obtain:
\[
\Gamma_2\gv((z\,\pleft{x})\:\pright{x}):(Q\,\pleft{x}) = (Q\,y_3)\gsub{y_3}{\pleft{x}}
\]
Similarly to the first law, by rule \pdef\ this implies
\begin{eqnarray*}
\Gamma_2&\gv&\prdef{y_3}{\pleft{x}}{((z\:\pleft{x})\:\pright{x})}{(Q\,y_3)}\;\;:\;\;[y_3!a](Q\,y_3)
\end{eqnarray*}
Again, the remainder follows from rule \absu.
\begin{defi}[Validity]
Validity $\Gamma\sgv a$ of an expression $a$ under a context $\Gamma$ is defined as the existence of a type:
\begin{eqnarray*}
\Gamma\gv a&&\text{if and only if there is an expression $b$ such that}\;\;\Gamma\sgv a:b
\end{eqnarray*}
Similarly to typing we use the notation $\Gamma_1=\cdots=\Gamma_n\sgv a_1=\cdots=a_n$ to indicate arguments about equality of contexts and expressions in proofs about validity. Similarly for congruence $\eqv$ or combinations of both.
We also use the notation $\Gamma\sgv a_1,\ldots,a_n$ as an abbreviation for $\Gamma\sgv a_1$, $\ldots$, $\Gamma\sgv a_n$.
As for typing, we also omit writing the empty context.
\end{defi}
\section{Examples}%
\label{examples}
The purpose of this section is to illustrate the basic style of axiomatizing theories and describing deductions when using~\dcalc.
Note that the example are presented with fully explicit expressions of~\dcalc, i.e.~we do not omit any subexpressions which could be inferred from other parts using pattern matching or proof tactics and we do not use a module concept for theories.
Such features should of course be part of any practically useful approach for formal deductions on the basis of \dcalc.
\begin{rem}[Notational conventions]
For convenience, in the examples we write $[x_1:a_1]\ldots [x_n:a_n]a$ as $[x_1:a_1;\ldots;x_n:a_n]a$ and $[x_1:a]\ldots[x_n:a]b$ as $[x_1,\ldots,x_n:a]b$ and similar for existential abstractions.
We also use combinations of these abbreviations such that \eg~$[x:a;y_1,y_2!b]c$ is shorthand for $[x:a][y_1!b][y_2!b]c$.
We write $[a\fun b]$ to denote $[x:a]b$ if $x\notin\free(b)$.
Similarly, we write $[a_1\fun\ldots[a_n\fun a]\ldots]$ as $[a_1;\ldots;a_n\fun a]$.
We write nested applications $(\ldots((a\,a_1)\,a_2)\ldots\,a_n)$ as $(a\,a_1\ldots a_n)$.
We also write $\Gamma\sgv a:b$ as $a:b$ if $\Gamma$ is empty or has been made clear from the context. When writing $a:b\;(\mathit{name})$ we introduce $name$ as abbreviation for $a$ of type $b$.
\end{rem}
\subsection{Additional axioms for negation}%
\label{examples.logic}
%The operators $[a\fun b]$, $[a,b]$, and $[a+b]$ have te show many properties of logical implication, conjunction, disjunction, and negation.
%We present some laws which follow directly from using an elimination rule:
%\begin{eqnarray*}
%\;[[x:[a,b]]\pleft{x},[x:[a,b]]\pright{x}]&:&[[[a,b]\fun a],[[a,b]\fun b]]\\
%\end{eqnarray*}
In Section~\ref{overview} we explained the necessity of additional negation axioms.
Note that, without additional axioms, even constructive laws of negation are missing, as for example, the contrapositive laws is not valid.
In principle many different axioms schemes are possible.
For example, one could be inspired by the rules for negation in sequent calculus ($\underline{A}$ stands for a sequence of formulas $A_1$, $\ldots$, $A_n$):
\[
\frac{\underline{C}\sgv A,\underline{B}}{\underline{C},\myneg A\sgv\underline{B}}\;(L\myneg)\qquad
\frac{\underline{C},A\sgv\underline{B}}{\underline{C}\sgv \myneg A,\underline{B}}\;(R\myneg)
\]
Theses rule together with the negations axioms inspire the following axiom schemes indexed over expressions $a$ and $b$:
\[
\negax^+_{a,b}:[[a+b]\fun[\myneg a\fun b]]\qquad\negax^-_{a,b}:[[a\fun b]\fun[\myneg a+b]]\qquad(*)
\]
where $\negax^+_{a,b},\negax^-_{a,b}\in\dvar$ are from an infinite set of variables $\mathit{I_{Ax}}$.
Formally, typing an expression $c$ to an expression $d$ under the above axiom schemes could be defined so as to require that $\free(d)=\emptyset$ and that there is a context $\Gamma$ consisting of (a finite sequence of) declarations of variables from $\mathit{I_{Ax}}$ such that $\Gamma\sgv c:d$.

\noindent
Assuming these axiom schemes we can now show the law of the excluded middle:
\begin{eqnarray*}
\;(\negax^-_{\myneg a,\myneg a}\;[x:\myneg a]x)&:&[a+\myneg a]
%[x:[a\fun b]](\negax^+_{b,\myneg a}\,(\mathit{sym}_{\myneg a,b}\,(\negax^-_{\myneg a,b}\,x)))&:&[[a\fun b]\fun[\myneg b\fun\myneg a]]\qquad\quad(\mathit{cp}_{a,b})
\end{eqnarray*}
Note that $\negax^-_{a,b}$ and $\negax^+_{\myneg a,b}$ state that universal abstractions
without dependencies (i.e.~$[a\fun b]$) are logically equivalent to the classical definition of implications:
\[
\negax^-_{a,b}:[[a\fun b]\,\fun[\myneg a+b]]\qquad
\negax^+_{\myneg a,b}:[[\myneg a+b]\fun[a\fun b]]
\]
On the basis of this equivalence it is easy to show the remaining laws of negation.
\subsection{Equality}%
\label{equality}
Basic axioms about an equality congruence relation on expressions of equal type can be formalized as context \emph{Equality}:
\begin{eqnarray*}
%Equality&:=&(\\
\mathit{Equality}:=\;(\;\noarg\!=_{\noarg}\!\noarg&:&[S:\prim][S;S\fun\prim]\\
E_1&:&[S:\prim;x:S]x\!=_{S}\!x\\
E_2&:&[S:\prim;x,y:S][x\!=_{S}\!y\:\fun\:y\!=_{S}\!x]\\
E_3&:&[S:\prim;x,y,z:S][x\!=_{S}\!y;y\!=_{S}\!z\:\fun\:x\!=_{S}\!z]\\
E_4&:&[S_1,S_2:\prim;x,y:S_1;F:[S_1\fun S_2]][x\!=_{S_1}\!y\:\fun\:(F\,x)\!=_{S_2}\!(F\,y)])
\end{eqnarray*}
Here $S$ is an variable used to abstract over the type of the expressions to be equal.
Note that, for better readability, we use the infix notation $x=_{S}y$, introduced by a declaration $\noarg\!=_{\noarg}\!\noarg:[S:\prim][S;S\fun S]$.
Equivalently we could have written $(\noarg\!=_{\noarg}\!\noarg\,S\,x\,y)$.
\subsection{Natural Numbers}%
\label{examples.nats}
Assuming the context \emph{Equality}, well-known axioms about naturals numbers including an induction principle can be formalized as context \emph{Naturals}:
\begin{eqnarray*}
%Naturals&:=&(\\
\mathit{Naturals}:=\quad(\quad \nats&:&\prim\\
0&:&\nats\\
s&:&[\nats\fun\nats]\\
\noarg\!+\!\noarg,\noarg\!*\!\noarg&:&[\nats;\nats\fun\nats]\\
S_1&:&[n:\nats]\myneg((s\,n)\!=_{\nats}\!0)\\
S_2&:&[n,m:\nats][(s\,n)\!=_{\nats}\!(s\,m)\:\fun\:n\!=_{\nats}\!m]\\
A_1&:&[n:\nats] 0+n=_{\nats}n\\
A_2&:&[n,m:\nats](s\,n)+m=_{\nats}(s\,(n+m))\\
M_1&:&[n:\nats]0*n=_{\nats}0\\
M_2&:&[n,m:\nats](s\,n)*m=_{\nats}m+(n*m)\\
\mathit{ind}&:&[P:[\nats\fun\prim]][(P\,0);[n:\nats][(P\,n)\fun(P\,(s\,n))]\fun[n:\nats](P\,n)]\quad)
\end{eqnarray*}
Note that, for better readability, we use infix notations $n+m$ and $n*m$ abbreviating the expressions $(+\,n\,m)$ and $(*\,n\,m)$.
Note that the induction principle ranges over propositions of type $\prim$ only and is thus weaker than the common one (see also Section~\ref{ex.casting}).

As an example, a simple property can be (tediously) deduced under the context $(\mathit{Equality},\mathit{Naturals},n:\nats)$ where $1:= (s\,0)$:
\begin{center}
\begin{tabular}{l@{$\;\;$}c@{$\;\;$}l}
$(E_3\;\nats\;(1\!+\!n)\;(s\,(0+n))\;(s\,n)$\\
$\quad (A_2\,0\,n)\;(E_4\;\nats\;\nats\;(0\!+\!n)\;n\;[k:\nats](s\,k)\;(A_1\,n)))$&:&$1+n\!=_{\nats}\!(s\,n)$
\end{tabular}
\end{center}
We give two examples of predicates defined on natural numbers as follows (where $2:=(s\,1)$):
\begin{eqnarray*}
\noarg\geq\noarg&:=&[n,m:\nats;k!\nats]n=_{\nats}m+k\\
\mathit{even}&:=&[n:\nats;m!\nats]n=2*m
\end{eqnarray*}
\noindent
The property $[n:\nats][(\mathit{even}\:n)\fun(\mathit{even}\,(2+n))]$ about even numbers can be deduced on the basis of the following typing
\begin{eqnarray*}
n\!:\!\nats,x\!:\!(\mathit{even}\:n)\gv\prdef{m}{1+\pleft{x}}{(\mathit{law}\,n\,\pleft{x}\,\pright{x})}{2+n=_{\nats}2*m}&\!\!:\!\!&(\mathit{even}\,(2+n))
\end{eqnarray*}
Here $\mathit{law}$ is an abbreviation for a proof of the following property
\begin{eqnarray*}
\mathit{law}&:&[n,m:\nats][n\!=_{\nats}\!2\!*\!m\fun\;2\!+\!n\!=_{\nats}\!2\!*\!(1\!+\!m)]
\end{eqnarray*}
A definition of \emph{law} can be derived from the axioms in a style similar (and even more tedious) to the above deduction.
This deduction of the property about even numbers is correct since
\begin{eqnarray*}
(\mathit{law}\,n\,\pleft{x}\,\pright{x})&:&2+n=_{\nats}2*(1+\pleft{x})\\
&\eqv&(2+n=_{\nats}2*m)\gsub{m}{1\!+\!\pleft{x}}
\end{eqnarray*}
\subsection{Casting types to the primitive constant}%
\label{ex.casting}
In many cases, the constraint of universally quantified variables to be of primitive type, i.e.~$x:\prim$, can be relaxed to arbitrary well typed expressions $a:b$ using an operation to cast arbitrary types to $\prim$.
For this reason we introduce an axiom scheme for a $\prim$-casting function $\cast_a$\footnote{For a discussion of axioms schemes see also Section~\ref{examples.logic}.} and we also assume the axioms schemes $\castin{a}$ and $\castout{a}$ essentially stating equivalence between casted and uncasted types:
\begin{eqnarray*}
\cast_a&:&[a\fun\prim]\\
\castin{a}&:&[x:a][x\fun(\cast_a\,x)]\\
\castout{a}&:&[x:a][(\cast_a\,x)\fun x]
\end{eqnarray*}
As an example, we can generalize the property of $\ff$ in Section~\ref{overview} to arbitrary well-typed $a$:
\[
x:a\gv[y:\ff](\castout{a}\;x\;(y\:(\cast_a\:x))):[\ff\fun x]
\]
As an another example, assuming one already has axiomatized the more general theory of integers, \eg~with type \emph{Int}, a constant $0_I$, and the relation $\geq$, one could instantiate the declaration of $N$ of the context \emph{Naturals} using casting as follows:
\begin{eqnarray*}
(\nats\;:=\;(\cast_{[\mathit{Int}\fun\prim]}\,[n!\mathit{Int}](n\geq 0_I)),\;\; 0\;:= \ldots)
\end{eqnarray*}
\subsection{Formalizing partial functions}%
\label{example.partial}
In the previous section we have used the application operator in \dcalc\ to model the application of total functions.
As an example of a partial function consider the predecessor function on natural numbers.
To introduce this function in~\dcalc, several approaches come to mind:
\begin{itemize}
\item
The predecessor function can be axiomatized as a total function over the type $\nats$.
\begin{eqnarray*}
p&:&[\nats\fun\nats]\\
P&:&[n:\nats](p\,(s\,n))=_{\nats}n
\end{eqnarray*}
The (potential) problem of this approach is the interpretation of $(p\,0)$ which may lead to unintuitive or harmful consequences.
Furthermore, if additional axioms are to be avoided, the declaration of $p$ must eventually by instantiated by some (total) function which defines a value for $0$.
\item
The predecessor function can be defined by separately formalizing the condition.
\begin{eqnarray*}
\mathit{nonZero}&:=&[i:\nats;j!\nats]i=_{\nats}(s\,j)\\
p&:=&[n:\nats;q:(\mathit{nonZero}\;n)]\pleft{q}
\end{eqnarray*}
While mathematically clean, this definition requires to always provide an additional argument instantiating $q$ when using the predecessor function.
\item
As a variant of the previous approach, the additional argument can be hidden into an adapted type of the predecessor function.
\begin{eqnarray*}
\pnats&:=&[i,j!\nats]i=_{\nats}(s\,j)\\
p&:=&[n:\pnats]\pleft{\pright{n}}
\end{eqnarray*}
While mathematically clean, this approach requires a more complex handling of the argument $n:\pnats$ when using the predecessor function.
For example, in $\pnats$ one can define the number {\bf 1} and apply the predecessor as follows:
\begin{eqnarray*}
{\bf 1}&:=&\prdef{i}{(s\,0)}{\prdef{j}{0}{(E_1\,N\,(s\,0))}{(s\,0)=_{\nats}(s\,j)}}{[j!\nats]i=_{\nats}(s\,j)} :\pnats
\end{eqnarray*}
\end{itemize}
As a consequence $p({\bf 1})\eqv \pleft{\pright{{\bf 1}}}\eqv 0$. Which of these (or other) approaches is best to use seems to depend on the organization and the goals of the formalization at hand.
\subsection{Defining functions from deductions}%
\label{example.functions}
Note that while the predecessor function can be directly defined, more complex functions and their (algorithmic) properties can be derived from the proofs of properties.
As a sketch of an example consider the following well-known property
\begin{eqnarray*}
\mathit{GCD}&:=&[x,y:\nats;k!\nats](\mathit{gcd}\;k\;x\;y) % chktex 35
\end{eqnarray*}
where $(\mathit{gcd}\;k\;x\;y)$ denotes the property that $k$ is the greatest common divisor of $x$ and $y$. % chktex 35
Given a (not necessarily constructive) deduction $P_{\mathit{GCD}}$ of type \emph{GCD}, one can then define the greatest common divisor $x\downarrow y$ and define deductions $d_1$ and $d_2$ proving well-known algorithmic properties.
\begin{eqnarray*}
\noarg\downarrow\noarg&:=&[x,y:\nats]\pleft{(P_{\mathit{GCD}}\;x\;y)}\\
d_1&:&[x,y:\nats](x+y)\downarrow x=_{\nats}y\downarrow x\\
d_2&:&[x,y:\nats]x\downarrow(x+y)=_{\nats}x\downarrow y
\end{eqnarray*}
\subsection{Sets}
When formalizing mathematical deductions, besides natural numbers and equality one needs formal systems for many more basic structures of mathematics.
For example, sets can be axiomatized by the following context using a formalized set comprehension principle.
\begin{eqnarray*}
\mathit{Sets}:=\quad(\quad
\mathbb{P}&:&[\prim\fun\prim],\\
\noarg\in_{\noarg}\noarg&:&[S:\prim][S;(\mathbb{P}\,S)\fun\prim],\\
{\{\noarg\}}_{\noarg}&:&[S:\prim][[S\fun\prim]\fun(\mathbb{P}\,S)],\\
I&:&[S:\prim;x:S;P:[S\fun\prim]][(P\,x)\fun x\in_{S}{\{[y:S](P\,y)\}}_S],\\
O&:&[S:\prim;x:S;P:[S\fun\prim]][x\in_{S}{\{[y:S](P\,y)\}}_S\,\fun\,(P\,x)]\quad)
\end{eqnarray*}
Note that, for better readability, we use the notation $x\in_{S}y$ and ${\{P\}}_S$ abbreviating the expressions $(\noarg\!\in_{\noarg}\!\noarg\;S\,x\,y)$ and
$({\{\noarg\}}_{\noarg}\,S\,P)$.
One can now define sets and set operators using set comprehension. Note the use of the $\prim$-casting function to ensure the set-defining properties are of type $\prim$.
\begin{eqnarray*}
\emptyset&:=& [S:\prim]{\{[x:S](\cast_{[\prim\sfun\prim]}\,\ff)\}}_S\quad:\quad[S:\prim](\mathbb{P}\,S)\\
\noarg\!\cup_{\noarg}\!\noarg&:=& [S:\prim,A,B:(\mathbb{P}\,S)]{\{[x:S](\cast_{[\prim,\prim]}\,[x\in_{S}A+x\in_{S}B])\}}_S\\
&&:\quad[S:\prim][(\mathbb{P}\,S);(\mathbb{P}\,S)\fun(\mathbb{P}\,S)]\\
\mathit{Even}&:=&{\{[x:\nats](\cast_{[\nats\fun\prim]}\,(\mathit{even}\;x))\}}_{\nats}\quad:\quad(\mathbb{P}\,\nats)
\end{eqnarray*}
Properties of individual elements can be deduced on the basis of the axiom $O$, for example
we can extract the property $P:=[n!\nats](x=_{\nats}2*n)$ of a member $x$ of \textit{Even} using the cast-removal axiom as follows (where $\mathit{Even}_P := [x:\nats](\cast_{[\nats\fun\prim]}\,(\mathit{even}\:x)))$:
\begin{eqnarray*}
x:\nats,\mathit{asm}:x\in_{\nats}\mathit{Even}&\sgv&(\castout{[\nats\fun\prim]}\;P\;(O\;\nats\;x\;\mathit{Even}_P\;\mathit{asm}))\;:\; P
\end{eqnarray*}
Note that in this formalization of sets the axiom of choice can be immediately derived as follows:
\begin{eqnarray*}
\Gamma& = &(X,I:\prim,A:[I\fun(\mathbb{P}\,X)],u:[x:I;y!X]y\in_{X}(A\,x))\\
\Gamma&\sgv&\prdef{F}{[i:I]\pleft{(u\,i)}\,}{[i:I]\pright{(u\,i)}}{([i:I](F\,i)\in_{X}(A\,i))}\\
&:&[F![I\fun X];i:I](F\,i)\in_{X}(A\,i)
\end{eqnarray*}
Note also that an alternative definition of naturals from integers can be given with sets as follows:
\begin{eqnarray*}
(\nats\;:=\;(\cast_{(\mathbb{P}\,\textit{Int})}\;{\{[n:\textit{Int}](\cast_{\prim}\,(n\geq 0_I))\}}_\textit{Int}),\;\;0\;:=\ldots)
\end{eqnarray*}
\subsection{Proof structuring}%
\label{example.groups}
To illustrate some proof structuring issues, we formalize the property of being a group as follows (writing $[a_1,a_2,\ldots]$ for $[a_1,[a_2,\ldots ]]$):
\begin{eqnarray*}
\mathit{Group}:=[S:\prim;\noarg\!*\noarg![S;S\fun S];e!S]&[&[x,y,z:S](x*y)*z =_S x*(y*z)\\
&,&[x:S]e*x =_S x\\
&,&[x:S;x'!S]x'*x =_S e\quad]
\end{eqnarray*}
As an example, assume \emph{Integers} as the context \emph{Naturals} (see Section~\ref{examples.nats}) extended with a subtraction operator $a-b$ with corresponding axioms.
First, we can show that $+$ and $0$ form a group:
Obviously, one can construct a deduction \emph{ded} with
\[
\mathit{Equality},\mathit{Integers}\gv\mathit{ded}:P_g
\]
where $P_g$ describes the group laws:
\begin{eqnarray*}
P_g:=&[\;&[x,y,z:\nats]\,(x+y)+z=_{\nats}x+(y+z)\\
&,&[x:\nats]\,0+x=_{\nats}x\\
&,&[x:\nats,x'!\nats]\,x'+x=_{\nats}0\quad]
\end{eqnarray*}
\emph{ded} can be turned into a proof of $(\mathit{Group}\:\nats)$ as follows:
\begin{eqnarray*}
\mathit{isGroup}&:=&\prdef{*}{+}{\prdef{e}{0}{\mathit{ded}}{P_g\gsub{0}{e}}}{P_g\gsub{0}{e}\gsub{+}{*}}\\
\mathit{Equality},\mathit{Integers}&\gv& \mathit{isGroup}:(\mathit{Group}\:\nats)
\end{eqnarray*}
It is well-known that the left-neutral element is also right-neutral, this means, when assuming $g$ to be a group over $S$ there is a proof $p$ such that
\[
Equality,S:\prim,g:(\mathit{Group}\:S)\gv p:[x:S](\pleft{g}\;x\;\pleft{\pright{g}})=_S x
\]
Here $\pleft{g}$ is the function of $g$ and $\pleft{\pright{g}}$ is the neutral element of $g$.
Note that the use of existential declarations is supporting the proof structuring as it hides $*$ and $e$ in the assumptions inside the $g:(Group\:S)$ assumption.
On the other hand, one has to explicitly access the operators using projections.

$p$ can be abstracted to an inference $p':=[S:\prim;g:(\mathit{Group}\:S)]p$ about a derived property of groups as follows:
\begin{eqnarray*}
%p'&:=&[S:\prim;g:\mathit{Group}(S)]p\\
\mathit{Equality}&\gv&p':[S:\prim;g:(\mathit{Group}\:S)][x:S]\;(\pleft{g}\;x\;\pleft{\pright{g}})=_S x
\end{eqnarray*}
Hence we can instantiate $p'$ to obtain the right-neutrality property of $0$ for integers.
\[
\mathit{Equality},\mathit{Integers}\gv(p'\;\nats\;\mathit{isGroup}):[x:\nats]\,x+0 =_{\nats}x
\]
\section{Properties}%
\label{properties}
In this Section we show confluence of reduction, several properties of typing, including uniqueness of types, and strong normalisation of valid, i.e.~typable, expressions. On the basis of these properties we show consistency of \dcalc\ in the sense that no expression is typing to $[x:\prim]x$ under the empty context.
We show the main proof ideas and indicate important steps in detail, significantly more detail can be found in~\cite{dcalculus}.
\begin{rem}[Inductive principles]
Besides structural induction on the definition of $\dexp$,
% we will mainly use structural induction \emph{on expressions with context} when showing a property $P(\Gamma,a)$ of expressions. We will make use of the following inductive inference rules
%\[
%\frac{P(\Gamma,a_1)\quad P(\Gamma,a_{n})}{P(\Gamma,\binop{a_1}{\ldots,a_n})}
%\qquad
%\frac{P(\Gamma,a_1)\; \ldots\; P(\Gamma,a_{n-1})\quad P((\Gamma,x:a_1),a_n)}{P(\Gamma,\binbop{x}{a_1}{\ldots,a_n})}
%\]
we frequently show properties about reduction relations by induction on the definition of single-step reduction.
Similarly for properties about typing.
\end{rem}
\begin{rem}[Renaming of variables]
Typically when we prove a property using some axiom, inference rule, or derived property, we just mention the identifier or this axiom, rule, or property and then use it with an instantiation \emph{renaming its variables so as to avoid name clashes with the proposition to be shown}.
In order not to clutter the presentation, these renamings are usually not explicitly indicated.
\end{rem}
\begin{rem}[Introduction of auxiliary identifiers]
We usually introduce explicitly all auxiliary identifiers appearing in deduction steps.
However, there are two important exceptions.
\begin{itemize}
\item
When using structural induction, if we consider a specific operator and decompose an expression $a$ \eg~by $a=\binop{a_1}{,\ldots,a_n}$ we usually introduce implicitly the new auxiliary identifiers $a_i$.
\item
When using induction on the definition of reduction, if we consider a specific axiom or structural rule which requires a syntactic pattern we usually introduce implicitly the new auxiliary identifiers necessary for this pattern.
\end{itemize}
\end{rem}
\noindent
We begin with some basic properties of substitution and its relation to reduction.
\begin{lem}[Basic properties of substitution]%
\label{sub.basic}
For all $a,b,c$ and $x,y$:
\begin{enumerate}
\item
If $x\notin\free(a)$ then $a\gsub{x}{b}=a$.
\item
If $x\neq y$ and $x\notin\free(c)$ then $a\gsub{x}{b}\gsub{y}{c}= a\gsub{y}{c}\gsub{x}{b\gsub{y}{c}}$.
\item
If $x\neq y$, $x\notin\free(c), y\notin\free(b)$ then $a\gsub{x}{b}\gsub{y}{c}= a\gsub{y}{c}\gsub{x}{b}$.
\end{enumerate}
\end{lem}
\begin{proof}
Proof is straightforward by structural induction on $a$.
\end{proof}
\begin{lem}[Substitution and reduction]%
\label{sub.rd}
For all $a,b,c$ and $x$: $a\rd b$ implies $a\gsub{x}{c}\rd b\gsub{x}{c}$, $c\gsub{x}{a}\rd c\gsub{x}{b}$, and $\free(b)\subseteq\free(a)$.
\end{lem}
\begin{proof}
Proof is straightforward by induction on the definition of single-step reduction.
\end{proof}
\noindent
Next we turn to basic decomposition properties of reduction.
\begin{lem}[Reduction decomposition]%
\label{rd.decomp}
For all $a_1,\ldots,a_n,b,b_1,\ldots,b_n$ and $x$:
\begin{enumerate}
\item
$\binop{a_1}{a_2}\!\rd\!b$, where $\binop{a_1}{a_2}\neq(a_1\,a_2)$ implies $b=\binop{b_1}{b_2}$ where $a_1\!\rd\!b_1$ and $a_2\!\rd\!b_2$.
\item
$\binbop{x}{a_1}{\ldots,a_n}\rd b$ implies $b=\binbop{x}{b_1}{\ldots,b_n}$, $a_i\rd b_i$, for $1\leq i\leq n$.
\item
$a.i\rd\binbop{x}{b_1}{\ldots,b_n}$, $i=1,2$, implies $a\rd\prsumop{c_1}{c_2}$ or $a\rd\prdef{y}{c_1}{c_2}{c_3}$, for some $c_1, c_2,c_3$ with $c_i\rd\binbop{x}{b_1}{\ldots,b_n}$.
\item
$a.i\rd\binop{a_1}{a_2}$, $i=1,2$, where $\binop{a_1}{a_2}$ is a sum, product or injection implies $a\rd\prsumopd{c_1}{c_2}$ or $a\rd\prdef{y}{c_1}{c_2}{c_3}$, for some $c_1, c_2,c_3$ with $c_i\rd\binop{a_1}{a_2}$.
\item
$a_1(a_2)\rd\binbop{x}{b_1}{\ldots,b_n}$ implies, for some $c_1, c_2, c_3, c_4$, one of the following cases
\begin{enumerate}
\item $a_1\rd\binbopd{y}{c_1}{c_2}$, $a_2\rd c_3$ and $c_2\gsub{y}{c_3}\rd\binbop{x}{b_1}{\ldots, b_n}$, or
\item $a_1\rd\case{c_1}{c_2}$ and either $a_2\rd\injl{c_3}{c_4}$ and $(c_1\,c_3)\rd\binbop{x}{b_1}{\ldots,b_n}$
 or $a_2\rd \injr{c_3}{c_4}$ and $(c_2\,c_4)\rd\binbop{x}{b_1}{\ldots,b_n}$.
\end{enumerate}
\item
$a_1(a_2)\rd\binop{b_1}{\ldots,b_n}$ where $\binop{b_1}{\ldots,b_n}\neq(b_1\,b_2)$ implies, for some $c_1$, $c_2$, $c_3$, and $c_4$, one of the following cases
\begin{enumerate}
\item $a_1\rd\binbopd{y}{c_1}{c_2}$, $a_2\rd c_3$, and $c_2\gsub{x}{c_3}\rd\binop{b_1}{\ldots,b_n}$, or
\item $a_1\rd\case{c_1}{c_2}$ and either $a_2\rd\injl{c_3}{c_4}$ and $(c_1\,c_3)\rd\binop{b_1}{\ldots,b_n}$
 or $a_2\rd\injr{c_4}{c_3}$ and $(c_2\,c_3)\rd\binop{b_1}{\ldots,b_n}$.
\end{enumerate}
\end{enumerate}
\end{lem}
\begin{proof}
Straightforward inductions, as these properties are very close to the definition of the reduction relation.
\end{proof}
\begin{lem}[Reduction decomposition for negation]%
\label{rd.decomp.neg}
For all $a_1,\ldots,a_n,b,b_1,\ldots,b_n$ and $x$:
\begin{enumerate}
%\item[$i$:]
%$\cast a_1\rd b$ implies $b=\cast b_1$ where $a_1\rd b_1$.
\item
$\myneg a\rd[x:b_1]b_2$ implies $a\rd[x!c_1]c_2$ for some $c_1, c_2$ where $c_1\rd b_1$ and $\myneg c_2\rd b_2$.
\item
$\myneg a\rd[b_1,b_2]$ implies $a\rd[c_1+c_2]$ for some $c_1, c_2$ where $\myneg c_1\rd b_1$ and $\myneg c_2\rd b_2$.
\item
$\myneg a\rd[x!b_1]b_2$ implies $a\rd[x:c_1]c_2$ for some $c_1, c_2$ where $c_1\rd b_1$ and $\myneg c_2\rd b_2$.
\item
$\myneg a\rd[b_1+b_2]$ implies $a\rd[c_1,c_2]$ for some $c_1, c_2$ where $\myneg c_1\rd b_1$ and $\myneg c_2\rd b_2$.
\item
$\myneg a\rd\prdef{x}{b_1}{b_2}{b_3}$ implies $a\rd\prdef{x}{b_1}{b_2}{b_3}$.
\item
$\myneg a\rd\case{b_1}{b_2}$ implies $a\rd\case{b_1}{b_2}$.
\item
$\myneg a\rd\injl{b_1}{b_2}$ implies $a\rd\injl{b_1}{b_2}$.
\item
$\myneg a\rd\injr{b_1}{b_2}$ implies $a\rd\injr{b_1}{b_2}$.
\end{enumerate}
\end{lem}
\begin{proof}
Straightforward inductions, as these properties are very close to the definition of the reduction relation.
\end{proof}
\subsection{Confluence of \texorpdfstring{$\srd$}{->}}%
\label{confl}
Classical confluence proofs for untyped $\lambda$-calculus could be used,
\eg~\cite{Bar:93} or~\cite{Takahashi95} (using parallel reduction) could be adapted to include the operators of \dcalc.
Due to the significant number of reduction axioms of~\dcalc, we use an alternative approach using explicit substitutions and an auxiliary relation of \emph{reduction with explicit substitution} which has detailed substitution steps on the basis of a definitional environment (this approach was basically already adopted in the Automath project~\cite{deBruijn80}) and which comprises sequences of negation-related reduction-steps into single steps.

The underlying idea is that reduction with explicit substitution  can be shown to be directly confluent which implies its confluence.
We then show that this implies confluence of $\srd$.
There are several approaches to reduction with explicit substitutions, \eg~\cite{ACCL91,AK2010}. %\cite{Kesner07}
Furthermore, there is a significant amount of more recent work in this context, however, as explicit substitution is not the main focus of this article we do not give an overview here.
However, we should note that, in general, confluence of calculi with explicit substitutions is proved by using confluence of the
underlying calculus without explicit substitutions.
Here it is the other way around: confluence of reduction with explicit substitution is used to show confluence of $\srd$.

The approach introduced below introduces a definitional environment as part of the reduction relation
to explicitly unfold single substitution instances and then discard substitution expressions when all instances are unfolded.
This approach, as far as basic lambda calculus operators are concerned, is essentially equivalent to the system $\Lambda_{sub}$
which has been defined using substitution~\cite{Milner2007} or placeholders~\cite{KC2008} to indicate particular occurrences to be substituted.
Both approaches slightly differ from ours as they duplicate the substituted expression on the right-hand side of the $\beta$-rule thus violating direct confluence.

We begin by defining expressions with substitution.
\begin{defi}[Expressions with substitution]
The set $\dexps$ of \emph{expressions with substitution} is an extension of the set $\dexp$ of expressions adding
a substitution operator.
\begin{eqnarray*}
\dexps&\!::=\!&\underbrace{\{\prim\}\,\mid\,\ldots\,\mid\,\myneg\dexps}_{\text{(see Definition~\ref{expression})}}\,\mid\,[\dvar\mydef\dexps]\dexps
\end{eqnarray*}
Expressions with substitution will be denoted by $\mathbf{a},\mathbf{b},\mathbf{c},\mathbf{d},\ldots$.
$[x\mydef\mathbf{a}]\mathbf{b}$ is an \emph{internalized substitution}.
As indicated by its name, the purpose of $[x\mydef\mathbf{a}]\mathbf{b}$ is to internalize the substitution function.
\end{defi}
\noindent
The function computing free variables (Definition~\ref{free}) is extended so as to treat internalized substitutions identical to abstractions.
\begin{eqnarray*}
\free([x\mydef\mathbf{a}]\mathbf{b})&=&\free(\mathbf{a})\union(\free(\mathbf{b})\!\setminus\!\{x\})
\end{eqnarray*}
The substitution function (Definition~\ref{free}) is extended analogously.
\begin{eqnarray*}
([z\mydef \mathbf{a}]\mathbf{b})\gsub{x}{y}&=&	\begin{cases}
													[z\mydef\mathbf{a}\gsub{x}{y}]\mathbf{b}&\text{if }z=x,\\
													[z\mydef\mathbf{a}\gsub{x}{y}](\mathbf{b}\gsub{x}{y})&\text{otherwise}
													\end{cases}
\end{eqnarray*}
Similar for $\alpha$-conversion (Definition~\ref{alpha}).
\[
\frac{y\notin\free(\mathbf{b})}{[x\mydef\mathbf{a}]\mathbf{b}\;=_{\alpha}\;[y\mydef\mathbf{a}]\mathbf{b}\gsub{x}{y}}
\]
As for expressions we will write variables as strings but always assume appropriate renaming of bound variables in order to avoid name clashes.
In order to define reduction with explicit substitution as a directly confluent relation, we first define the auxiliary notion of \emph{negation-reduction} $\mathbf{a}\nur\mathbf{b}$ which comprises application sequences of axioms $\nu_1$, $\ldots$, $\nu_5$ in a restricted context.
\begin{defi}[Negation-reduction]
Single-step \emph{negation-reduction} $\mathbf{a}\snur\mathbf{b}$ is the smallest relation on expressions with explicit substitution satisfying the axiom and the inference rules of Table~\ref{nurd.rules}.
\begin{table}[!htb]
\fbox{
\begin{minipage}{0.96\textwidth}
\begin{center}
\begin{tabular}{@{$\;$}l@{}r@{$\;$}c@{$\;$}ll@{}r@{$\;$}c@{$\;$}l}
$\;$\\[-3mm]
$\mathit{(\nu_1)}$&$\;\myneg\myneg\mathbf{a}$&$\snur$&$\mathbf{a}$\\
$\mathit{(\nu_2)}$&$\;\myneg[\mathbf{a},\mathbf{b}]$&$\snur$&$[\myneg\mathbf{a}+\myneg\mathbf{b}]$&$\mathit{(\nu_3)}$&$\;\myneg[\mathbf{a}+\mathbf{b}]$&$\snur$&$[\myneg\mathbf{a},\myneg\mathbf{b}]$\\
$\mathit{(\nu_4)}$&$\;\myneg[x:\mathbf{a}]\mathbf{b}$&$\snur$&$[x!\mathbf{a}]\myneg\mathbf{b}$&$\mathit{(\nu_5)}$&$\;\myneg[x!\mathbf{a}]\mathbf{b}$&$\snur$&$[x:\mathbf{a}]\myneg\mathbf{b}$
\end{tabular}
\end{center}
\begin{align*}
\\[-6mm]
\qquad(\prsumop{\_}{\_}_1)\quad\frac{\mathbf{a}_1\snur\mathbf{a}_2}{\prsumop{\mathbf{a}_1}{\mathbf{b}}\snur\prsumop{\mathbf{a}_2}{\mathbf{b}}}
\qquad(\prsumop{\_}{\_}_2)\quad\frac{\mathbf{b}_1\snur\mathbf{b}_2}{\prsumop{\mathbf{a}}{\mathbf{b}_1}\snur\prsumop{\mathbf{a}}{\mathbf{b}_2}}
\end{align*}
\begin{align*}
\\[-8mm]
(\binbop{\_}{\_}{\_}_2)\quad\frac{\mathbf{b}_1\snur\mathbf{b}_2}{\binbop{x}{\mathbf{a}}{\mathbf{b}_1}\snur\binbop{x}{\mathbf{a}}{\mathbf{b}_2}}
\qquad
(\myneg{\_}_1)\quad\frac{\mathbf{a}_1\snur\mathbf{a}_2}{\myneg\mathbf{a}_1\snur\myneg\mathbf{a}_2}
\qquad
\end{align*}
\end{minipage}
}\caption{Axioms and rules for negation-reduction.\label{nurd.rules}}
\end{table}
$n$-step negation-reduction $\mathbf{a}\nurn{n}\mathbf{b}$ ($n\geq 0$) and negation-reduction $\mathbf{a}\nur\mathbf{b}$ are defined as follows:
\begin{eqnarray*}
\mathbf{a}\nurn{n}\mathbf{b}&:=&\exists\mathbf{b}_1,\ldots\mathbf{b}_{n-1}:\mathbf{a}\snur\mathbf{b}_1,\ldots,\mathbf{b}_{n-1}\snur\mathbf{b}.\\
\mathbf{a}\nur\mathbf{b}&:=&\exists k\geq 0:\mathbf{a}\nurn{k}\mathbf{b}
\end{eqnarray*}
\end{defi}
\noindent
Not surprisingly, $\snur$ is confluent.
\begin{lem}[Confluence of $\snur$]%
\label{nurd.confl}
$\snur$ is confluent, i.e.~For all $\mathbf{a},\mathbf{b},\mathbf{c}$: $\mathbf{a}\nur\mathbf{b}$ and $\mathbf{a}\nur\mathbf{c}$ imply $\mathbf{b}\nur\mathbf{d}$, and $\mathbf{c}\nur\mathbf{d}$ for some $\mathbf{d}$.
\end{lem}
\begin{proof}
The proof of confluence of $\snur$ is straightforward by induction on the size $\sz(\mathbf{a})$ which counts all nodes and leaves of $\mathbf{a}$'s expression tree~\footnote{This ensures we can apply the inductive hypothesis \eg~to $\myneg\mathbf{a}_1$ when considering an $\mathbf{a}=[\mathbf{a}_1,\mathbf{a}_2]$, since $\sz(\myneg \mathbf{a}_1)=\sz(\mathbf{a}_1)+1<\sz(\mathbf{a}_1)+\sz(\mathbf{a}_2)+1=\sz(\mathbf{a})$.}.
For the inductive base,  the case $\sz(\mathbf{a})=1$ is obviously true.
Consider an expression $a$ with $\sz(\mathbf{a})=n>0$ and assume confluence for all expressions $\mathbf{b}$ with $\sz(\mathbf{b})<\sz(\mathbf{a})$.
By systematic investigation of critical pairs one can show local confluence of $\snur$ from $\mathbf{a}$.
Furthermore, $\nur$ is terminating from $\mathbf{a}$ which can be seen by defining a weight function that squares the weight (incremented by one) in case of negations.
Hence one can apply the diamond lemma~\cite{NEWMAN1942} to obtain confluence of $\snur$ for $\mathbf{a}$, which completes the inductive step.
\end{proof}
\noindent
In order to define reduction with explicit substitution, we need to introduce the notion of environments, which are used to record the definitions which are currently available for substitution.
\begin{defi}[Environment]
\emph{Environments}, denoted by $E$, $E_1$, $E_2$, etc.~are finite sequences of definitions $(x_1\mydef \mathbf{a}_1,\ldots ,x_n\mydef \mathbf{a}_n)$, where $x_i$ are variables and $x_i\neq x_j$. $\{x_1,\ldots,x_n\}$ is called the \emph{domain} of $E$.
The lookup of an variable $x$ in the domain of an environment $E$ is defined by $E(x)= \mathbf{a}_i$.
$E,x\mydef\mathbf{a}$ denotes the extension of $E$ on the right by a definition $x\mydef\mathbf{a}$.
$E_1,E_2$ denotes the concatenation of two environments.
The empty environment is written as $()$.
%Similarly to induction with context we will use induction \emph{with environment}.
\end{defi}
\begin{defi}[Single-step reduction with explicit substitution]
\emph{Single-step reduction reduction with explicit substitution} $E\sgv\mathbf{a}\smrd\mathbf{b}$ is the smallest relation on expressions with explicit substitution satisfying the axiom and the inference rules of Table~\ref{mred.rules}.
\begin{table}[!htb]
\fbox{
\begin{minipage}{0.96\textwidth}
\begin{center}
\begin{tabular}{@{$\;$}l@{}r@{$\;$}c@{$\;$}ll@{}r@{$\;$}c@{$\;$}l}
$\;$\\[-3mm]
$\mathit{(\beta_1^{\mu})}$&$E\sgv([x:\mathbf{a}]\mathbf{b}\,\mathbf{c})$&$\smrd$&$[x\mydef\mathbf{c}]\mathbf{b}$&$\mathit{(\beta_2^{\mu})}$&$E\sgv([x!\mathbf{a}]\mathbf{b}\,\mathbf{c})$&$\smrd$&$[x\mydef\mathbf{c}]\mathbf{b}$\\
$\mathit{(\beta_3)}$&$E\sgv(\case{\mathbf{a}}{\mathbf{b}}\,\injl{\mathbf{c}}{\mathbf{d}})$&$\smrd$&$(\mathbf{a}\,\mathbf{c})$&$\mathit{(\beta_4)}$&$E\sgv(\case{\mathbf{a}}{\mathbf{b}}\,\injr{\mathbf{c}}{\mathbf{d}})$&$\smrd$&$(\mathbf{b}\,\mathbf{d})$\\[2mm]
$\mathit{(use)}$&$E\sgv x$&$\smrd$&$\mathbf{a}$&\multicolumn{4}{l}{if $E(x)=\mathbf{a}$}\\
$\mathit{(rem)}$&$E\sgv[x\mydef\mathbf{a}]\mathbf{b}$&$\smrd$&$\mathbf{b}$&\multicolumn{4}{l}{if $x\notin\free(\mathbf{b})$}\\[2mm]
$\mathit{(\pi_1)}$&$E\sgv\pleft{\prdef{x}{\mathbf{a}}{\mathbf{b}}{\mathbf{c}}}$&$\smrd$&$\mathbf{a}$&$\mathit{(\pi_2)}$&$E\sgv\pright{\prdef{x}{\mathbf{a}}{\mathbf{b}}{\mathbf{c}}}$&$\smrd$&$\mathbf{b}$\\
$\mathit{(\pi_3)}$&$E\sgv\pleft{[\mathbf{a},\mathbf{b}]}$&$\smrd$&$\mathbf{a}$&$\mathit{(\pi_4)}$&$E\sgv\pright{[\mathbf{a},\mathbf{b}]}$&$\smrd$&$\mathbf{b}$\\
$\mathit{(\pi_5)}$&$E\sgv\pleft{[\mathbf{a}+\mathbf{b}]}$&$\smrd$&$\mathbf{a}$&$\mathit{(\pi_6)}$&$E\sgv\pright{[\mathbf{a}+\mathbf{b}]}$&$\smrd$&$\mathbf{b}$\\[2mm]
$\mathit{(\nu_6)}$&$E\sgv\myneg\prim$&$\smrd$&$\prim$&$\mathit{(\nu_7)}$&$E\sgv\myneg\prdef{x}{\mathbf{a}}{\mathbf{b}}{\mathbf{c}}$&$\smrd$&$\prdef{x}{\mathbf{a}}{\mathbf{b}}{\mathbf{c}}$\\
$\mathit{(\nu_8)}$&$E\sgv\myneg\injl{\mathbf{a}}{\mathbf{b}}$&$\smrd$&$\injl{\mathbf{a}}{\mathbf{b}}$&$\mathit{(\nu_9)}$&$E\sgv\myneg\injr{\mathbf{a}}{\mathbf{b}}$&$\smrd$&$\injr{\mathbf{a}}{\mathbf{b}}$\\
$\mathit{(\nu_{10})}$&$E\sgv\myneg\case{\mathbf{a}}{\mathbf{b}}$&$\smrd$&$\case{\mathbf{a}}{\mathbf{b}}$
\end{tabular}
\end{center}
\begin{align*}
\mathit{(\nu)}\quad&\frac{\exists k>0:\mathbf{a}\nurn{k}\mathbf{b}}{E\gv\mathbf{a}\smrd\mathbf{b}}
\end{align*}
\begin{align*}
\\[-8mm]
\mathit{(\oplus{\overbrace{(\_,\ldots,\_)}^{n}}_i)}\quad&\frac{E\gv\mathbf{a}_i\smrd\mathbf{b}_i}{E\gv\binop{\mathbf{a}_1,\ldots,\mathbf{a}_i}{\ldots,\mathbf{a}_n}\smrd \binop{\mathbf{a}_1,\ldots,\mathbf{b}_i}{\ldots,\mathbf{a}_n}}\\
\mathit{(\oplus_x{\overbrace{(\_,\ldots,\_)}^{n}}_i)}\quad&\frac{E\gv\mathbf{a}_i\smrd\mathbf{b}_i}{E\gv\binbop{x}{\mathbf{a}_1,\ldots,\mathbf{a}_i}{\ldots,\mathbf{a}_n}\smrd \binbop{x}{\mathbf{a}_1,\ldots,\mathbf{b}_i}{\ldots,\mathbf{a}_n}}
\end{align*}
\begin{align*}
\\[-8mm]
\mathit{(L_{\smydef})}\quad&\frac{E\gv\mathbf{a}\smrd\mathbf{b}}{E\gv[x\mydef\mathbf{a}]\mathbf{c}\smrd[x\mydef\mathbf{b}]\mathbf{c}}&\quad
\mathit{(R_{\smydef})}\quad&\frac{E,x\mydef\mathbf{a}\gv\mathbf{b}\smrd\mathbf{c}}{E\gv[x\mydef\mathbf{a}]\mathbf{b}\smrd[x\mydef\mathbf{a}]\mathbf{c}}\\[-4mm]
\end{align*}
\end{minipage}
}\caption{Axioms and rules for single-step reduction with explicit substitution.\label{mred.rules}}
\end{table}
Compared to (general) reduction $\to$, the axiom $\beta_1$ and $\beta_2$ have been decomposed into several axioms:
\begin{itemize}
\item $\beta^{\mu}_1$ and $\beta^{\mu}_2$ are reformulation of $\beta_1$ and $\beta_2$ using internalized substitution
\item \emph{use} is unfolding single usages of definitions
\item \emph{rem} is removing a definition without usage
\end{itemize}
The axioms $\nu_1,\ldots,\nu_5$, which are not directly confluent \eg~for $\myneg\myneg[a,b]$, have been removed and replaced by the rule $\nu$.
Furthermore there are two more structural rules $(L_{\smydef})$ and $(R_{\smydef})$ related to substitution.
Note that the rule $(R_{\smydef})$ is pushing a definition onto the environment $E$ when evaluating the body of a definition.
Furthermore renaming may be necessary before using this rule to ensure that $E,x\mydef\mathbf{a}$ is well defined.
\end{defi}
\begin{defi}[Reduction with explicit substitution]
Reduction with explicit substitution $E\sgv\mathbf{a}\mrd\mathbf{b}$ of $\mathbf{a}$ to $\mathbf{b}$ is defined as the reflexive and transitive closure of $E\sgv\mathbf{a}\smrd\mathbf{b}$.
If two expressions $\mathbf{a}$ and $\mathbf{b}$ $\mrd$-reduce to a common expression we write $E\sgv\mathbf{a}\mrdr\mathbf{b}$.
Zero-or-one-step reduction with explicit substitution and $n$-step reduction with explicit substitution are defined as follows
\begin{eqnarray*}
E\sgv\mathbf{a}\mrdn{01}\mathbf{b}&:=&E\sgv\mathbf{a}\smrd\mathbf{b}\;\vee\;\mathbf{a}=\mathbf{b}\\
E\sgv\mathbf{a}\mrdn{n}\mathbf{b}&:=&\exists\mathbf{b}_1,\ldots\mathbf{b}_{n-1}:E\sgv\mathbf{a}\smrd\mathbf{b}_1,\ldots,E\sgv\mathbf{b}_{n-1}\smrd\mathbf{b}.
\end{eqnarray*}
\end{defi}
\begin{rem}[Avoidance of name clashes through appropriate renaming]
Note that renaming is necessary to prepare use of the axiom \emph{use}:
For example when $\mrd$-reducing $y\mydef x\sgv [x:\prim][y,x]$, the expression $[x:\prim][y,x]$ needs to be renamed to \eg~$[z:\prim][y,z]$ before substituting $y$ by $x$.
\end{rem}
\noindent
Reduction ($\rd$) can obviously be embedded into $\mrd$.
\begin{lem}[Reduction implies reduction with explicit substitution]%
\label{rd.mrd}
For all $a,b$: $a\rd b$ implies $()\sgv a\mrd b$.
\end{lem}
\begin{proof}
Proof is by induction on the definition of $a\rd b$.
\end{proof}
\begin{rem}[Sketch of the confluence proof of $\smrd$]
First we show that $\smrd$ commutes with $\nur$.
On the basis of these results, by induction on expressions with substitution one can establish direct confluence of $\smrd$, i.e. $E\sgv\mathbf{a}\smrd\mathbf{b}$ and $E\sgv \mathbf{a}\smrd\mathbf{c}$ imply $E\sgv\mathbf{b}\mrdn{01}\mathbf{d}$ and $E\sgv\mathbf{c}\mrdn{01}\mathbf{d}$ for some $\mathbf{d}$.
Confluence follows by two subsequent inductions.
\end{rem}
\begin{lem}[Commutation of single-step reduction with explicit substitution and negation-reduction]%
\label{nurd.comm}
For all $\mathbf{a}$, $\mathbf{b}$ and $\mathbf{c}$: $E\sgv\mathbf{a}\smrd\mathbf{b}$ and $\mathbf{a}\nur\mathbf{c}$ imply  $\mathbf{b}\nur\mathbf{d}$ and $E\sgv\mathbf{c}\mrdn{01}\mathbf{d}$ for some $\mathbf{d}$.
This can be graphically displayed as follows (leaving out the environment $E$):
\[\begin{tikzcd}[row sep= scriptsize, column sep= large]
    \mathbf{a}\ar[r,"\stackrel{*}{\myneg}" description]\arrow[d,"{:=}"]&\mathbf{c}\arrow[d,"\stackrel{01}{:=}"]\\
    \mathbf{b}\ar[r,"\stackrel{*}{\myneg}" description] &\mathbf{d}
\end{tikzcd}
\]
\end{lem}
\begin{proof}
First we prove by induction on $\mathbf{a}$ %with environment $E$
that for all $\mathbf{b}$ and $\mathbf{c}$, $E\sgv\mathbf{a}\smrd\mathbf{b}$ and $\mathbf{a}\snur\mathbf{c}$ imply $\mathbf{b}\nur\mathbf{d}$ and $E\sgv\mathbf{c}\mrdn{01}\mathbf{d}$ for some $\mathbf{d}$.
Obviously we only have to consider those cases in which there exists a $\mathbf{c}$ such that $\mathbf{a}\snur\mathbf{c}$ (see Table~\ref{nurd.rules}).
In case of the axioms $\nu_1$ to $\nu_5$ the proposition follows directly.
In case of the rules for $\mathbf{a}=\prsumop{\mathbf{a}_1}{\mathbf{a}_2}$ or $\mathbf{a}=\binbop{x}{\mathbf{a}_1}{\mathbf{a}_2}$ the proposition follows either immediately or from the inductive hypothesis.
In case of the rule for $\mathbf{a}=\myneg\mathbf{a}_1$ the proposition follows from confluence of $\snur$ (Lemma~\ref{nurd.confl}).

We can now prove the main property by induction on the length $n$ of $\mathbf{a}\nurn{n}\mathbf{c}$:
In case of $n=0$ the property is trivial.
Otherwise, let  $E\sgv\mathbf{a}\smrd\mathbf{b}$ and $\mathbf{a}\nurn{n}\mathbf{c}'\snur\mathbf{c}$ for some $\mathbf{c}'$.
By inductive hypothesis we know there is an expression $\mathbf{d}'$ such that
$E\sgv\mathbf{c}'\mrdn{01}\mathbf{d}'$ and $\mathbf{b}\nur\mathbf{d}'$.
This situation can be graphically summarized as follows (leaving out the environment $E$):
\[\begin{tikzcd}[row sep= scriptsize, column sep= large]
    \mathbf{a}\ar[r,"\stackrel{n}{\myneg}" description]\ar[d,"{:=}"]&\mathbf{c}'\ar[r,"\myneg" description]\ar[d,"\stackrel{01}{:=}"]&\mathbf{c}\\
    \mathbf{b}\ar[r,"\stackrel{*}{\myneg}" description]             &\mathbf{d}'
\end{tikzcd}
\]
Be definition of $E\sgv\mathbf{c}'\mrdn{01}\mathbf{d}'$  there are two cases:
\begin{enumerate}
\item
$\mathbf{c}'=\mathbf{d}'$: Then we know that $\mathbf{b}\nur\mathbf{c}'$ and hence $\mathbf{b}\nur\mathbf{c}$.
Hence $\mathbf{d}=\mathbf{c}$ where $\mathbf{b}\nur\mathbf{d}$ and $E\sgv\mathbf{c}\mrdn{01}\mathbf{d}$.
This situation can be graphically summarized as follows (leaving out the environment $E$):
\[\begin{tikzcd}[row sep= scriptsize, column sep= huge]
    \mathbf{a}\ar[r, "\stackrel{n}{\myneg}" description]\arrow[d,"{:=}"]&\mathbf{d}'=\mathbf{c}'\ar[r,"\myneg" description]& \mathbf{d}=\mathbf{c}\\
    \mathbf{b}\ar[ru,"\stackrel{*}{\myneg}" description]
\end{tikzcd}
\]
\item
$E\sgv\mathbf{c}'\smrd\mathbf{d}'$.
By the argument at the beginning of this proof we know there is an expression $\mathbf{d}$ such that $\mathbf{d}'\nur\mathbf{d}$ and $E\sgv\mathbf{c}\mrdn{01}\mathbf{d}$.
Hence $\mathbf{b}\nur\mathbf{d}'\nur\mathbf{d}$ and $E\sgv\mathbf{c}\mrdn{01}\mathbf{d}$ which completes the proof.
This situation can be graphically summarized as follows (leaving out the environment $E$):
\begin{align*}
\begin{tikzcd}[row sep= scriptsize, column sep= large, ampersand replacement=\&]
    \mathbf{a}\ar[r,"\stackrel{n}{\myneg}" description]\arrow[d,"{:=}"]\& \mathbf{c}'\arrow[r,"{\myneg}" description]\ar{d}{\stackrel{01}{:=}}\&\mathbf{c}\arrow[d,"\stackrel{01}{:=}"]\\
    \mathbf{b}\ar[r,"\stackrel{*}{\myneg}" description]                \&\mathbf{d}'\ar[r,"\stackrel{*}{\myneg}" description]                 \&\mathbf{d}
\end{tikzcd}
\\[\dimexpr-1.2\baselineskip+\dp\strutbox]&\qedhere
\end{align*}
\end{enumerate}
\end{proof}
\begin{lem}[Direct confluence of $\smrd$]%
\label{mrd.confl.dir}
For all $E,\mathbf{a},\mathbf{b},\mathbf{c}$: $E\sgv\mathbf{a}\smrd\mathbf{b}$ and $E\sgv\mathbf{a}\smrd\mathbf{c}$ imply $E\sgv\mathbf{b}\mrdn{01}\mathbf{d}$, and $E\sgv\mathbf{c}\mrdn{01}\mathbf{d}$ for some $\mathbf{d}$.
\end{lem}
\begin{proof}
Proof is by induction on $\mathbf{a}$ %with environment $E$
where in each inductive step we investigate critical pairs.

Due to the definition of $\smrd$, critical pairs of $E\sgv\mathbf{a}\smrd\mathbf{b}$ and $E\sgv\mathbf{a}\smrd\mathbf{c}$ in which at least one of the steps is using axiom $\nu$ on top-level, i.e.~where $\exists k>0:\mathbf{a}\nurn{k}\mathbf{b}$ or $\exists k>0:\mathbf{a}\nurn{k}\mathbf{c}$, can be resolved thanks to Lemmas~\ref{nurd.confl} and~\ref{nurd.comm}. Note also that the reduction axioms and rules of $\mrd$ do not duplicate on their right hand side any element appearing on the left hand side, hence the axiom \emph{use} may never violate direct confluence in a critical pair.

Next we show the interesting cases of explicit definitions and applications, the other cases are straightforward.
\begin{itemize}
\item
$\mathbf{a}=[x\mydef\mathbf{a}_1]\mathbf{a}_2$: The matching axiom and rules are \emph{rem}, \emph{L$_{\smydef}$}, and \emph{R$_{\smydef}$}. The use of \emph{L$_{\smydef}$} versus \emph{R$_{\smydef}$} follow directly from the inductive hypothesis. The interesting cases are the use of \emph{rem}, versus \emph{L$_{\smydef}$} or \emph{R$_{\smydef}$}: Hence we may assume that $x\notin\free(\mathbf{a}_2)$ and need to consider the following cases:
\begin{enumerate}
\item
$\mathbf{b}=\mathbf{a}_2$ and $\mathbf{c}=[x\mydef\mathbf{b}_1]\mathbf{a}_2$ where $E\sgv\mathbf{a}_1\smrd\mathbf{b}_1$.
We have $\mathbf{d}=\mathbf{a}_2$ since $x\notin \free(\mathbf{a}_2)$ and $\mathbf{c}$ reduces in one-step to $\mathbf{a}_2$.
This situation can be graphically summarized as follows (leaving out the environment $E$):
\[\begin{tikzcd}[row sep= scriptsize, column sep= large]
    \mathbf{a}=[x\mydef\mathbf{a}_1]\mathbf{a}_2\ar[r,"{:=}" description]\arrow[d,"{:=}"]& \mathbf{d}=\mathbf{b}=\mathbf{a}_2\\
    \mathbf{c}=[x\mydef\mathbf{b}_1]\mathbf{a}_2\arrow[ru,"{:=}" description]
\end{tikzcd}
\]
\item
$\mathbf{b}=\mathbf{a}_2$ and $\mathbf{c}=[x\mydef \mathbf{a}_1]\mathbf{c}_2$ where $E,x\mydef\mathbf{a}_1 \sgv\mathbf{a}_2\smrd\mathbf{c}_2$:
Obviously $x\notin\free(\mathbf{c}_2)$ and therefore $E\sgv\mathbf{a}_2\smrd\mathbf{c}_2$.
Therefore $\mathbf{d}=\mathbf{c}_2$ where $E\sgv\mathbf{c}\smrd\mathbf{c}_2$ and $E\sgv\mathbf{b}\smrd\mathbf{c}_2$.
This situation can be graphically summarized as follows (leaving out the environment $E$):
\[\begin{tikzcd}[row sep= scriptsize, column sep= large]
    \mathbf{a}=[x\mydef\mathbf{a}_1]\mathbf{a}_2\arrow[r,"{:=}" description]\arrow[d,"{:=}"]&\mathbf{b}=\mathbf{a}_2\arrow[d,"{:=}"]\\
    \mathbf{c}=[x\mydef\mathbf{a}_1]\mathbf{c}_2\arrow[r,"{:=}" description]                & \mathbf{d}=\mathbf{c}_2
\end{tikzcd}
\]
\item
The other two cases are symmetric (exchange of $\mathbf{b}$ and $\mathbf{c}$).
\end{enumerate}
\item
$\mathbf{a}=(\mathbf{a}_1\,\mathbf{a}_2)$: The four matching axiom and two matching rules are $\beta_1^{\mu}$, $\beta_2^{\mu}$, $\beta_3$, $\beta_4$, ${(\_\,\_)}_1$, and ${(\_\,\_)}_2$. Several cases have to be considered: The use of ${(\_\,\_)}_1$ versus ${(\_\,\_)}_2$ can be argued in a straightforward way using the inductive hypothesis.
The simultaneous application of two different axioms on top-level is obviously not possible.
The interesting remaining cases are the usage of one of the four axioms versus one of the structural rules.
\begin{enumerate}
\item
The first case is the use of $\beta_1^{\mu}$, i.e. $\mathbf{a}_1=[x:\mathbf{a}_3]\mathbf{a}_4$ and $\mathbf{b}=[x\mydef\mathbf{a}_2]\mathbf{a}_4$, versus one of the rules ${(\_\,\_)}_1$, and ${(\_\,\_)}_2$.
Two subcases need to be considered:
\begin{enumerate}
\item
Use of rule ${(\_\,\_)}_1$, i.e.~$\mathbf{c}=(\mathbf{c}_1\,\mathbf{a}_2)$ where $E\sgv\mathbf{a}_1=[x:\mathbf{a}_3]\mathbf{a}_4\smrd\mathbf{c}_1$:
By definition of $\smrd$, there are two subcases:
\begin{enumerate}
\item
$\mathbf{c}_1=[x:\mathbf{c}_3]\mathbf{a}_4$ where $E\sgv\mathbf{a}_3\smrd\mathbf{c}_3$: This means that $E\sgv\mathbf{c}\smrd[x\mydef\mathbf{a}_2]\mathbf{a}_4=\mathbf{b}$, i.e.~$\mathbf{d}=\mathbf{b}$ is a single-step reduct of $\mathbf{c}$.
This situation can be graphically summarized as follows (leaving out the environment $E$):
\[\begin{tikzcd}[row sep= scriptsize, column sep= large]
    \mathbf{a}=(\mathbf{a}_1\,\mathbf{a}_2)=([x:\mathbf{a}_3]\mathbf{a}_4\,\mathbf{a}_2)\arrow[r,"{:=}" description]\arrow[d,"{:=}"]& \mathbf{d}=\mathbf{b}=[x\mydef\mathbf{a}_2]\mathbf{a}_4\\
    \mathbf{c}=(\mathbf{c}_1\,\mathbf{a}_2)=([x:\mathbf{c}_3]\mathbf{a}_4\,\mathbf{a}_2)\arrow[ru,"{:=}" description]
\end{tikzcd}
\]
\item
$\mathbf{c}_1= [x:\mathbf{a}_3]\mathbf{c}_4$ where $E\sgv\mathbf{a}_4\smrd\mathbf{c}_4$:
By definition of $\smrd$ we know that also $E\sgv[x\mydef\mathbf{a}_2]\mathbf{a}_4\smrd[x\mydef\mathbf{a}_2]\mathbf{c}_4$.
Hence $\mathbf{d}=[x\mydef\mathbf{a}_2]\mathbf{c}_4$ with $E\sgv\mathbf{b}=[x\mydef\mathbf{a}_2]\mathbf{a}_4\smrd[x\mydef\mathbf{a}_2]\mathbf{c}_4=\mathbf{d}$ and $E\sgv\mathbf{c}=(\mathbf{c}_1\,\mathbf{a}_2)=([x:\mathbf{a}_3]\mathbf{c}_4\,\mathbf{a}_2)\smrd[x\mydef\mathbf{a}_2]\mathbf{c}_4=\mathbf{d}$.
This situation can be graphically summarized as follows (leaving out the environment $E$):
\[\begin{tikzcd}[row sep= scriptsize, column sep= large]
    \mathbf{a}=(\mathbf{a}_1\,\mathbf{a}_2)=([x:\mathbf{a}_3]\mathbf{a}_4\,\mathbf{a}_2)\arrow[r,"{:=}" description]\arrow[d,"{:=}"]&\mathbf{b}=[x\mydef\mathbf{a}_2]\mathbf{a}_4\arrow[d,"{:=}"]\\
    \mathbf{c}=(\mathbf{c}_1\,\mathbf{a}_2)=([x:\mathbf{a}_3]\mathbf{c}_4\,\mathbf{a}_2)\arrow[r,"{:=}" description] &\mathbf{d}=  [x\mydef\mathbf{a}_2]\mathbf{c}_4
\end{tikzcd}
\]
\end{enumerate}
\item
The use of rule ${(\_\,\_)}_2$ can be argued in a similar style.
\end{enumerate}
\item
The second case is the use of $\beta_2^{\mu}$ and can be argued in a similar style.
\item
The third case is the use of $\beta_3$, i.e.~$\mathbf{a}_1=\case{\mathbf{a}_3}{\mathbf{a}_4}$, $\mathbf{a}_2=\injl{\mathbf{a}_5}{\mathbf{a}_6}$ and $\mathbf{b}=\mathbf{a}_3(\mathbf{a}_5)$, versus one of the rules ${(\_\,\_)}_1$, and ${(\_\,\_)}_2$.
The property follows by an obvious case distinction on whether $E\sgv\case{\mathbf{a}_3}{\mathbf{a}_4}\smrd\mathbf{c}$ is reducing in $\mathbf{a}_3$ or $\mathbf{a}_4$, and similarly for $E\sgv\injl{\mathbf{a}_5}{\mathbf{a}_6}\smrd\mathbf{c}$.
\item
The fourth case is symmetric to the third one.
\qedhere
\end{enumerate}
\end{itemize}
\end{proof}
\begin{lem}[Confluence of $\smrd$]%
\label{mrd.confl}
For all $E,\mathbf{a},\mathbf{b},\mathbf{c}$: $E\sgv\mathbf{a}\mrd\mathbf{b}$ and $E\sgv\mathbf{a}\mrd\mathbf{c}$ implies $E\sgv\mathbf{b}\mrdr\mathbf{c}$.
\end{lem}
\begin{proof}
Using Lemma~\ref{mrd.confl.dir}, by structural induction on the number $n$ of transition steps one can show that  $E\sgv\mathbf{a}\mrdn{n}\mathbf{b}$ and $E\sgv\mathbf{a}\smrd\mathbf{c}$ implies $E\sgv\mathbf{b}\mrdn{01}\mathbf{d}$, and $E\sgv\mathbf{c}\mrdn{*}\mathbf{d}$ where for some $\mathbf{d}$.
Using this intermediate result, by structural induction on the number of transition steps $n$ one can show for any $m$ that $E\sgv\mathbf{a}\mrdn{n}\mathbf{b}$ and $E\sgv\mathbf{a}\mrdn{*}\mathbf{c}$ implies $E\sgv\mathbf{b}\mrdn{*}\mathbf{d}$ and $E\sgv\mathbf{c}\mrdn{*}\mathbf{d}$ for some $\mathbf{d}$.
This proves confluence of $\smrd$.
\end{proof}
By Lemma~\ref{rd.mrd} we know that reduction implies reduction with explicit substitution. The reverse direction is obviously not true. However we can show that $E\sgv\mathbf{a}\mrd\mathbf{b}$ implies $\mathbf{a}'\rd\mathbf{b}'$ where
$\mathbf{a}'$ and $\mathbf{b}'$ result from $\mathbf{a}$ and $\mathbf{b}$ by maximal evaluation of definitions, i.e.~by maximal application of the axioms \emph{use} and \emph{rem}.
Therefore we introduce the relation of \emph{definition evaluation}
\begin{defi}[Single-step definition-evaluation]
\emph{Single-step definition-evaluation} $E\sgv\mathbf{a}\dmrd\mathbf{b}$ is defined just like single-step reduction with explicit substitution $E\sgv\mathbf{a}\smrd\mathbf{b}$ (see Table~\ref{mred.rules}) but without the rule $\nu$ and without any axioms except \emph{rem} and \emph{use}.
\end{defi}
\begin{lem}[Confluence of $\dmrd$]%
\label{mrd.confl.def}
$\dmrd$ is confluent.
\end{lem}
\begin{proof}
By removing all the axiom cases except \emph{rem} and \emph{use} and the rule $\nu$ in the proof of Lemma~\ref{mrd.confl.dir} (where these axioms only interacted with the structural rules for internalized substitution) it can be turned into a proof of direct confluence of $\dmrd$.
The confluence of $\dmrd$
then follows as in the proof of Lemma~\ref{mrd.confl}.
\end{proof}
\begin{lem}[Strong normalization of definition-evaluation]%
\label{mrd.sn.def}
There are no infinite chains $E\sgv\mathbf{a}_1\dmrd\mathbf{a}_2 \dmrd\mathbf{a}_3 \dmrd \ldots$
\end{lem}
\begin{proof}
Since all definitions are non-recursive the property is straightforward.
One way to see this is to define a weight $W(E,\mathbf{a})$ such that
\[
E\gv\mathbf{a}\dmrd\mathbf{b}\quad\text{implies}\quad W(E,\mathbf{b})<W(E,\mathbf{a})
\]
This can be achieved by defining
$W(E,x)= W(E,E(x))+1$ if $x$ is in the domain of $E$ otherwise $W(E,x)=1$,
$W(E,[x\mydef\mathbf{a}]\mathbf{b})=W(E,\mathbf{a})+W((E,x:=\mathbf{a}),\mathbf{b})+1$, and so on.
\end{proof}
\begin{defi}[Definitional normal form]
For any $\mathbf{a}$, let $\dnf{E}(\mathbf{a})$ denote the \emph{definitional normal form}, i.e.~the expression resulting from $\mathbf{a}$ by maximal application of definition evaluation steps under environment $E$. This definition is sound since due to Lemma~\ref{mrd.sn.def} the evaluation of definitions always terminates and by Lemma~\ref{mrd.confl.def} the maximal evaluation of definitions delivers a unique result.
\end{defi}
\begin{lem}[Basic properties of $\dnf{E}(\mathbf{a})$]%
\label{dnf.basic}
For all $E,x,\mathbf{a}_1,\ldots \mathbf{a}_n$:
\begin{enumerate}
\item
$\dnf{E}(\binop{\mathbf{a}_1}{\ldots,\mathbf{a}_n})=\binop{\dnf{E}(\mathbf{a}_1)}{\ldots,\dnf{E}(\mathbf{a}_n)}$.
\item
$\dnf{E}(\binbop{x}{\mathbf{a}_1}{\ldots,\mathbf{a}_n})=\binbop{x}{\dnf{E}(\mathbf{a}_1)}{\ldots,\dnf{E}(\mathbf{a}_n)}$.
\item
$\dnf{E}([x\mydef \mathbf{a}_1]\mathbf{a}_2)=\dnf{E}(\mathbf{a}_2\gsub{x}{\mathbf{a}_1})=\dnf{E}(\mathbf{a}_2)\gsub{x}{\dnf{E}(\mathbf{a}_1)}=\dnf{E,x\smydef\mathbf{a}_1}(\mathbf{a}_2)$.
\end{enumerate}
\end{lem}
\begin{proof}
\hfill
\begin{enumerate}
\item
The proof is by induction on the length of an arbitrary definition evaluation $E\sgv\binop{\mathbf{a}_1}{\ldots,\mathbf{a}_n}\dmrd\dnf{E}(\binop{\mathbf{a}_1}{\ldots,\mathbf{a}_n})$.
It is obvious since the outer operation is never affected by definition evaluation.
\item
Similarly to (1)
\item
Using (1), (2), and the Lemmas~\ref{sub.basic},~\ref{rd.mrd}, and~\ref{mrd.confl.def}, these equalities can be shown either directly or by straightforward inductive arguments.
\qedhere
%\item
%Intuitively, due to the first three properties, the effect of $\dnf{E}$ is to replace some identifiers by expressions, which does not affect the reduction $a_1\rd a_2$.
%Furthermore, no axiom of reduction has a free identifier as its left-hand side.
%Hence this part can be shown by a straightforward induction on $a_1\rd a_2$.
\end{enumerate}
\end{proof}
\begin{lem}[Embedding property of reduction with explicit substitution]%
\label{mrd.rd}
For all $E,\mathbf{a},\mathbf{b}$:
$E\sgv\mathbf{a}\mrd\mathbf{b}$ implies $\dnf{E}(\mathbf{a})\rd\dnf{E}(\mathbf{b})$. As a consequence, for all $a$ and $b$, $\sgv a\mrdr b$ implies $a\rdr b$.
\end{lem}
\begin{proof}
First we show by induction on single-step reduction with explicit substitution that $E\sgv\mathbf{a}\smrd\mathbf{b}$ implies $\dnf{E}(\mathbf{a})\rd\dnf{E}(\mathbf{b})$.
The proof is straightforward using Lemmas~\ref{sub.rd} and~\ref{dnf.basic}, and the fact that, for any $E$, $\mathbf{a}\nurn{n}\mathbf{b}$ implies $\dnf{E}(\mathbf{a})\rdn{n}\dnf{E}(\mathbf{b})$ which follows by a straightforward induction on the definition of $\mathbf{a}\nurn{n}\mathbf{b}$.

Now we turn to the main proposition.
Obviously, we have $E\sgv\mathbf{a}\mrdn{n}\mathbf{b}$ for some $n$.
The proof is by induction on $n$. For $n=0$ the property is trivial. For $n>0$ assume  $E\sgv\mathbf{a}\smrd\mathbf{c}\mrdn{n-1}\mathbf{b}$. By inductive hypothesis $\dnf{E}(\mathbf{c})\rd\dnf{E}(\mathbf{b})$. By the above argument we know that $\dnf{E}(\mathbf{a})\rd\dnf{E}(\mathbf{c})$ which implies the proposition.

For the immediate consequence assume that there is an expression $\mathbf{c}$ such that $\sgv a\mrd\mathbf{c}$ and
$\sgv b\mrd\mathbf{c}$. We have just shown that also $\dnf{()}(a)\rd\dnf{()}(\mathbf{c})$ and $\dnf{()}(b)\rd\dnf{()}(\mathbf{c})$. The property follows since obviously $a=\dnf{()}(a)$ and $b=\dnf{()}(b)$. Hence $a\rdr b$.
\end{proof}
\begin{thm}[Confluence of $\srd$]%
\label{rd.confl}
For all $a,b,c$:
$a\rd b$ and $a\rd c$ implies $b\rdr c$.
As a consequence $a\eqv b$ implies $a\rdr b$.
\end{thm}
\begin{proof}
By Lemma~\ref{rd.mrd} we know that $\sgv a\mrd b$ and $\sgv a\mrd c$. Due to confluence of $\smrd$ (Lemma~\ref{mrd.confl}) we know that $\sgv b\mrdr c$.
By Lemma~\ref{mrd.rd} we obtain $b\rdr c$. The consequence follows by induction on the definition of $a\eqv b$.
\end{proof}
\noindent
As an immediate consequence of confluence we note some basic properties of congruence.
\begin{cor}[Basic properties of congruence]%
\label{congr.basic}
For all $x,y,a,b,c,d$:
\begin{enumerate}
\item $a\eqv b$ implies $a\gsub{x}{c}\eqv b\gsub{x}{c}$ and $c\gsub{x}{a}\eqv c\gsub{x}{b}$.
\item $\binbop{x}{a_1}{\ldots,a_n}\eqv\binbop{y}{b_1}{\ldots,b_n}$ iff $x=y$ and $a_i\eqv b_i$ for $i=1,\ldots,n$.
\item $\prsumop{a}{b}\eqv\prsumop{c}{d}$%$\injl{a}{b}\eqv\injl{c}{d}$, or $\injr{a}{b}\eqv\injr{c}{d}$,
iff $a\eqv c$ and $b\eqv d$.
\end{enumerate}
\end{cor}
\begin{proof}
Follow from Lemmas~\ref{sub.rd} and~\ref{rd.decomp}, and Theorem~\ref{rd.confl}.
\end{proof}
\subsection{Properties of typing}%
\label{type.prop}
Based on the confluence result we can now show the main properties of typing, in particular subject reduction and uniqueness of types.
We frequently show properties of typing by induction on the definition of typing where the inductive base corresponds to the type axiom and the inductive step corresponds to one of the type inference rules.
We begin by two properties (\ref{type.weak},\ref{type.xtrct}) related to the rule \weak, a property (\ref{start.confl}) related to the rule \mystart, and a property (\ref{eqv.env}) related to the rule \myconv.
\begin{lem}[Context weakening]%
\label{type.weak}
For all $\Gamma_1,\Gamma_2,x,a,b,c$:
$(\Gamma_1,\Gamma_2)\sgv a:b$ and $\Gamma_1\sgv c$ imply $(\Gamma_1,x:c,\Gamma_2)\sgv a:b$.
\end{lem}
\begin{proof}
Proof is by induction on the definition of $(\Gamma_1,\Gamma_2)\sgv a:b$.
For all the different type rules the property follows directly from the definition of typing and the inductive hypothesis.
\end{proof}
\begin{lem}[Context extraction]%
\label{type.xtrct}
For all $\Gamma_1,\Gamma_2,a,b$:
$(\Gamma_1,x:a,\Gamma_2)\sgv b$ implies $\Gamma_1\sgv a$
\end{lem}
\begin{proof}
$(\Gamma_1,x:a,\Gamma_2)\sgv b$ means that there is an expression $c$ where $(\Gamma_1,x:a,\Gamma_2)\sgv b:c$.
The property that $(\Gamma_1,x:a,\Gamma_2)\sgv b:c$ implies $\Gamma_1\sgv a$ is shown by induction on the definition of $(\Gamma_1,x:a,\Gamma_2)\sgv b:c$.
For all the different type rules the property follows directly from the definition of typing and the inductive hypothesis.
\end{proof}
\begin{lem}[Start property]%
\label{start.confl}
For all $\Gamma,x,a,b$:
$(\Gamma,x:a)\sgv x:b$ implies $a\eqv b$.
\end{lem}
\begin{proof}
The proof is by induction on $\Gamma,x:a\sgv x:b$.
Obviously the only relevant rules are \mystart\ and \myconv.
In case of \mystart, we directly obtain $a=b$.
In case of \myconv, we know that $\Gamma,x:a\sgv x:c$ for some $c$ where $c\eqv b$.
By inductive hypothesis it follows that $a\eqv c$ and hence obviously $a\eqv b$.
\end{proof}
\begin{lem}[Context equivalence and typing]%
\label{eqv.env}
Let $\Gamma_a=(\Gamma_1,x:a,\Gamma_2)$ and $\Gamma_b=(\Gamma_1,x:b,\Gamma_2)$ for some $\Gamma_1,\Gamma_2,x,a,b$ where $a\eqv b$ and $\Gamma_1\sgv b$:
For all $c,d$: If $\Gamma_a\sgv c:d$ then $\Gamma_b\sgv c:d$.
\end{lem}
\begin{proof}
Let $\Gamma_a$ and $\Gamma_b$ where $a\eqv b$ and $\Gamma_1\sgv b$ as defined above:
The property is shown by induction on the definition of $(\Gamma_1,x:a,\Gamma_2)\sgv c:d$.
In all cases, the property follows directly from the definition of typing and the inductive hypothesis.

As an example we show the case of rule \mystart:
We have $\Gamma_a=(\Gamma_1,x:a,\Gamma_2)=(\Gamma_3,y:d)$ and $c=y$ for some $y$ and $\Gamma_3$ where $\Gamma_3\sgv d$.
There are two cases:
\begin{itemize}
\item If $x=y$ then $\Gamma_a\sgv x:a$, i.e.~$d=a$, $\Gamma_2=()$, $\Gamma_1=\Gamma_3$, $\Gamma_a=(\Gamma_1,x:a)$, $\Gamma_b=(\Gamma_1,x:b)$,  and $\Gamma_1\sgv a$.
Since $\Gamma_1\sgv b$, by rule \weak~it follows that $\Gamma_b\sgv a$.
By rule \mystart~we obtain $\Gamma_b\sgv x:b$.
Hence by rule \myconv~it follows that $\Gamma_b\sgv x:a$.
This can be graphically illustrated as follows:
\[
\inferrule*[left={\normalfont\myconv}]{
		\inferrule*[left={\normalfont\mystart}]{\Gamma_1\sgv b}{\Gamma_1,x:b\gv x:b}
		\quad b\eqv a\quad
	\inferrule*[left={\normalfont\weak}]{\Gamma_1\sgv a\quad\Gamma_1\sgv b}{\Gamma_1,x:b\sgv a}}
{\Gamma_b=(\Gamma_1,x:b)\gv x:a}
\]
Note that the proof of this case one does not need the assumption $\Gamma_a =(\Gamma_1,x:a)\sgv x:a$.
\item
If $x\neq y$ then we have $\Gamma_2=(\Gamma_4,y:d)$ for some $\Gamma_4$ where $\Gamma_3=(\Gamma_1,x:a,\Gamma_4)$.
Since $\Gamma_1,x:a,\Gamma_4\sgv d$, by inductive hypothesis we know that also $\Gamma_1,x:b,\Gamma_4\sgv d$ and therefore, by rule \mystart, we obtain $\Gamma_b=(\Gamma_1,x:b,\Gamma_2)\sgv y:d$.
\qedhere
\end{itemize}
\end{proof}
\noindent
Next we list several decomposition properties of typing and of validity.
\begin{lem}[Typing decomposition]%
\label{type.decomp}
For all $\Gamma,x,a_1,a_2,a_3,b$:
\begin{enumerate}
\item
$\Gamma\sgv\binbop{x}{a_1}{a_2}:b$ implies $b\eqv[x:a_1]c$ for some $c$ where $(\Gamma,x:a_1)\sgv a_2:c$.
\item
$\Gamma\sgv\injl{a_1}{a_2}:b$ implies $b\eqv[c+a_2]$ for some $c$ where $\Gamma\sgv a_1:c$ and $\Gamma\sgv a_2$.

\noindent
$\Gamma\sgv\injr{a_1}{a_2}:b$ implies $b\eqv[a_1+c]$ for some $c$ where $\Gamma\sgv a_2:c$ and $\Gamma\sgv a_1$.
\item
$\Gamma\sgv\prsumop{a_1}{a_2}:b$ implies $b\eqv[c_1,c_2]$ for some $c_1$, $c_2$ where $\Gamma\sgv a_1:c_1$ and $\Gamma\sgv a_2:c_2$.
\item
$\Gamma\sgv\myneg a_1:b$ implies $b\eqv c$ for some $c$ where $\Gamma\sgv a_1:c$.
\item
$\;\sgv\prim:b$ implies $b\eqv\prim$.
\item
$\Gamma\sgv(a_1\,a_2):b$ implies $b\eqv c_2\gsub{x}{a_2}$ where $\Gamma\sgv a_1:[x\!:\!c_1]c_2$ and $\Gamma\sgv a_2\!:\!c_1$ for some $c_1$,$c_2$.
\item
$\Gamma\sgv\pleft{a_1}:b$ implies $b\eqv c_1$ where $\Gamma\sgv a_1:[x!c_1]c_2$ or $\Gamma\sgv a_1:[c_1,c_2]$ for some $c_1$, $c_2$.
\item
$\Gamma\sgv\pright{a_1}:b$ implies, for some $c_1$, $c_2$, that either $b\eqv c_2$ where $\Gamma\sgv a_1:[c_1,c_2]$ or $b\eqv c_2\gsub{x}{\pleft{a_1}}$ where $\Gamma\sgv a_1:[x!c_1]c_2$.
\item
$\Gamma\sgv\prdef{x}{a_1}{a_2}{a_3}:b$ implies $b\eqv[x!c]a_3$ for some $c$ where $\Gamma\sgv a_1:c$, $\Gamma\sgv a_2:a_3\gsub{x}{a_1}$, and $(\Gamma,x:c)\sgv a_3$.
\item
$\Gamma\sgv\case{a_1}{a_2}:b$ implies $b\eqv[x:[c_1+c_2]]c$ for some $c_1$, $c_2$, and $c$ where $\Gamma\sgv a_1:[x:c_1]c$, $\Gamma\sgv a_2:[x:c_2]c$ and $\Gamma\sgv c$.
\end{enumerate}
\end{lem}
\begin{proof}
In all cases the proof is by induction on the definition of typing or by a straightforward application of previous properties.
In particular,~\ref{rd.decomp},~\ref{rd.confl}, and~\ref{type.weak} are needed.
\end{proof}
\begin{lem}[Validity decomposition]%
\label{val.decomp}
For all $\Gamma,x,a_1,\ldots,a_n$:
\begin{enumerate}
\item
$\Gamma\sgv\binop{a_1}{,\ldots a_n}$ implies $\Gamma\sgv a_1$, $\ldots$, $\Gamma\sgv a_n$.
\item
$\Gamma\sgv\binbop{x}{a_1}{a_2}$ implies $\Gamma\sgv a_1$ and $(\Gamma,x:a_1)\sgv a_2$
\item
$\Gamma\sgv\prdef{x}{a_1}{a_2}{a_3}$ implies $\Gamma\sgv a_2$ and $\Gamma\sgv a_1:b$ for some $b$ where $(\Gamma,x:b)\sgv a_3$
\end{enumerate}
\end{lem}
\begin{proof}
In all cases the proof is by induction on the definition of typing or by a straightforward application of previous properties.
In particular, Lemma~\ref{type.xtrct} is needed.
\end{proof}
\noindent
A central prerequisite to the proof of closure of reduction and typing against validity is the following substitution property of typing.
In order to state the property, we need an auxiliary definition.
\begin{defi}[Context substitution]
The substitution function (Definition~\ref{free}) is extended to contexts $\Gamma\gsub{x}{a}$, where $a$ is an expression, as follows:
\begin{eqnarray*}
()\gsub{x}{a}&=&()\\
(y:b,\Gamma)\gsub{x}{a}&=&
\begin{cases}
(y:b\gsub{x}{a},\Gamma)&\text{if}\;x=y\\
(y:b\gsub{x}{a},\Gamma\gsub{x}{a})&\text{otherwise}
\end{cases}
\end{eqnarray*}
\end{defi}
\begin{lem}[Substitution and typing]%
\label{type.sub}
Assume that $\Gamma_a=(\Gamma_1,x:a,\Gamma_2)$ and $\Gamma_b=(\Gamma_1,\Gamma_2\gsub{x}{b})$ for some $\Gamma_1,\Gamma_2,x,a,b$ where $\Gamma_1\sgv b:a$.
For all $c,d$: If $\Gamma_a\sgv c:d$ then $\Gamma_b\sgv c\gsub{x}{b}:d\gsub{x}{b}$.
\end{lem}
\begin{proof}
Let $\Gamma_a$ and $\Gamma_b$ as defined above and assume $\Gamma_1\sgv b:a$.
The proof that $\Gamma_a\sgv c:d$ implies $\Gamma_b\sgv c\gsub{x}{b}:d\gsub{x}{b}$ is by induction on the definition of $\Gamma_a\sgv c:d$.
We show two interesting cases.
\begin{itemize}
\item Rule \myconv:
We have $\Gamma_a\sgv c:d'$ for some $d'$ where $d'\eqv d$ and $\Gamma_a\sgv d':e$ for some $e$.
By inductive hypothesis  $\Gamma_b\sgv c\gsub{x}{b}:d'\gsub{x}{b}$ and $\Gamma_b\sgv d'\gsub{x}{b}:e\gsub{x}{b}$.
By Lemma~\ref{congr.basic}(1) we can infer that $d'\gsub{x}{b}\eqv d\gsub{x}{b}$. % chktex 36
Therefore we can apply the rule \myconv~to obtain $\Gamma_b\sgv c\gsub{x}{b}:d\gsub{x}{b}$.
This can be graphically illustrated as follows (IH~denotes the use of an inductive hypothesis):
\[
\inferrule*[left={\normalfont\myconv}]{
		\inferrule*[left=\ref{congr.basic}(1)]{d'\eqv d}{d'\gsub{x}{b}\eqv d\gsub{x}{b}}\;
  \inferrule*[left=IH]{\Gamma_a\sgv c:d'}{\Gamma_b\sgv c\gsub{x}{b}:d'\gsub{x}{b}}\;
	\inferrule*[left=IH]{\Gamma_a\sgv d':e}{\Gamma_b\sgv d'\gsub{x}{b}:e\gsub{x}{b}}}
{\Gamma_b\gv c\gsub{x}{b}:d\gsub{x}{b}}
\]
\item Rule \appl:
We have $c=(c_1\,c_2)$, $d=d_2\gsub{y}{c_2}$ for some $y$, $c_1$, $c_2$, and $d_2$ where $\Gamma_a\sgv (c_1\,c_2):d$, $\Gamma_a\sgv c_1:[y:d_1]d_2$, and $\Gamma_a\sgv c_2:d_1$.
Obviously, we can assume that $x\neq y$.
We need to show that $\Gamma_b\sgv(c_1\,c_2)\gsub{x}{b}:d_2\gsub{y}{c_2}\gsub{x}{b}$.
We have:
\begin{eqnarray*}
\Gamma_b\gv &&c_1\gsub{x}{b}\\
&:&\quad\text{(inductive hypothesis on $\Gamma_a\sgv c_1:[y:d_1]d_2$)}\\
&&([y:d_1]d_2)\gsub{x}{b}\\
&=&\quad\text{(definition of substitution)}\\
&&[y:d_1\gsub{x}{b}](d_2\gsub{x}{b})
\end{eqnarray*}
Furthermore, since $\Gamma_a\sgv c_2:d_1$, by inductive hypothesis we know that $\Gamma_b\sgv c_2\gsub{x}{b}:d_1\gsub{x}{b}$.
Hence by type rule \appl~we obtain
\[\Gamma_b\sgv(c_1\gsub{x}{b}\,c_2\gsub{x}{b}):d_2\gsub{x}{b}\gsub{y}{c_2\gsub{x}{b}}\qquad (\ast)
\]
We can now argue as follows:
\begin{eqnarray*}
\Gamma_b\gv &&(c_1\,c_2)\gsub{x}{b}\\
&=&\quad\text{(definition of substitution)}\\
&&(c_1\gsub{x}{b}\,c_2\gsub{x}{b})\\
&:&\quad\text{(see $\ast$)}\\
&&d_2\gsub{x}{b}\gsub{y}{c_2\gsub{x}{b}}\\
&=&\quad\text{(Lemma~\ref{sub.basic}(2))}\\ % chktex 36
&&d_2\gsub{y}{c_2}\gsub{x}{b}
\end{eqnarray*}
\end{itemize}
The other cases can be shown in a similar style.
\end{proof}

\noindent
We begin the closure properties with the relatively straightforward property that validity is closed against typing.
\begin{lem}[Valid expressions have valid types]%
\label{valid.type}
For all $\Gamma,a,b$:
If $\Gamma\sgv a:b$  then $\Gamma\sgv b$.
\end{lem}
\begin{proof}
By induction on the definition of $\Gamma\sgv a:b$ we show that $\Gamma\sgv a:b$ implies $\Gamma\sgv b$.
As an example, we illustrate the rule for applications
\begin{itemize}
\item Rule \appl:
We have $a=(a_1\,a_2)$ and $b=b_2\gsub{x}{a_2}$ for some $x$, $a_1$, $a_2$, and $b_2$ where $\Gamma\sgv a_1:[x:b_1]b_2$ for some $b_1$ and $\Gamma\sgv a_2:b_1$.
From the inductive hypothesis we know that $\Gamma\sgv[x:b_1]b_2$.
By Lemma~\ref{val.decomp}(2) this implies $\Gamma,x:b_1\sgv b_2$. % chktex 36
From $\Gamma,x:b_1\sgv b_2$ and from substitution (Lemma~\ref{type.sub}, note that $\Gamma\sgv a_2:b_1$) we can infer that $\Gamma\sgv b_2\gsub{x}{a_2}=b$.
\end{itemize}
The other cases can be shown in a similar style.
\end{proof}
\noindent
We are now ready to show the preservation of types under a reduction step.
\begin{lem}[Preservation of types under reduction steps]%
\label{strd.type}
For all $\Gamma,a,b,c$:
$\Gamma\sgv a:c$ and $a\srd b$ imply $\Gamma\sgv b:c$.
\end{lem}
\begin{proof}
Proof by induction on the definition of $a\srd b$.
Note that by Lemma~\ref{valid.type} we have $\Gamma\sgv c$.
This means that in order to show $\Gamma\sgv b:c$, it is sufficient to show $\Gamma\sgv b:b_1$ for some $b_1$ where $b_1\eqv c$.
$\Gamma\sgv b:c$ then follows by applying rule \myconv.
We show the cases of axiom $\beta_1$ and of structural rule $\binbop{x}{\_}{\_}_1$ in detail:
\begin{itemize}
\item Axiom $\beta_1$:
We have $a=([x:a_1]a_2\;a_3)$ and $b=a_2\gsub{x}{a_3}$ and $\Gamma\sgv([x:a_1]a_2\;a_3):c$ where $\Gamma\sgv a_3:a_1$.
The following type relations can be derived:
\begin{eqnarray*}
&&\Gamma\sgv([x:a_1]a_2\;a_3):c\\
&\Rightarrow&\quad\text{(Lemma~\ref{type.decomp}(6))}\\ % chktex 36
&&\Gamma\sgv[x:a_1]a_2:[x:d]e\quad\text{where $c\eqv e\gsub{x}{a_3}$ ($\ast$) and $\Gamma\sgv a_3:d$}\\
&\Rightarrow&\quad\text{(Lemma~\ref{type.decomp}(1))}\\ % chktex 36
&&(\Gamma,x:a_1)\sgv a_2:b_2\quad\text{where $[x:d]e\eqv[x:a_1]b_2$}\\
&\Rightarrow&\quad\text{(Lemma~\ref{type.sub}, since $\Gamma\sgv a_3:a_1$)}\\
&&\Gamma\sgv a_2\gsub{x}{a_3}:b_2\gsub{x}{a_3}
\end{eqnarray*}
From $[x:d]e\eqv[x:a_1]b_2$, by Lemma~\ref{congr.basic}(2) it follows that $e\eqv b_2$. % chktex 36
We can therefore argue as follows:
\begin{eqnarray*}
&&b_2\gsub{x}{a_3}\\
&\eqv&\quad\text{(by Lemma~\ref{congr.basic}(1) since $e\eqv b_2$)}\\ % chktex 36
&&e\gsub{x}{a_3}\\
&\eqv&\quad\text{(see $\ast$ in the first argument chain)}\\
&&c
\end{eqnarray*}
Therefore we can apply rule \myconv~to derive $\Gamma\sgv b=a_2\gsub{x}{a_3}:c$.
This can be graphically illustrated as follows:
\[
\inferrule*[left={\normalfont\myconv}]{\Gamma\sgv a_2\gsub{x}{a_3}:b_2\gsub{x}{a_3}\qquad b_2\gsub{x}{a_3}\eqv c\qquad\Gamma\sgv c}{\Gamma\gv a_2\gsub{x}{a_3}:c}
\]
\item Rule $\binbop{x}{\_}{\_}_1$:
We have $a=\binbop{x}{a_1}{a_3}$, $b=\binbop{x}{a_2}{a_3}$ where $a_1\srd a_2$, and $\Gamma\sgv\binbop{x}{a_1}{a_3}:c$.
By Lemma~\ref{type.decomp}(1), $c\eqv[x:a_1]c_3$ for some $c_3$ where $(\Gamma,x:a_1)\sgv a_3:c_3$. % chktex 36
By Lemma~\ref{type.xtrct} this implies that $\Gamma\sgv a_1$ and therefore by inductive hypothesis we know that $\Gamma\sgv a_2$.
Therefore we can apply Lemma~\ref{eqv.env} which implies that $(\Gamma,x:a_2)\sgv a_3:c_3$.
By definition of typing $\Gamma\sgv \binbop{x}{a_2}{a_3}:[x:a_2]c_3$.
Since $[x:a_2]c_3\eqv[x:a_1]c_3\eqv c$, we can apply type rule \myconv~to derive $\Gamma\sgv b=\binbop{x}{a_2}{a_3}:c$.
\end{itemize}
The other cases can be shown in a similar style using also Lemma~\ref{sub.rd}.
\end{proof}
\noindent
A simple inductive argument extends Lemma~\ref{strd.type} to general reduction.
This property is often referred to as \emph{subject reduction}.
\begin{thm}[Subject reduction: Types are preserved under reduction]%
\label{rd.type}
For all $\Gamma,a,b,c$:
$a\rd b$ and $\Gamma\sgv a:c$ imply that $\Gamma\sgv b:c$.
\end{thm}
%\begin{proof}
%Proof by induction on the number of reduction steps in $a\rdn{n} b$ using Lemma~\ref{strd.type}.
%\end{proof}
\noindent
Theorem~\ref{rd.type} can be reformulated using the validity notation.
\begin{cor}[Valid expressions are closed against reduction]%
\label{valid.rd}
For all $\Gamma,a,b$:
$\Gamma\sgv a$ and $a\rd b$ implies $\Gamma\sgv b$
\end{cor}
\noindent
Subject reduction can be extended analogously to types.
\begin{lem}[Elements are preserved under type reduction]%
\label{type.rd}
For all $\Gamma,a,b,c$:
$\Gamma\sgv a:b$ and $b\rd c$ implies $\Gamma\sgv a:c$.
\end{lem}
\begin{proof}
By Lemma~\ref{valid.type} $\Gamma\sgv b$.
By Corollary~\ref{valid.rd} $\Gamma\sgv c$.
The proposition then follows from rule \myconv.
\end{proof}
%\noindent
%
%A straightforward argument leads to a closure result of validity w.r.t. context reduction.
%
%\begin{lem}[Validity is closed against context reduction]
%\label{valid.rd.con}
%For all $\Gamma_1,\Gamma_2,x,a,b,c$: If
%$(\Gamma_1,x:b,\Gamma_2)\sgv a$ and $b\rd c$ then $(\Gamma_1,x:c,\Gamma_2)\sgv a$
%\end{lem}
%
%\begin{proof}
%From $(\Gamma_1,x:b,\Gamma_2)\sgv a$ by repeated application of the type rule for universal abstraction we can infer $(\Gamma_1,x:b)\sgv[\Gamma_2]a$.
%By Lemma~\ref{type.xtrct}, it follows that $\Gamma\sgv b$.
%By Corollary~\ref{valid.rd} we know that $\Gamma_1\sgv c$.
%Obviously $b\eqv c$ and hence Lemma~\ref{eqv.env} can be applied which implies $(\Gamma_1,x:c)\sgv[\Gamma_2]a$.
%By repeated application of Lemma~\ref{val.decomp}(2) we obtain $(\Gamma_1,x:c,\Gamma_2)\sgv a$.
%\end{proof}
%
\noindent
Before we show uniqueness of types, we need a Lemma about the effect of removing unnecessary type declarations from a context.

\begin{lem}[Context contraction]%
\label{context.reduction}
For all $\Gamma_1,\Gamma_2,x,a,b,c$:
$(\Gamma_1,x:c,\Gamma_2)\sgv a:b$ and $x\notin\free([\Gamma_2]a)$ implies
 $(\Gamma_1,\Gamma_2)\sgv a:b'$ and $b\rd b'$ for some $b'$.
\end{lem}
Note that in this lemma $x\notin\free([\Gamma_2]a)$ does not necessarily imply $x\notin\free([\Gamma_2]b)$, for example take  $x:\prim\sgv\prim:([y:\prim]\prim\,x)$, and hence the reduction of the type to some expression where $x$ does not occur free, in the example $([y:\prim]\prim\,x)\rd \prim$, is necessary.
\begin{proof}
The proof is by induction on $a$. In the base cases of $a=\prim$ and $a=y$ context contraction can be show by a simple induction on the definition of typing using Theorem~\ref{rd.confl} and Lemma~\ref{rd.type}.
We show two other interesting cases.
\begin{itemize}
\item $a=\binbop{x}{a_1}{a_2}$:
By~\ref{type.decomp}(1) $(\Gamma_1,x:c,\Gamma_2)\sgv\binbop{y}{a_1}{a_2}:b$ implies $b\eqv[y:a_1]d$ for some $d$ where $(\Gamma_1,x:c,\Gamma_2,y:a_1)\sgv a_2:d$. % chktex 36
Since $x\notin\free([\Gamma_2]\binbop{x}{a_1}{a_2})=\free([\Gamma_2,y:a_1]a_2)$, by inductive hypothesis $(\Gamma_1,\Gamma_2,y:a_1)\sgv a_2:d'$ where $d\rd d'$ for some $d'$.
By definition of typing this implies $(\Gamma_1,\Gamma_2)\sgv\binbop{y}{a_1}{a_2}:[y:a_1]d'$.
Since $b\eqv[y:a_1]d'$ by~\ref{rd.confl} and~\ref{type.rd} we know that $(\Gamma_1,\Gamma_2)\sgv\binbop{y}{a_1}{a_2}:b'$ where $b'$ is a common reduct of $b$ and $[y:a_1]d'$.
\item $a=(a_1\:a_2)$:
By~\ref{type.decomp}(6) we have  $b\eqv c_2\gsub{y}{a_2}$ where $y\neq x$, $(\Gamma_1,x:c,\Gamma_2)\sgv a_1:[y:c_1]c_2$ and $(\Gamma_1,x:c,\Gamma_2)\sgv a_2:c_1$ for some $c_1$, $c_2$. % chktex 36
Since $x\notin\free([\Gamma_2](a_1\,a_2))$ we know that $x\notin\free([\Gamma_2]a_1)\cup\free([\Gamma_2]a_2)$.
By inductive hypothesis $(\Gamma_1,\Gamma_2)\sgv a_1:d$ where $[y:c_1]c_2\rd d$, for some $d$, and $(\Gamma_1,\Gamma_2)\sgv a_2:e$ where $c_1\rd e$, for some $e$.

By~\ref{rd.decomp}(1) we know that $d=[y:d_1]d_2$ where $c_1\rd d_1$ and $c_2\rd d_2$, for some $d_1$ and $d_2$. % chktex 36
Since $e\eqv d_1$, by~\ref{rd.confl} and~\ref{type.rd} we know that $(\Gamma_1,\Gamma_2)\sgv a_2:e'$ where $e'$ is a common reduct of $e$ and $d_1$.
Since $[y:d_1]d_2\rd [y:e']d_2$, by~\ref{type.rd} we know that $(\Gamma_1,\Gamma_2)\sgv a_1:[y:e']d_2$.

Hence by definition of typing $(\Gamma_1,\Gamma_2)\sgv (a_1\:a_2):d_2\gsub{y}{a_2}$.
Since $b\eqv c_2\gsub{y}{a_2}\eqv d_2\gsub{y}{a_2}$, by~\ref{rd.confl} and~\ref{type.rd} we know that
$(\Gamma_1,\Gamma_2)\sgv a:b'$ where  $b'$ is a common reduct of $b$ and $d_2\gsub{y}{a_2}$.
\item $a=\prdef{y}{a_1}{a_2}{a_3}$:
By~\ref{type.decomp}(9) we have $b\eqv[y!d]a_3$ for some $d$ where $(\Gamma_1,x:c,\Gamma_2)\sgv a_1:d$, $(\Gamma_1,x:c,\Gamma_2)\sgv a_2:a_3\gsub{y}{a_1}$, and $(\Gamma_1,x:c,\Gamma_2,y:d)\sgv a_3$. % chktex 36
Since $x\notin\free([\Gamma_2]\prdef{y}{a_1}{a_2}{a_3})$ we know that $x\notin\free([\Gamma_2]a_1)\cup\free([\Gamma_2]a_2)\cup\free([\Gamma_2]a_3)$.
By inductive hypothesis $(\Gamma_1,\Gamma_2)\sgv a_1:d'$ where $d\rd d'$ for some $d'$ and $(\Gamma_1,\Gamma_2)\sgv a_2:e$ where $a_3\gsub{y}{a_1}\rd e$ for some $e$.

By~\ref{valid.type} we know that $(\Gamma_1,\Gamma_2)\sgv d'$.
By~\ref{type.xtrct} we have $\Gamma_1\sgv c$ and therefore by~\ref{type.weak} we know that $(\Gamma_1,x:c,\Gamma_2)\sgv d'$.
Hence from $(\Gamma_1,x:c,\Gamma_2,y:d)\sgv a_3$ by~\ref{eqv.env} we can conclude that $(\Gamma_1,x:c,\Gamma_2,y:d')\sgv a_3$.
By inductive hypothesis, since $x\notin\free([\Gamma_2,y:d']a_3)$, we obtain $(\Gamma_1,\Gamma_2,y:d')\sgv a_3$.
Since $(\Gamma_1,\Gamma_2)\sgv a_1:d'$ we can apply~\ref{type.sub} to obtain $(\Gamma_1,\Gamma_2)\sgv a_3\gsub{y}{a_1}$.
Since $(\Gamma_1,\Gamma_2)\sgv a_2:e$ and $e\eqv a_3\gsub{y}{a_1}$ by rule \myconv\ we obtain $(\Gamma_1,\Gamma_2)\sgv a_2:a_3\gsub{y}{a_1}$.

Hence by definition of typing $\Gamma_1,\Gamma_2\sgv a:[y!d']a_3$ which implies $\Gamma_1,\Gamma_2\sgv a:b'$ where $b\rd b'=[y!d']a_3$.
\end{itemize}
The other cases can be shown in a similar style.
\end{proof}
\noindent
We can now show uniqueness of types. Most cases are straightforward, except for the weakening rule where we need the context contraction result above.
\begin{thm}[Uniqueness of types]%
\label{type.confl}
For all $\Gamma,a,b,c$:
$\Gamma\sgv a:b$ and $\Gamma\sgv a:c$ implies $b\eqv c$.
\end{thm}
\begin{proof}
Proof by induction on the definition of $\Gamma\sgv a:b$.
We show three interesting cases:
In each case, if we look at a type rule $\Gamma\sgv a:b$ we have to show that if also $\Gamma\sgv a:c$ then $b\eqv c$.
\begin{itemize}
\item Rule \mystart:
We have $a=x$, $\Gamma=(\Gamma',x:b)$ for some $\Gamma'$ and $x$ where $(\Gamma',x:b)\sgv x:b$ and $\Gamma'\sgv b:d$ for some $d$:
Let $(\Gamma',x:b)\sgv x:c$.
By Lemma~\ref{start.confl} we know that $b\eqv c$.
\item Rule \weak:
We have $\Gamma=(\Gamma',x:d)$ for some $\Gamma'$ and $x$ where $(\Gamma',x:d)\sgv a:b$, $\Gamma'\sgv a:b$, and $\Gamma'\sgv d:e$ for some $e$.

Let $(\Gamma',x:d)\sgv a:c$.
Since $x\notin\free(a)$, by Lemma~\ref{context.reduction} we know that $\Gamma'\sgv a:c'$ for some $c'$ with $c\rd c'$.
By inductive hypothesis $b\eqv c'$ which implies $b\eqv c$.
\item Rule \appl:
We have $a=(a_1\,a_2)$ and $b=b_2\gsub{x}{a_2}$ for some $x$, $a_1$, $a_2$, and $b_2$ where $\Gamma\sgv(a_1\,a_2):b_2\gsub{x}{a_2}$,  $\Gamma\sgv a_1:[x:b_1]b_2$, and $\Gamma\sgv a_2:b_1$.

Let $\Gamma\sgv(a_1\,a_2):c$.
By Lemma~\ref{type.decomp}(6) we know that $c\eqv c_2\gsub{y}{a_2}$ for some $c_1$, $c_2$ where $\Gamma\sgv a_1:[y:c_1]c_2$ and $\Gamma\sgv a_2:c_1$. % chktex 36
By inductive hypothesis applied to $\Gamma\sgv a_1:[x:b_1]b_2$ it follows that $[x:b_1]b_2\eqv [y:c_1]c_2$.
Hence obviously $x=y$ and by basic properties of congruence (Lemma~\ref{congr.basic}(2)) it follows that $b_2\eqv c_2$. % chktex 36
Using Lemma~\ref{congr.basic}(1) we can argue $b\eqv b_2\gsub{x}{a_2}\eqv c_2\gsub{x}{a_2}\eqv c$. % chktex 36
\end{itemize}
The other cases can be shown in a similar style.
\end{proof}
\subsection{Strong normalization}
The idea for the proof of strong normalization of valid expressions in \dcalc\ is to classify expressions according to structural properties, in order to make an inductive argument work.
For this purpose we define structural skeletons called \emph{norms} as the subset of expressions built from the primitive $\prim$ and products only and we define a partial function assigning norms to expressions.
Norms are a reconstruction of a concept of \emph{simple types}, consisting of atomic types and product types, within \dcalc\footnote{The construction of norms inside \dcalc~is due to convenient reuse of existing structures and definitions, norms could also be introduced as a separate mathematical structure}.
Analogously to simple types, norms allow for classifying valid expressions by different degrees of structural complexity.
The idea of the strong normalisation argument is first to prove that all valid expressions are \emph{normable}, i.e~they are in the domain of the norming function, and then to prove that all normable expressions are \emph{strongly normalizable}, i.e~they are not the origin of an infinite reduction sequence.

The good news is that in \dcalc\ we are not dealing with unconstrained parametric types as for example in System F (see \eg~\cite{girard1989proofs}),
and therefore we will be able to use more elementary methods to show strong normalisation as used for the simply typed lambda calculus.
A common such method is to define a notion of \emph{reducible} expressions satisfying certain \emph{reducibility conditions} suitable for inductive arguments both on type structure and on reduction length, and then to prove that all reducible expressions satisfy certain \emph{reducibility properties} including strong normalisation, and to prove that that all typable expressions are reducible (\cite{tait1967intensional},\cite{Girard72}).
We will basically adopt this idea, but the bad news is that common definitions of reducibility (\eg~\cite{girard1989proofs}) apparently cannot be adapted to include the reduction of negation.
A more suitable basis for our purposes was found to be the notion of \emph{computable expressions} as defined in language theoretical studies of Automath~\cite{vDaalen77}, see also~\cite{VANBENTHEMJUTTING1994}.
This approach is basically extended here to cover additional operators, including negation.

We motivate the basic idea of the proof (precise definitions can be found below):
Consider the following condition necessary to establish strong normalization for an application $(a\,b)$ in the context of an inductive proof
(where $\sn{}$ denotes the set of strongly normalizable expressions):
\begin{itemize}
\item
If $a\rd[x:c_1]c_2\in\sn{}$ and $b\in\sn{}$ then $c_2\gsub{x}{b}\in\sn{}$.
\end{itemize}
Similarly, the following condition is necessary to establish strong normalization for a negation $\myneg a$ in the context of an inductive proof.
\begin{itemize}
\item
If $a\rd[x:c_1]c_2\in\sn{}$ then $\myneg c_2\in\sn{}$.
\end{itemize}
\noindent
These and other properties inspire the definition of the set of \emph{computable} expressions $C_{\Gamma}$ which
are normable, strongly normalizable, and satisfy the property that $a,b\in C_{\Gamma}$ implies $\myneg a,(a\,b)\in C_{\Gamma}$.
Unfortunately, the closure properties of computable expressions do not include existential and universal abstraction.
Instead, we need to prove the stronger property
that normable expressions are computable under norm-preserving substitutions of their free variables to computable expressions.

This property implies that all normable expressions are computable and therefore that normability and computability are equivalent notions.
The logical relations between the various notions can be summarized as follows ($\Gamma\sngv a$ denotes normability of $a$ under $\Gamma$):
\[
\Gamma\sgv a\quad\Rightarrow\quad\Gamma\sngv a\quad\Leftrightarrow\quad a\in\ce_{\Gamma}\quad\Rightarrow\quad a\in\sn{}
\]
We begin with some basic definitions and properties related to strong normalization.
\begin{defi}[Strongly normalizable expressions]%
\label{sn}
The set of \emph{strongly normalizable expressions} is denoted by $\sn{}$.
An expression $a$ is in $\sn{}$ iff there is no infinite sequence of one-step reductions $a\srd a_1\srd a_2\srd a_3\ldots$
\end{defi}
\begin{lem}[Basic properties of strongly normalizable expressions]%
\label{sn.basic}
For all $a,b,c,x$:
\begin{enumerate}
\item
$a\gsub{x}{b}\in\sn{}$ and $b\rd c$ imply $a\gsub{x}{c}\in\sn{}$
\item
$a,b,c\in\sn{}$ implies $\binbop{x}{a}{b}\in\sn{}$, $\pleft{a}$, $\pright{a}\in\sn{}$, $\prdef{x}{a}{b}{c}\in\sn{}$, $[a,b]$, $[a+b]$, $\injl{a}{b}$, $\injr{a}{b}$, and $\case{a}{b}\in\sn{}$.
\end{enumerate}
\end{lem}
\begin{proof}
These properties can be shown through a proof by contradiction using elementary properties of reduction (Lemmas~\ref{sub.rd} and~\ref{rd.decomp}).
\end{proof}
\noindent
Next we show conditions under which applications and negations are strongly normalizing.
\begin{lem}[Strong normalization conditions]%
\label{sn.cond}
$\;$
\begin{enumerate}
\item
For all $a,b\in\sn{}$: $(a\,b)\in\sn{}$ if for any $c_1$, $c_2$, $d_1$, $d_2$, and $x$, the following conditions are satisfied:
\begin{itemize}
\item[$(C_1)$] $a\rd\binbop{x}{c_1}{c_2}$ implies that  $c_2\gsub{x}{b}\in\sn{}$
\item[$(C_2)$] $a\rd\case{c_1}{c_2}$ and $b\rd\injl{d_1}{d_2}$ implies $(c_1\,d_1)\in\sn{}$
\item[$(C_3)$] $a\rd\case{c_1}{c_2}$ and $b\rd\injr{d_1}{d_2}$ implies $(c_2\,d_2)\in\sn{}$
\end{itemize}
\item
For all $a\in\sn{}$: $\myneg a\in\sn{}$ if for any $b$ and $c$, the following conditions are satisfied:
\begin{itemize}
\item[$(C_1)$] $a\rd\prsumop{b}{c}$ implies $\myneg b$, $\myneg c\in\sn{}$
\item[$(C_2)$] $a\rd\binbop{x}{b}{c}$ implies $\myneg c\in\sn{}$
\end{itemize}
\end{enumerate}
\end{lem}
\begin{proof}
These properties can be shown through a proof by contradiction using elementary properties of reduction (Lemmas~\ref{sub.rd},~\ref{rd.decomp}, and~\ref{rd.decomp.neg}) and of strong normalisation (Lemma~\ref{sn.basic}).
\end{proof}

\noindent
Norms are a subset of expressions representing structural skeletons of expressions.
Norms play an important role to classify expressions in the course of the proof of strong normalization.
\begin{defi}[Norm, norming, normable expressions]
The set of norms $\dnrm$ is generated by the following rules
\begin{eqnarray*}
\dnrm&\!::=\!&\prim\,\mid\,[\dnrm,\dnrm]
\end{eqnarray*}
Obviously $\dnrm\subset\dexp$.
We will use the notation $\bar{a}$, $\bar{b}$, $\bar{c}$, $\ldots$ to denote norms.
The partial \emph{norming} function $\nrm{\Gamma}{a}$ defines for some expression $a$ the norm of $a$ under a context $\Gamma$.
It is defined by the equations in Table~\ref{norm.def}.
\begin{table}[!htb]
\fbox{
\begin{minipage}{0.96\textwidth}
\begin{eqnarray*}
\\[-8mm]
\nrm{\Gamma}{\prim}&=&\prim\\
\nrm{\Gamma}{x}&=&\nrm{\Gamma}{\Gamma(x)}\quad\text{if}\;\Gamma(x)\;\text{is defined}\\
\nrm{\Gamma}{\binbop{x}{a}{b}}&=&[\nrm{\Gamma}{a},\nrm{\Gamma,x:a}{b}]\\
\nrm{\Gamma}{(a\,b)}&=&\bar{c}\quad\text{if}\;\nrm{\Gamma}{a}=[\nrm{\Gamma}{b},\bar{c}]\\
\nrm{\Gamma}{\prdef{x}{a}{b}{c}}&=&[\nrm{\Gamma}{a},\nrm{\Gamma}{b}]\quad\text{if}\;\nrm{\Gamma}{b}=\nrm{\Gamma,x:a}{c}\\
\nrm{\Gamma}{\prsumop{a}{b}}&=&[\nrm{\Gamma}{a},\nrm{\Gamma}{b}]\\
\nrm{\Gamma}{\pleft{a}}&=&\bar{a}\quad\text{if}\;\nrm{\Gamma}{a}=[\bar{a},\bar{b}]\\
\nrm{\Gamma}{\pright{a}}&=&\bar{b}\quad\text{if}\;\nrm{\Gamma}{a}=[\bar{a},\bar{b}]\\
\nrm{\Gamma}{\injl{a}{b}}&=&[\nrm{\Gamma}{a},\nrm{\Gamma}{b}]\\
\nrm{\Gamma}{\injr{a}{b}}&=&[\nrm{\Gamma}{a},\nrm{\Gamma}{b}]\\
\nrm{\Gamma}{\case{a}{b}}&=&[[\bar{a},\bar{b}],\bar{c}]\quad\text{if}\;\nrm{\Gamma}{a}=[\bar{a},\bar{c}],\;\nrm{\Gamma}{b}=[\bar{b},\bar{c}]\\
\nrm{\Gamma}{\myneg a}&=&\nrm{\Gamma}{a}
\end{eqnarray*}
\end{minipage}
}
\caption{Norming\label{norm.def}}
\end{table}
The partial norming function is well-defined, in the sense that one can show by structural induction on $a$ that, if defined, $\nrm{\Gamma}{a}$ is unique.
An expression $a$ is normable relative under context $\Gamma$ iff $\nrm{\Gamma}{a}$ is defined.
This is written as $\Gamma\sngv a$.
Similarly to typing we use the notation $\Gamma\sngv a_1,\ldots,a_n$ as an abbreviation for $\Gamma\sngv a_1$, $\ldots$, $\Gamma\sngv a_n$.
\end{defi}
\begin{rem}[Examples]
There are valid and invalid normable expressions and there are strongly normalisable expression which are neither valid nor normable.
We present some examples.
We will show later that all valid expressions are normable (\ref{val.nrm}) and that
all normable expressions are strongly normalizable (\ref{nrm.ce} and~\ref{sn.valid}).
Let $\Gamma =(p,q\!:\!\prim,\;z\!:\![x\!:\!p][y\!:\!q]\prim,\;w\!:\![x\!:\!\prim]x)$.
Consider the expression $[x:p](z\,x)$:
\begin{itemize}
\item We have $\Gamma\sgv[x\!:\!p](z\,x)$.
\item We have $\Gamma\sngv[x\!:\!p](z\,x)$ since $\nrm{\Gamma}{[x\!:\!p](z\,x)}=[\prim,\nrm{\Gamma,x:p}{(z\,x)}]=[\prim,[\prim,\prim]]$. The latter equality is true since $\nrm{\Gamma,x:p}{z}=\nrm{\Gamma,x:p}{[x\!:\!p][y\!:\!q]\prim}=[\prim,[\prim,\prim]]$ and $\nrm{\Gamma,x:p}{x}=\nrm{\Gamma,x:p}{p}=\prim$.
\end{itemize}
Consider the expression $[x\!:\!p](z\,p)$:
\begin{itemize}
\item We do not have $\Gamma\sgv[x:p](z\,p)$ since we do not have $\Gamma\sgv p:p$.
\item We have $\Gamma\sngv[x:p](z\,p)$ since $\nrm{\Gamma}{[x:p](z\,p)}=[\prim,[\prim,\prim]]$ and $\nrm{\Gamma,x:p}{p}=\prim$
\end{itemize}
As a third example consider the expression $[x:[y:\prim]y](x\,x)$:
\begin{itemize}
\item Obviously $[x:[y:\prim]y](x\,x)\in\sn{}$.
\item We do not have $\;\sgv [x:[y:\prim]y](x\,x)$ since the application $(x\,x)$ cannot be typed.
\item We do not have $\;\sngv [x:[y:\prim]y](x\,x)$: The definition of norming leads us to the expression $\nrm{x:[y:\prim]y}{(x\,x)}$ which is not defined since this would require that $\nrm{x:[y:\prim]y}{x}=[\nrm{x:[y:\prim]y}{x},\bar{a}]$ for some $\bar{a}$.
Hence the norming condition for application is violated.
\end{itemize}
\end{rem}
\noindent
We show several basic properties of normable expressions culminating in the property that all valid expressions are normable.
Some of these properties and proofs are structurally similar to the corresponding ones for valid expressions.
However, due to the simplicity of norms, the proofs are much shorter.
\begin{lem}[Norm equality in context]%
\label{nrm.eq}
Let $\Gamma_a=(\Gamma_1,x:a,\Gamma_2)$ and $\Gamma_b=(\Gamma_1,x:b,\Gamma_2)$ for some $\Gamma_1,\Gamma_2,x,a,b$. For all $c$: If $\Gamma_1\sngv a,b$, $\nrm{\Gamma_1}{a}=\nrm{\Gamma_1}{b}$, and  $\Gamma_a\sngv c$ then  $\Gamma_b\sngv c$ and $\nrm{\Gamma_a}{c}=\nrm{\Gamma_b}{c}$.
\end{lem}
\begin{proof}
The straightforward proof is by structural induction on $c$.
\end{proof}
\begin{lem}[Substitution and norming]%
\label{nrm.sub}
Let $\Gamma_a=(\Gamma_1,x:a,\Gamma_2)$ and $\Gamma_b=(\Gamma_1,\Gamma_2\gsub{x}{b})$ for some $\Gamma_1,\Gamma_2,x,a,b$. For all $c$: If $\Gamma_1\sngv a,b$, $\nrm{\Gamma_1}{a}=\nrm{\Gamma_1}{b}$, and  $\Gamma_a\sngv c$ then  $\Gamma_b\sngv c\gsub{x}{b}$ and $\nrm{\Gamma_a}{c}=\nrm{\Gamma_b}{c\gsub{x}{b}}$.
\end{lem}
\begin{proof}
The proof is by structural induction on $c$ and follows from the definition of norming and of substitution.
\end{proof}
\begin{lem}[Reduction preserves norms]%
\label{nrm.rd}
For all $\Gamma,a,b$:
$\Gamma\sngv a$ and $a\rd b$ implies $\nrm{\Gamma}{a}=\nrm{\Gamma}{b}$
\end{lem}
\begin{proof}
It is obviously sufficient to show the property for single-step reduction which we do here by induction on the definition of single-step reduction.
The proof is straightforward, \eg~in case of axiom $\beta_1$, we have $a=([x:a_1]a_2\:a_3)$, $b=a_2\gsub{x}{a_3}$, and $\nrm{\Gamma}{([x:a_1]a_2\:a_3)}=\nrm{\Gamma,x:a_1}{a_2}$ where $\nrm{\Gamma}{a_1}=\nrm{\Gamma}{a_3}$.
Therefore by Lemma~\ref{nrm.sub} we know that $\nrm{\Gamma,x:a_1}{a_2}=\nrm{\Gamma}{a_2\gsub{x}{a_3}}$ which implies the proposition.

Similarly, \eg~in case of structural rule $\binbop{x}{\_}{\_}_1$ we have $a=\binbop{x}{a_1}{a_3}$, $b=\binbop{x}{a_2}{a_3}$, and $a_1\srd a_2$.
By inductive hypothesis $\nrm{\Gamma}{a_1}=\nrm{\Gamma}{a_2}$.
By Lemma~\ref{nrm.eq} we know that $\nrm{\Gamma,x:a_1}{a_3}\!=\nrm{\Gamma,x:a_2}{a_3}$. Therefore $\nrm{\Gamma}{a}=[\nrm{\Gamma}{a_1},\nrm{\Gamma,x:a_1}{a_3}]=[\nrm{\Gamma}{a_2},\nrm{\Gamma,x:a_2}{a_3}]=\nrm{\Gamma}{b}$.
\end{proof}
\begin{lem}[Context extension]%
\label{norm.ext}
For all $\Gamma_1,\Gamma_2,x,a,b$ where $x\notin\dom(\Gamma_1,\Gamma_2)$:
$\Gamma_1,\Gamma_2\sngv a$ implies $(\Gamma_1,x:b,\Gamma_2)\sngv a$ and $\nrm{\Gamma_1,\Gamma_2}{a}=\nrm{\Gamma_1,x:b,\Gamma_2}{a}$.
\end{lem}
\begin{proof}
Obviously $\Gamma_1,\Gamma_2\sngv a$ implies that $\free([\Gamma_1,\Gamma_2]a)=\emptyset$.
Therefore the additional declaration $x:b$ will never be used when evaluating $\nrm{\Gamma_1,x:b,\Gamma_2}{a}$.
Hence the successful evaluation of $\nrm{\Gamma_1,\Gamma_2}{a}$ can be easily transformed into an evaluation of $\nrm{\Gamma_1,x:b,\Gamma_2}{a}$ with identical result.
\end{proof}
\begin{lem}[Typing implies normability and preserves norm]%
\label{val.nrm}
For all $\Gamma,a,b$:
If $\Gamma\sgv a:b$ then $\sngv[\Gamma]a,[\Gamma]b$ and $\nrm{\Gamma}{a}=\nrm{\Gamma}{b}$.
As a consequence $\Gamma\sgv a$ implies $\Gamma\sngv a$.
\end{lem}
\begin{proof}
Proof by induction on the definition of $\Gamma\sgv a:b$.
Note that $\sngv[\Gamma]a$ implies $\Gamma\sngv a$ but not vice-versa.
The stronger conclusion $\sngv[\Gamma]a$ is needed \eg~in the rule \absu.
We consider some interesting cases in detail:
\begin{itemize}
\item Rule \ax:
Obvious, since $a=b=\prim$ and $\Gamma=()$.
\item Rule \mystart:
We have $a=x$, $\Gamma=(\Gamma',x:b)$ for some $\Gamma'$ and $x$, and $\Gamma'\sgv b:c$ for some $c$.
By inductive hypothesis $\sngv[\Gamma']b$ and hence obviously $\Gamma'\sngv b$.
Since $x\notin\free(b)$, by Law~\ref{norm.ext} we know that  $\Gamma\sngv b$ and $\nrm{\Gamma}{b}=\nrm{\Gamma'}{b}$.
Hence, by definition of norming $\sngv[\Gamma]x$, where $\nrm{\Gamma}{x}=\nrm{\Gamma}{b}$.
\item Rule \weak:
We have $\Gamma=(\Gamma',x:c)$ for some $\Gamma'$, $x$, and $c$ where $\Gamma'\sgv c$ and $\Gamma'\sgv a:b$.
By inductive hypothesis $\sngv[\Gamma']a,[\Gamma']b,[\Gamma']c$ (obviously this implies $\Gamma'\sngv a,b,c$) and $\nrm{\Gamma'}{a}=\nrm{\Gamma'}{b}$.
Since $x\notin\free(a)\cup\free(b)$, by Law~\ref{norm.ext} we know that $\Gamma\sngv a$ and $\nrm{\Gamma}{a}=\nrm{\Gamma'}{a}$ as well as
$\Gamma\sngv b$ and $\nrm{\Gamma}{b}=\nrm{\Gamma'}{b}$.
Hence by definition of norming $\sngv[\Gamma]a$,$[\Gamma]b$ and $\nrm{\Gamma}{a}=\nrm{\Gamma}{b}$.
\item Rule \myconv:
We have $\Gamma\sgv a:b$ where $\Gamma\sgv a:c$, $c\eqv b$, and $\Gamma\sgv b:d$ for some $c$ and $d$ (note that we are here using the rule \myconv~with $b$ and $c$ exchanged).
By inductive hypothesis  $\sngv[\Gamma]a,[\Gamma]b,[\Gamma]c$ (obviously this implies $\Gamma\sngv a,b,c$) and $\nrm{\Gamma}{a}=\nrm{\Gamma}{c}$.
Hence by Laws~\ref{rd.confl} and~\ref{nrm.rd} we know that $\nrm{\Gamma}{c}=\nrm{\Gamma}{b}$ which implies $\nrm{\Gamma}{a}=\nrm{\Gamma}{b}$.
\item Rule \absu:
We have $a=[x:c]a_1$ and $b=[x:c]a_2$, for some $c$, $a_1$, and $a_2$ where $(\Gamma,x:c)\sgv a_1:a_2$.
By inductive hypothesis we know that $\sngv[\Gamma,x:c]a_1,[\Gamma,x:c]a_2$ (which can also be written as $\sngv[\Gamma][x:c]a_1$,$[\Gamma][x:c]a_2$) and $\nrm{\Gamma,x:c}{a_1}=\nrm{\Gamma,x:c}{a_2}$.
Furthermore, \eg~$\sngv[\Gamma][x:c]a_1$ implies $\Gamma\sngv c$ and therefore
by definition of norming $\nrm{\Gamma}{a}=[\nrm{\Gamma}{c},\nrm{\Gamma,x:c}{a_1}]=[\nrm{\Gamma}{c},\nrm{\Gamma,x:c}{a_2}]=\nrm{\Gamma}{b}$.
\item Rule \appl:
We have $a=(a_1\,a_2)$ and $b=c_2\gsub{x}{a_2}$, for some $a_1$, $a_2$, $x$, and $c_2$ where $\Gamma\sgv a_1:[x:c_1]c_2$ and $\Gamma\sgv a_2:c_1$ for some $c_1$.
By inductive hypothesis we know that $\sngv[\Gamma]a_1,[\Gamma]a_2,[\Gamma][x:c_1]c_2,[\Gamma]c_1$ (obviously this implies $\Gamma\sngv a_1,a_2,[x:c_1]c_2,c_1$) and
$\nrm{\Gamma}{[x:c_1]c_2}=\nrm{\Gamma}{a_1}$ as well as $\nrm{\Gamma}{c_1}=\nrm{\Gamma}{a_2}$.
Hence $\nrm{\Gamma}{a_1}=[\nrm{\Gamma}{a_2},\nrm{\Gamma,x:c_1}{c_2}]$ and therefore by definition of norming $\nrm{\Gamma}{(a_1\,a_2)}=\nrm{\Gamma,x:c_1}{c_2}$ which implies $\Gamma\sngv a$ and (obviously) $\sngv[\Gamma]a$.
Since $\nrm{\Gamma}{a_2}=\nrm{\Gamma}{c_1}$ we can apply Law~\ref{nrm.sub} to obtain $\nrm{\Gamma,x:c_1}{c_2}=\nrm{\Gamma}{c_2\gsub{x}{a_2}}$.
Hence $\Gamma\sngv b$ and (obviously) $\sngv[\Gamma]b$ and $\nrm{\Gamma}{a}=\nrm{\Gamma}{(a_1\,a_2)}=\nrm{\Gamma}{c_2\gsub{x}{a_2}}=\nrm{\Gamma}{b}$.
\end{itemize}
For the consequence, $\Gamma\sgv a$ means that $\Gamma\sgv a:b$, for some $b$.
By the property just shown this implies $\Gamma\sngv a$.
\end{proof}
\noindent
We introduce a norm-based induction principle that we will use several times in the following proofs.
\begin{defi}[Induction on the size of norms]
The \emph{size} of a norm $\bar{a}$ is defined as the number of primitive constants $\prim$ it contains.
A property $P(\Gamma,\bar{a})$ is shown by \emph{norm-induction} iff for all $\bar{b}$ we know that:
If $P(\Gamma,\bar{c})$ for all $\bar{c}$ of size strictly smaller that $\bar{b}$ then $P(\Gamma,\bar{b})$.
This can be reformulated into a more convenient form for its use in proofs.
\begin{itemize}
\item {\bf Inductive base:}
$P(\Gamma,\prim)$.
\item {\bf Inductive step:}
For all $\bar{b},\bar{c}$:
If $P(\Gamma,\bar{a})$ for all $\bar{a}$ of size strictly smaller than the size of $[\bar{b},\bar{c}]$ then $P(\Gamma,[\bar{b},\bar{c}])$.
\end{itemize}
\end{defi}
\noindent
Computable expressions are organized according to norm structure and satisfy a number of additional conditions necessary for an inductive proof of strong normalization.
\begin{defi}[Computable expressions]%
\label{ce}
The set of \emph{computable expressions} of norm $\bar{a}$ under context $\Gamma$ is denoted by $\ce_{\Gamma}(\bar{a})$.
$a\in\ce_{\Gamma}(\bar{a})$ iff $a\in\sn{}$, $\Gamma\sngv a$ where $\nrm{\Gamma}{a}=\bar{a}$, and if $\bar{a}=[\bar{b},\bar{c}]$ for some $\bar{b}$ and $\bar{c}$ then
the following \emph{computability conditions} are satisfied:
\begin{itemize}[align=left]
\item[($\alpha$)]
For all $x$, $b$, $c$: If $a\rd\binbop{x}{b}{c}$ or $\myneg a\rd\binbop{x}{b}{c}$ then $c\in\ce_{\Gamma,x:b}(\bar{c})$ and $c\gsub{x}{d}\in\ce_{\Gamma}(\bar{c})$ for any $d\in\ce_{\Gamma}(\bar{b})$.
\item[($\beta$)]
For all $b$, $c$: If $a\rd\prsumop{b}{c}$ or $\myneg a\rd\prsumop{b}{c}$ then $b\in\ce_{\Gamma}(\bar{b})$ and $c\in\ce_{\Gamma}(\bar{c})$.
\item[($\gamma$)]
For all $x$, $b$, $c$, $d$: $a\rd\prdef{x}{b}{c}{d}$ implies both $b\in\ce_{\Gamma}(\bar{b})$ and $c\in\ce_{\Gamma}(\bar{c})$,
$a\rd\injl{b}{d}$ implies $b\in\ce_{\Gamma}(\bar{b})$, and $a\rd\injr{d}{c}$ implies $c\in\ce_{\Gamma}(\bar{c})$.
\item[($\delta$)]
If $\bar{b}=[\bar{b}_1,\bar{b}_2]$, for some $\bar{b}_1$ and $\bar{b}_2$, then, for all $b_1$, $b_2$:
$a\rd\case{b_1}{b_2}$ implies both $b_1\in\ce_{\Gamma}([\bar{b}_1,\bar{c}])$ and $b_2\in\ce_{\Gamma}([\bar{b}_2,\bar{c}])$.
\end{itemize}
\end{defi}
\begin{rem}[Motivation for the computability conditions]
The four conditions are motivated by the strong normalization condition for applications (Lemma~\ref{sn.cond}(1)) and negations (Lemma~\ref{sn.cond}(2)). % chktex 36
\end{rem}
\begin{lem}[Computable expressions are well-defined for all norms]
For all $\Gamma$ and $\bar{a}$, the set $\ce_{\Gamma}(\bar{a})$ exists and is well-defined.
\end{lem}
\begin{proof}
Proof by norm-induction on $\bar{a}$ follows directly from the definition of computable expressions.
\end{proof}
\noindent
We begin with some basic properties of computable expressions.
\begin{lem}[Basic properties of computable expressions]%
\label{ce.basic}
For all $\Gamma,\Gamma_1,\Gamma_2,a,a_1,a_2,b,x$:
\begin{enumerate}
\item
$a\in\ce_{\Gamma}(\nrm{\Gamma}{a})$ and $a\rd b$ imply $b\in\ce_{\Gamma}(\nrm{\Gamma}{a})$.
\item
$\Gamma_1\sngv a_1$, $a\in\ce_{\Gamma_1,x:a_1,\Gamma_2}(\bar{a})$, and $a_1\rd a_2$ imply $a\in\ce_{\Gamma_1,x:a_2,\Gamma_2}(\bar{a})$.
\end{enumerate}
\end{lem}
\begin{proof}
For (1), we assume that $a\in\ce_{\Gamma}(\bar{a})$ where $\nrm{\Gamma}{a}=\bar{a}$.
Obviously $a\in\sn{}$ and therefore also $b\in\sn{}$ and by Lemma~\ref{nrm.rd} we obtain $\nrm{\Gamma}{b}=\nrm{\Gamma}{a}$.
We have to show that $b\in\ce_{\Gamma}(\bar{a})$:
It is easy to prove the computability conditions, since from $b\rd c$ we can always infer $a\rd c$ and hence use the corresponding condition from the assumption $a\in\ce_{\Gamma}(\bar{a})$.
%\item[$iv$:]

For (2) the proof is by norm-induction on $\bar{a}$ and only requires properties~\ref{nrm.eq} and~\ref{nrm.rd} of norming.
%
%\end{itemize}
\end{proof}
\noindent
The closure of computable expressions against negation is shown first due to its necessity in other monotonicity arguments.
\begin{lem}[Computable expressions are closed against negation]%
\label{ce.mon.neg}
For all $\Gamma, a$:
$a\in\ce_{\Gamma}(\nrm{\Gamma}{a})$ implies $\myneg a\in\ce_{\Gamma}(\nrm{\Gamma}{a})$.
\end{lem}
\begin{proof}
Assume that $a\in\ce_{\Gamma}(\nrm{\Gamma}{a})$.
Let $\bar{a}=\nrm{\Gamma}{a}$.
By definition of norming obviously $\nrm{\Gamma}{\myneg a}=\bar{a}$.
We show that $a\in\ce_{\Gamma}(\bar{a})$ implies $\myneg a\in\ce_{\Gamma}(\bar{a})$ by norm-induction on $\bar{a}$.

{\bf Inductive base:} We have $\bar{a}=\prim$, therefore the computability conditions become trivial and it remains to show
that $\myneg a\in\sn{}$.
Since $a\in\sn{}$, we can apply Lemma~\ref{sn.cond}(2) whose conditions $(C_1)$ and $(C_2)$ become trivial since by definition of norming and Lemma~\ref{nrm.rd} they imply that $\bar{a}\neq\prim$. % chktex 36

{\bf Inductive step:} Let $\bar{a}=[\bar{b},\bar{c}]$ for some $\bar{b}$ and $\bar{c}$.
First, we show that $\myneg a\in\sn{}$.
Since $a\in\sn{}$, according to Lemma~\ref{sn.cond}(2), for any $x,b,c$, we need to show conditions $(C_1)$ and $(C_2)$: % chktex 36
\begin{itemize}
\item $(C_1)$:
Let $a\rd\prsumop{b}{c}$.
Since $a\in\ce_{\Gamma}(\bar{a})$, by computability condition $\beta$ we know that $b\in\ce_{\Gamma}(\bar{b})$ and $c\in\ce_{\Gamma}(\bar{c})$. By inductive hypothesis $\myneg b\in\ce_{\Gamma}(\bar{b})$ and $\myneg c\in\ce_{\Gamma}(\bar{c})$ and therefore obviously $\myneg b,\myneg c\in\sn{}$.
\item $(C_2)$:
Let $a\rd\binbop{x}{b}{c}$.
Since $a\in\ce_{\Gamma}(\bar{a})$, by computability condition $\alpha$  we know that $c\in\ce_{\Gamma,x:b}(\bar{c})$.
By inductive hypothesis $\myneg c\in\ce_{\Gamma,x:b}(\bar{c})$ and therefore obviously $\myneg c\in\sn{}$.
\end{itemize}
Therefore we know that $\myneg a\in\sn{}$.
It remains to show the computability conditions for $\myneg a$:
\begin{itemize}[align=left]
\item[($\alpha$)]
We have to consider two cases:
\begin{enumerate}
\item
$\myneg a\rd[x:a_1]a_2$ or $\myneg a\rd[x!a_1]a_2$ for some $x$, $a_1$, and $a_2$.
By Lemma~\ref{rd.decomp.neg}(1,3) we know that $a\rd\binbop{x}{a_1'}{a_2'}$ for some $a_1'$ and $a_2'$ where $a_1'\rd a_1$ and $\myneg a_2'\rd a_2$. % chktex 36
The first part of $\alpha$ can be argued as follows:
\begin{eqnarray*}
&&a\in\ce_{\Gamma}([\bar{b},\bar{c}])\\
&\Rightarrow&\quad\text{(by $\alpha$, since $a\rd\binbop{x}{a_1'}{a_2'}$)}\\
&&a_2'\in\ce_{\Gamma,x:a_1'}(\bar{c})\\
&\Rightarrow&\quad\text{(inductive hypothesis)}\\
&&\myneg a_2'\in\ce_{\Gamma,x:a_1'}(\bar{c})\\
&\Rightarrow&\quad\text{(by Lemma~\ref{ce.basic}(1,2), since $a_1'\rd a_1$ and $\myneg a_2'\rd a_2$)}\\ % chktex 36
&&a_2\in\ce_{\Gamma,x:a_1}(\bar{c})
\end{eqnarray*}
For the second clause, for any $d\in\ce_{\Gamma}(\bar{b})$, we can argue as follows:
\begin{eqnarray*}
&&a\in\ce_{\Gamma}([\bar{b},\bar{c}])\\
&\Rightarrow&\quad\text{(by $\alpha$, since $a\rd\binbop{x}{a_1'}{a_2'}$)}\\
&&a_2'\gsub{x}{d}\in\ce_{\Gamma}(\bar{c})\\
&\Rightarrow&\quad\text{(inductive hypothesis)}\\
&&\myneg(a_2'\gsub{x}{d})\in\ce_{\Gamma}(\bar{c})\\
&\Rightarrow&\quad\text{(definition of substitution)}\\
&&(\myneg a_2')\gsub{x}{d}\in\ce_{\Gamma}(\bar{c})\\
&\Rightarrow&\quad\text{(by Lemmas~\ref{ce.basic}(1) and~\ref{sub.rd}, since $\myneg a_2'\rd a_2$)}\\ % chktex 36
&&a_2\gsub{x}{d}\in\ce_{\Gamma}(\bar{c})
\end{eqnarray*}
\item
$\myneg\myneg a\rd[x:a_1]a_2$ or $\myneg\myneg a\rd[x!a_1]a_2$ for some $x$, $a_1$, and $a_2$.
By Lemma~\ref{rd.decomp.neg}(1,3) we know that $\myneg a\rd\binbop{x}{a_1'}{a_2'}$ for some $a_1'$ and $a_2'$ where  $a_1'\rd a_1$ and $\myneg a_2'\rd a_2$. % chktex 36
Applying Lemma~\ref{rd.decomp.neg}(1,3) again we know that $a\rd\binbopd{x}{a_1''}{a_2''}$ for some $a_1''$ and $a_2''$ where $a_1''\rd a_1'$ and $\myneg a_2''\rd a_2'$. This means that $a_1''\rd a_1$ and $\myneg\myneg a_2''\rd a_2$. % chktex 36
This case can then be argued similarly to the first case just applying the inductive hypothesis twice.
\end{enumerate}
\item[($\beta$)]
Similarly to $\alpha$, we have to consider two cases:
\begin{enumerate}
\item
$\myneg a\rd[a_1,a_2]$ or $\myneg a\rd[a_1+a_2]$ for some $a_1$ and $a_2$.
By Lemma~\ref{rd.decomp.neg}(2,4) we know that $a\rd\prsumop{a_1'}{a_2'}$ for some $a_1'$ and $a_2'$ where $\myneg a_1'\rd a_1$ and $\myneg a_2'\rd a_2$. % chktex 36

Since $a\in\ce_{\Gamma}(\bar{a})$ and $a\rd\prsumop{a_1'}{a_2'}$ by computability condition $\beta$ we obtain $a_1'\in\ce_{\Gamma}(\bar{b})$ and $a_2'\in\ce_{\Gamma}(\bar{c})$.
By inductive hypotheses we know that also $\myneg a_1'\in\ce_{\Gamma}(\bar{b})$ and $\myneg a_2'\in\ce_{\Gamma}(\bar{c})$.
By Lemma~\ref{ce.basic}(1) we get $a_1\in\ce_{\Gamma}(\bar{b})$ and $a_2\in\ce_{\Gamma}(\bar{c})$. % chktex 36
\item
$\myneg\myneg a\rd[a_1,a_2]$ or $\myneg\myneg a\rd[a_1+a_2]$ for some $a_1$ and $a_2$.
By Lemma~\ref{rd.decomp.neg}(2,4) we know that $\myneg a\rd\prsumop{a_1'}{a_2'}$ for some $a_1'$ and $a_2'$ where $\myneg a_1'\rd a_1$ and $\myneg a_2'\rd a_2$. % chktex 36
Applying Lemma~\ref{rd.decomp.neg}(2,4) again we know that $a\rd\prsumopd{a_1''}{a_2''}$ for some $a_1''$ and $a_2''$ where $\myneg a_1''\rd a_1'$ and $\myneg a_2''\rd a_2'$. This means that $\myneg\myneg a_1''\rd a_1$ and $\myneg\myneg a_2''\rd a_2$. % chktex 36

Since $a\in\ce_{\Gamma}(\bar{a})$ and $a\rd\prsumopd{a_1''}{a_2''}$ by computability condition $\beta$ we obtain $a_1''\in\ce_{\Gamma}(\bar{b})$ anf $a_2''\in\ce_{\Gamma}(\bar{c})$.
By inductive hypotheses (applied twice)  we know that also $\myneg\myneg a_1'\in\ce_{\Gamma}(\bar{b})$ and $\myneg\myneg a_2'\in\ce_{\Gamma}(\bar{c})$.
By Lemma~\ref{ce.basic}(1) we get $a_1\in\ce_{\Gamma}(\bar{b})$ and $a_2\in\ce_{\Gamma}(\bar{c})$. % chktex 36
\end{enumerate}
\item[($\gamma$)]
We have three cases:
First, if  $\myneg a\rd\prdef{x}{a_1}{a_2}{a_3}$ for some $x$, $a_1$, $a_2$, $a_3$, then
by Lemma~\ref{rd.decomp.neg}(5) we know that $a\rd\prdef{x}{a_1}{a_2}{a_3}$. % chktex 36
Since $a\in\ce_{\Gamma}([\bar{b},\bar{c}])$, by computability condition $\gamma$, we know that $a_1\in\ce_{\Gamma}(\bar{b})$, $a_2\in\ce_{\Gamma}(\bar{c})$.
Second, if $\myneg a\rd\injl{a_1}{a_2}$for some $a_1$ and $a_2$, then
by Lemma~\ref{rd.decomp.neg}(7) we know that $a\rd\injl{a_1}{a_2}$. % chktex 36
Since $a\in\ce_{\Gamma}([\bar{b},\bar{c}])$, by computability condition $\gamma$, we know that $a_1\in\ce_{\Gamma}(\bar{b})$.
The third case $\myneg a\rd\injr{a_1}{a_2}$ is shown in a similar way.
\item[($\delta$)]
Let $\bar{b}=[\bar{b}_1,\bar{b}_2]$ for some $\bar{b}_1$ and $\bar{b}_2$.
Let $\myneg a\rd\case{a_1}{a_2}$ for some $a_1$ and $a_2$.
By Lemma~\ref{rd.decomp.neg}(6) we know that $a\rd\case{a_1}{a_2}$. % chktex 36
By computability condition $\delta$ we know that $a_1\in\ce_{\Gamma}([\bar{b}_1,\bar{c}])$ and $a_2\in\ce_{\Gamma}([\bar{b}_2,\bar{c}])$.
\qedhere
\end{itemize}
\end{proof}
\noindent
Closure of computability against application has a proof that is not very difficult but somewhat lengthy due to a repetition of similar arguments for the different computability conditions.
\begin{lem}[Closure of computable expressions against application]%
\label{ce.mon.appl}
For all $\Gamma$, $a$, $b$:
$\Gamma\sngv (a\,b)$, $a\in\ce_{\Gamma}(\nrm{\Gamma}{a})$, and $b\in\ce_{\Gamma}(\nrm{\Gamma}{b})$ implies $\nrm{\Gamma}{a}=[\nrm{\Gamma}{b},\bar{c}]$ and $(a\,b)\in\ce_{\Gamma}(\bar{c})$ for some $\bar{c}$.
\end{lem}
\begin{proof}
Assume that $\Gamma\sngv (a\,b)$, $a\in\ce_{\Gamma}(\nrm{\Gamma}{a})$, and $b\in\ce_{\Gamma}(\nrm{\Gamma}{b})$.
Let $\bar{a}=\nrm{\Gamma}{a}$ and $\bar{b}=\nrm{\Gamma}{b}$.
$\Gamma\sngv (a\,b)$ implies that $\bar{a}=[\bar{b},\bar{c}]$ for some $\bar{c}$.
By norm-induction on $\bar{a}$ we will show that $a\in\ce_{\Gamma}(\bar{a})$ and $b\in\ce_{\Gamma}(\bar{b})$ implies $(a\,b)\in\ce_{\Gamma}(\bar{c})$.

The inductive base is trivial since $\bar{a}\neq\prim$.
For the inductive step we first need to show that $(a\,b)\in\sn{}$.
By Lemma~\ref{sn.cond}(1) we have to show conditions $(C_1)$, $(C_2)$, and $(C_3)$. For any $x$, $b_1$, $c_1$, $c_2$, $d_1$, $d_2$: % chktex 36
\begin{itemize}
\item $(C_1)$:
Let $a\rd\binbop{x}{b_1}{c_1}$.
Since $a\in\ce_{\Gamma}([\bar{b},\bar{c}])$, by computability condition $\alpha$, for any $d\in\ce_{\Gamma}(\bar{b})$ we know that $c_1\gsub{x}{d}\in\ce_{\Gamma}(\bar{c})$.
Hence also $c_1\gsub{x}{b}\in\ce_{\Gamma}(\bar{c})$ and therefore obviously $c_1\gsub{x}{b}\in\sn{}$ hence condition $(C_1)$ is satisfied.
\item $(C_2)$:
Let $a\rd\case{c_1}{c_2}$ and $b\rd\injl{d_1}{d_2}$.
From $b\rd\injl{d_1}{d_2}$, by Lemma~\ref{nrm.rd} and by definition of norming we know that $\bar{b}=[\bar{d}_1,\bar{d}_2]$, for some $\bar{d}_1$, $\bar{d}_2$, and hence $\bar{a}=[[\bar{d}_1,\bar{d}_2],\bar{c}]$.
Since $b\in\ce_{\Gamma}([\bar{d}_1,\bar{d}_2])$ by computability condition $\gamma$ we know that $d_1\in\ce_{\Gamma}(\bar{d}_1)$.
Since $a\in\ce_{\Gamma}([[\bar{d}_1,\bar{d}_2],\bar{c}])$ by computability condition $\delta$ we know that
$c_1\in\ce_{\Gamma}([\bar{d}_1,\bar{c}])$, $c_2\in\ce_{\Gamma}([\bar{d}_2,\bar{c}])$.

By inductive hypothesis (the size of $[\bar{d}_1,\bar{c}]$ is strictly smaller than that of $\bar{a}$), we know that $(c_1\,d_1)\in\ce_{\Gamma}(\bar{c})$.
By definition of computability therefore $(c_1\,d_1)\in\sn{}$.
\item $(C_3)$: Proof is similar to $(C_2)$.
\end{itemize}
Therefore by Lemma~\ref{sn.cond}(1) we know that $(a\,b)\in\sn{}$ % chktex 36
It remains to show the computability conditions for $(a\,b)$.
Let $\bar{c}=[\bar{d},\bar{e}]$ for some $\bar{d}$ and $\bar{e}$:
\begin{itemize}[align=left]
\item[($\alpha$)]
If $(a\,b)\rd\binbop{x}{a_1}{a_2}$ or $\myneg(a\,b)\rd\binbop{x}{a_1}{a_2}$, for some $a_1$ and $a_2$ then, since $\Gamma\sngv(a\,b)$, by Lemma~\ref{nrm.rd} we have
\[\nrm{\Gamma}{\binbop{x}{a_1}{a_2}}=[\nrm{\Gamma}{a_1},\nrm{\Gamma,x:a_1}{a_2}]=[\bar{d},\bar{e}]
\]
Therefore $\nrm{\Gamma}{a_1}=\bar{d}$ and $\nrm{\Gamma,x:a_1}{a_2}=\bar{e}$.
We have to show that $a_2\in\ce_{\Gamma,x:a_1}(\bar{e})$ and that $a_2\gsub{x}{d}\in\ce_{\Gamma}(\bar{e})$ for any $d\in\ce_{\Gamma}(\bar{d})$.
We have to distinguish two cases:
\begin{enumerate}
\item
If $(a\,b)\rd\binbop{x}{a_1}{a_2}$ then by Lemma~\ref{rd.decomp}(5) we know that there are two subcases % chktex 36
\begin{enumerate}
\item
$a\rd[y:a_3]a_4$, $b\rd b'$, and $a_4\gsub{y}{b'}\rd\binbop{x}{a_1}{a_2}$ for some $y$, $a_3$, $a_4$, and $b'$.
We can argue as follows:
\begin{eqnarray*}
&&a\in\ce_{\Gamma}([\bar{b},\bar{c}])\\
&\Rightarrow&\text{(by $\alpha$, since by Lemma~\ref{ce.basic}(1) we know that $b'\in\ce_{\Gamma}(\bar{b})$)}\\ % chktex 36
&&a_4\gsub{x}{b'}\in\ce_{\Gamma}(\bar{c})\\
&\Rightarrow&\text{(by Lemma~\ref{ce.basic}(1), since $a_4\gsub{y}{b'}\rd\binbop{x}{a_1}{a_2}$)}\\ % chktex 36
&&\binbop{x}{a_1}{a_2}\in\ce_{\Gamma}(\bar{c})
\end{eqnarray*}
\item
$a\rd\case{c_1}{c_2}$ and either $b\rd\injl{b_1}{b_2}$ and $(c_1\,b_1)\rd\binbop{x}{a_1}{a_2}$ or $b\rd\injr{b_1}{b_2}$ and $(c_2\,b_2)\rd\binbop{x}{a_1}{a_2}$ for some $c_1$, $c_2$, $b_1$, and $b_2$.
This means that $\bar{b}=[\bar{b}_1,\bar{b}_2]$ for some $\bar{b}_1$ and $\bar{b}_2$ where $\nrm{\Gamma}{c_1}=[\bar{b}_1,\bar{c}]$ and $\nrm{\Gamma}{c_2}=[\bar{b}_2,\bar{c}]$.
Since $\Gamma\sngv (a\:b)$ we also know that $\nrm{\Gamma}{b_1}=\bar{b}_1$ and  $\nrm{\Gamma}{b_2}=\bar{b}_2$.
We will show the first case $b\rd\injl{b_1}{b_2}$, the proof of the second case is similar.

From $b\rd\injl{b_1}{b_2}$, by computability condition $\gamma$ we know that $b_1\in\ce_{\Gamma}(\bar{b}_1)$.
From $a\rd\case{c_1}{c_2}$, by computability condition $\delta$ we know that $c_1\in\ce_{\Gamma}([\bar{b}_1,\bar{c}])$ and $c_2\in\ce_{\Gamma}([\bar{b}_2,\bar{c}])$.
Since the size of $[\bar{b}_1,\bar{c}]$ is strictly smaller than the size of $\bar{a}$ we can apply the inductive hypothesis to obtain
$(c_1\,b_1)\in\ce_{\Gamma}(\bar{c})$.
Hence, since $(c_1\,b_1)\rd\binbop{x}{a_1}{a_2}$, by Lemma~\ref{ce.basic}(1) we have $\binbop{x}{a_1}{a_2}\in\ce_{\Gamma}(\bar{c})$. % chktex 36
\end{enumerate}
Therefore in both cases we have shown $\binbop{x}{a_1}{a_2}\in\ce_{\Gamma}(\bar{c})$.
From computability condition $\alpha$ we obtain $a_2\in\ce_{\Gamma,x:a_1}(\bar{e})$ and $a_2\gsub{x}{d}\in\ce_{\Gamma}(\bar{e})$.
\item If $\myneg(a\,b)\rd\binbop{x}{a_1}{a_2}$, then we show the case $\binbop{x}{a_1}{a_2}=[x:a_1]a_2$ (the case $\binbop{x}{a_1}{a_2}=[x!a_1]a_2$ runs analogously):
By Lemma~\ref{rd.decomp.neg}(1) we know that $(a\,b)\rd[x!a_1']a_2'$ for some $a_1'$ and $a_2'$ where $a_1'\rd a_1$ and $\myneg a_2'\rd a_2$. % chktex 36
By reasoning similarly to the first case we can show that $[x!a_1']a_2'\in\ce_{\Gamma}(\bar{c})$.

The first part of $\alpha$ can be seen as follows:
Since $[x!a_1']a_2'\in\ce_{\Gamma}(\bar{c})$ we know that $a_2'\in\ce_{\Gamma,x:a_1}(\bar{e})$.
By Lemma~\ref{ce.mon.neg}, this implies $\myneg a_2'\in\ce_{\Gamma,x:a_1}(\bar{e})$.
Therefore by Lemma~\ref{ce.basic}(1), since we have $\myneg a_2'\rd a_2$, we know that $a_2\in\ce_{\Gamma,x:a_1}(\bar{e})$. % chktex 36
%\begin{eqnarray*}
%&&[x!a_1']a_2'\in\ce_{\Gamma}(\bar{c})\\
%&\Rightarrow&\text{(by $\alpha$)}\\
%&&a_2'\in\ce_{\Gamma,x:a_1}(\bar{e})\\
%&\Rightarrow&\text{(by Lemma~\ref{ce.mon.neg})}\\
%&&\myneg a_2'\in\ce_{\Gamma,x:a_1}(\bar{e})\\
%&\Rightarrow&\text{(by Lemma~\ref{ce.basic}(1), since $\myneg a_2'\rd a_2$)}\\
%&&a_2\in\ce_{\Gamma,x:a_1}(\bar{e})
%\end{eqnarray*}

Let $d\in\ce_{\Gamma}(\bar{d})$. We can show the second part of $\alpha$ as follows:
\begin{eqnarray*}
&&[x!a_1']a_2'\in\ce_{\Gamma}(\bar{c})\\
&\Rightarrow&\text{(by $\alpha$, since $d\in\ce_{\Gamma}(\bar{d})$)}\\
&&a_2'\gsub{x}{d}\in\ce_{\Gamma}(\bar{e})\\
&\Rightarrow&\text{(by Lemma~\ref{ce.mon.neg})}\\
&&\myneg(a_2'\gsub{x}{d})\in\ce_{\Gamma,x:a_1}(\bar{e})\\
&\Rightarrow&\text{(definition of substitution)}\\
&&(\myneg a_2')\gsub{x}{d}\in\ce_{\Gamma,x:a_1}(\bar{e})\\
&\Rightarrow&\text{(by Lemma~\ref{ce.basic}(1) and~\ref{sub.rd}, since $\myneg a_2'\rd a_2$)}\\ % chktex 36
&&a_2\gsub{x}{d}\in\ce_{\Gamma,x:a_1}(\bar{e})
\end{eqnarray*}
\end{enumerate}
\item[($\beta$)]
We have to distinguish two cases:
\begin{enumerate}
\item
If $(a\,b)\rd\prsumop{a_1}{a_2}$ for some $a_1$ and $a_2$, then, since $\Gamma\sngv (a\,b)$, by Lemma~\ref{nrm.rd} we have
\[\nrm{\Gamma}{(a\,b)}=[\nrm{\Gamma}{a_1},\nrm{\Gamma}{a_2}]=[\bar{d},\bar{e}]
\]
Similarly to the corresponding case for condition $\alpha$ we can show that $\prsumop{a_1}{a_2}\in\ce_{\Gamma}(\bar{c})$.
By computability condition $\beta$ we obtain $a_1\in\ce_{\Gamma}(\bar{d})$ and $a_2\in\ce_{\Gamma}(\bar{e})$.
\item
If $\myneg(a\,b)\rd\prsumop{a_1}{a_2}$ for some $a_1$ and $a_2$, then, since obviously $\Gamma\sngv\myneg(a\:b)$, by Lemma~\ref{nrm.rd} we have
\[\nrm{\Gamma}{\myneg(a\,b)}=[\nrm{\Gamma}{a_1},\nrm{\Gamma}{a_2}]=[\bar{d},\bar{e}]
\]
We show the case $\prsumop{a_1}{a_2}=[a_1,a_2]$ (the case $\prsumop{a_1}{a_2}=[a_1+a_2]$ runs analogously):
By Lemma~\ref{rd.decomp.neg}(2) we know that $(a\,b)\rd[a_1'+a_2']$ where $\myneg a_1'\rd a_1$ and $\myneg a_2'\rd a_2$ for some $a_1'$ and $a_2'$. % chktex 36
Therefore by elementary properties of reduction (Lemma~\ref{rd.decomp}(6)) there are two cases % chktex 36
\begin{enumerate}
\item
$a\rd\binbop{x}{a_3}{a_4}$, $b\rd b'$, and $\myneg a_4\gsub{x}{b'}\rd[a_1'+a_2']$ for some $x$, $a_3$, $a_4$, and $b'$.
By the same argument as in the corresponding case for condition $\alpha$ we can show that $[a_1'+a_2']\in\ce_{\Gamma}(\bar{c})$.
\item
$a\rd\case{c_1}{c_2}$ and either $b\rd\injl{b_1}{b_2}$ and $(c_1\,b_1)\rd[a_1'+a_2']$ or $b\rd\injr{b_1}{b_2}$ and $(c_2\,b_2)\rd[a_1'+a_2']$ for some $c_1$, $c_2$, $b_1$, and $b_2$.
By the same argument as  in the corresponding case for condition $\alpha$ we can show that $[a_1'+a_2']\in\ce_{\Gamma}(\bar{c})$.
\end{enumerate}
Hence in all cases we have $[a_1'+a_2']\in\ce_{\Gamma}(\bar{c})$.
By Lemma~\ref{ce.mon.neg} we have $\myneg[a_1'+a_2']\in\ce_{\Gamma}(\bar{c})$ and hence by Lemma~\ref{ce.basic}(1) we have $[\myneg a_1',\myneg a_2']\in\ce_{\Gamma}(\bar{c})$ and therefore by definition of computable expressions (condition $\beta$) we obtain $\myneg a_1'\in\ce_{\Gamma}(\bar{d})$ and $\myneg a_2'\in\ce_{\Gamma}(\bar{e})$. % chktex 36
Since $\myneg a_1'\rd a_1$ and $\myneg a_2'\rd a_2$ by Lemma~\ref{ce.basic}(1) we have $a_1\in\ce_{\Gamma}(\bar{d})$ and $a_2\in\ce_{\Gamma}(\bar{e})$. % chktex 36
\end{enumerate}
\item[($\gamma$)]
We have to distinguish three cases:
\begin{enumerate}
\item
If $(a\,b)\rd\prdef{x}{a_1}{a_2}{c}$ for some $a_1$, $a_2$, and $c$ then, since $\Gamma\sngv(a\,b)$, by Lemma~\ref{nrm.rd} we have
\[\nrm{\Gamma}{(a\,b)}=[\nrm{\Gamma}{a_1},\nrm{\Gamma}{a_2}]=[\bar{d},\bar{e}]
\]
Similarly to the corresponding case for condition $\alpha$ we can shown that $\prdef{x}{a_1}{a_2}{c}\in\ce_{\Gamma}(\bar{c})$.
By computability condition $\gamma$ we know that $a_1\in\ce_{\Gamma}(\bar{d})$ and $a_2\in\ce_{\Gamma}(\bar{e})$.
\item
If $(a\,b)\rd\injl{a_1}{c}$ for some $a_1$ and $c$, then, since $\Gamma\sngv(a\:b)$, by Lemma~\ref{nrm.rd} we have
\[\nrm{\Gamma}{(a\,b)}=[\nrm{\Gamma}{a_1},\nrm{\Gamma}{c}]=[\bar{d},\bar{e}]
\]
Similarly to the corresponding case for condition $\alpha$ we can shown that $\injl{a_1}{c}\in\ce_{\Gamma}(\bar{c})$.
By computability condition $\gamma$ we obtain $a_1\in\ce_{\Gamma}(\bar{d})$.
\item
$(a\,b)\rd\injr{c}{a_2}$ for some $a_2$ and $c$:
This case is similar to the previous one.
\end{enumerate}
\item[($\delta$)]
Let $\bar{d}=[\bar{d}_1,\bar{d}_2]$ for some $\bar{d}_1$ and $\bar{d}_2$.
If $(a\,b)\rd\case{a_1}{a_2}$ for some $a_1$ and $a_2$ then, since $\Gamma\sngv(a\,b)$, by Lemma~\ref{nrm.rd} we have
\[\nrm{\Gamma}{\case{a_1}{a_2}}=[\bar{d},\bar{e}]
\]
where $\nrm{\Gamma}{a_1}=[\bar{d}_1,\bar{e}]$ and $\nrm{\Gamma}{a_2}=[\bar{d}_2,\bar{e}]$.
Similarly to the corresponding case for condition $\alpha$ we can shown that $\case{a_1}{a_2}\in\ce_{\Gamma}(\bar{c})$.
By computability condition $\delta$ we obtain $a_1\in\ce_{\Gamma}([\bar{d}_1,\bar{e}])$ and $a_2\in\ce_{\Gamma}([\bar{d}_2,\bar{e}])$.
\qedhere
\end{itemize}
\end{proof}
\noindent
We now show the remaining closure properties of computable expressions.
\begin{lem}[Closure properties of computable expressions]%
\label{ce.mon}
For all $\Gamma,x,a,b,c,\bar{a},\bar{b},\bar{c}$:
\begin{enumerate}
\item
$a\in\ce_{\Gamma}(\bar{a})$, $b\in\ce_{\Gamma}(\bar{b})$, and $c\in\ce_{\Gamma,x:a}(\bar{b})$ implies $\prsumop{a}{b},\injl{a}{b},\injr{a}{b}\in\ce_{\Gamma}([\bar{a},\bar{b}])$, and $\prdef{x}{a}{b}{c}\in\ce_{\Gamma}([\bar{a},\bar{b}])$.
\item
$a\in\ce_{\Gamma}([\bar{b},\bar{c}])$ implies $\pleft{a}\in\ce_{\Gamma}(\bar{b})$ and $\pright{a}\in\ce_{\Gamma}(\bar{c})$.
\item
$a\in\ce_{\Gamma}([\bar{a},\bar{c}])$ and $b\in\ce_{\Gamma}([\bar{b},\bar{c}])$ implies $\case{a}{b}\in\ce_{\Gamma}([[\bar{a},\bar{b}],\bar{c}])$.
\end{enumerate}
\end{lem}
\begin{proof}
For all $\Gamma$, $x$, $a$, $b$, $c$, $\bar{a}$, $\bar{b}$, $\bar{c}$:
\begin{enumerate}
\item
Let $a\in\ce_{\Gamma}(\bar{a})$, $b\in\ce_{\Gamma}(\bar{b})$, and $c\in\ce_{\Gamma,x:a}(\bar{b})$.
Hence $\bar{a}=\nrm{\Gamma}{a}$ and $\bar{b}=\nrm{\Gamma}{b}$.

First we show that $[a,b],[a+b],\injl{a}{b},\injr{a}{b}\in\ce_{\Gamma}([\bar{a},\bar{b}])$.
Obviously $\Gamma\sngv[a,b],[a+b],\injl{a}{b},\injr{a}{b}$ and $[a,b],[a+b],\injl{a}{b},\injr{a}{b}\in\sn{}$.
Using Lemmas~\ref{rd.decomp},~\ref{rd.decomp.neg},~\ref{ce.basic} and~\ref{ce.mon.neg}, it is straightforward to show the computability conditions for these operators.

Next we turn to the case of a protected definition: From the assumptions we know that $\nrm{\Gamma}{\prdef{x}{a}{b}{c}}=[\bar{a},\bar{b}]$ where $\nrm{\Gamma,x:a}{c}=\bar{b}$.
%Let $\bar{a}=\nrm{\Gamma}{a}$ and $\bar{b}=\nrm{\Gamma}{b}$.
Assume $a\in\ce_{\Gamma}(\bar{a})$, $b\in\ce_{\Gamma}(\bar{b})$, and $c\in\ce_{\Gamma,x:a}(\bar{b})$.
We have to show $\prdef{x}{a}{b}{c}\in\ce_{\Gamma}([\bar{a},\bar{b}])$.

Since by definition of computable expressions $a,b,c\in\sn{}$, by Lemma~\ref{sn.basic}(2) we have $\prdef{x}{a}{b}{c}\in\sn{}$. % chktex 36
It remains to show the computability conditions:
By Lemmas~\ref{rd.decomp}(2) and~\ref{rd.decomp.neg}(5), $\prdef{x}{a}{b}{c}\rd d$ and $\myneg\prdef{x}{a}{b}{c}\rd d$, for some $d$, each imply $d=\prdef{x}{a'}{b'}{c'}$ for some $a'$, $b'$, and $c'$. % chktex 36
Therefore, the computability conditions $\alpha$, $\beta$, and $\delta$ are trivially satisfied.
The condition $\gamma$ is trivially satisfied except for the case $\prdef{x}{a}{b}{c}\rd\prdef{x}{a'}{b'}{c'}$ for some $a'$, $b'$, and $c'$.
By Lemma~\ref{rd.decomp}(2) we know that $a\rd a'$, $b\rd b'$, and $c\rd c'$. % chktex 36
By Lemma~\ref{ce.basic}(1) we know that $a'\in\ce_{\Gamma}(\bar{a})$ and $b'\in\ce_{\Gamma}(\bar{b})$. % chktex 36
\item
Let  $a\in\ce_{\Gamma}([\bar{b},\bar{c}])$:
From the assumptions we know that that $\nrm{\Gamma}{\pleft{a}}=\bar{b}$ where
$\nrm{\Gamma}{a}=[\bar{b},\bar{c}]$.
Assume $a\in\ce_{\Gamma}([\bar{b},\bar{c}])$. We have to show $\pleft{a}\in\ce_{\Gamma}(\bar{b})$.
Since $a\in\sn{}$ by Lemma~\ref{sn.basic}(2) we know that $\pleft{a}\in\sn{}$. % chktex 36
The computability conditions for $\pleft{a}$ can be shown in a similar type of argument as for the case of application (Lemma~\ref{ce.mon.appl}).
The case $\pright{a}\in\ce_{\Gamma}(\bar{c})$ can be shown in a similar style.
\item
Assume that $a\in\ce_{\Gamma}([\bar{a},\bar{c}])$ and $b\in\ce_{\Gamma}([\bar{b},\bar{c}])$.
We will show that $\case{a}{b}\in\ce_{\Gamma}([[\bar{a},\bar{b}],\bar{c}])$.
Obviously $\Gamma\sngv\case{a}{b}$ and $\case{a}{b}\in\sn{}$.
Since case distinction or negated case distinctions always reduce to case distinctions, the conditions $\alpha$, $\beta$, and $\gamma$ are trivially satisfied.
It remains to show the condition $\delta$:
Assume that $\case{a}{b}\rd\case{a_1}{a_2}$.
By Lemma~\ref{rd.decomp}(1) we know that $a\rd a_1$ and $a\rd a_2$. % chktex 36
The required conclusions $a_1\in\ce_{\Gamma}([\bar{a},\bar{c}])$ and $a_2\in\ce_{\Gamma}([\bar{b},\bar{c}])$ follow by Lemma~\ref{ce.basic}(1). % chktex 36
\qedhere
\end{enumerate}
\end{proof}
\noindent
Note that abstraction is missing from the properties of Lemma~\ref{ce.mon} since from $a\in\ce_{\Gamma,x:b}(\bar{a})$ it is not clear how to conclude that $a\gsub{x}{c}\in\ce_{\Gamma}$ for any $c\in\ce_{\Gamma}(\nrm{\Gamma}{b})$.
This inspires the following Lemma.
\begin{lem}[Abstraction closure for computable expressions]%
\label{ce.abs}
For all $\Gamma$, $x$, $a$, and $b$ where $\Gamma\sngv\binbop{x}{a}{b}$:
If $a\in\ce_{\Gamma}(\nrm{\Gamma}{a})$, $b\in\ce_{\Gamma,x:a}(\nrm{\Gamma,x:a}{b})$, and for all $c$ with $c\in\ce_{\Gamma}(\nrm{\Gamma}{a})$ we have $b\gsub{x}{c}\in\ce_{\Gamma}(\nrm{\Gamma,x:a}{b})$ (this last assumption about substitution is crucial)
then $\binbop{x}{a}{b}\in\ce_{\Gamma}(\nrm{\Gamma}{\binbop{x}{a}{b}})$.
\end{lem}
\begin{proof}
Let $\bar{a}=\nrm{\Gamma}{a}$ and $\bar{b}=\nrm{\Gamma,x:a}{b}$.
From $a\in\ce_{\Gamma}(\bar{a})$ and $b\in\ce_{\Gamma,x:a}(\bar{b})$ we obviously obtain $a,b\in\sn{}$.
By Lemma~\ref{sn.basic}(2) this implies $[x:a]b\in\sn{}$. % chktex 36

To show that $\binbop{x}{a}{b}\in\ce_{\Gamma}([\bar{a},\bar{b}])$,
it remains to show the computability conditions for $\binbop{x}{a}{b}$:
The computability conditions $\beta$, $\gamma$, and $\delta$ are trivially satisfied since (negated) abstractions reduce to abstractions only.
It remains to show the condition $\alpha$:
By Lemma~\ref{rd.decomp}(2), only the cases $[x:a]b\rd[x:c]d$, $[x!a]b\rd[x!c]d$, $\myneg[x:a]b\rd[x!c]d$, and $\myneg[x!a]b\rd[x:c]d$ for some $c$ and $d$ are possible. % chktex 36
By Lemmas~\ref{rd.decomp}(2) and~\ref{rd.decomp.neg}(1,3) we know that $a\rd c$ and either $b\rd d$ or $\myneg b\rd d$. % chktex 36
We have to show the following properties:
\begin{itemize}
\item
We have to show $d\in\ce_{\Gamma,x:c}(\bar{b})$:
From the assumption $b\in\ce_{\Gamma,x:a}(\bar{b})$ either by directly using Lemma~\ref{ce.basic}(1) or by first applying Lemma~\ref{ce.mon.neg} we obtain $d\in\ce_{\Gamma,x:a}(\bar{b})$. % chktex 36
Since $a\rd c$, by Lemma~\ref{ce.basic}(2) we obtain that $d\in\ce_{\Gamma,x:c}(\bar{b})$. % chktex 36
\item
Let $e\in\ce_{\Gamma}(\bar{a})$.
We have to show that $d\gsub{x}{e}\in\ce_{\Gamma}(\bar{b})$ which follows from the assumption about substitution (by instantiating $c$ to $e$).
\qedhere
\end{itemize}
\end{proof}
\noindent
To prove computability of all normable expressions, due to Lemma~\ref{ce.abs}, we need to prove the stronger property
that normable expressions are computable under any substitution of their free variables to computable expressions.
First we need to extend the notion of substitution.
\begin{defi}[Extended substitution]
The substitution operation $a\gsub{x}{b}$ to replace free occurrences of $x$ in $a$ by $b$ can be extended as follows:
Given sequences of pairwise disjoint variables $X=(x_1,\ldots,x_n)$ and (arbitrary) expressions $B=(b_1,\ldots,b_n)$ where $n\geq 0$, a \emph{substitution function} $\sigma_{X,B}$ is defined on expressions and contexts as follows:
\begin{eqnarray*}
\sigma_{X,B}(a)&=&a\gsub{x_1}{b_1}\ldots\gsub{x_n}{b_n}\\[4mm]
\sigma_{X,B}(())&=&()\\
\sigma_{X,B}(x:a,\Gamma)&=&
\begin{cases}
(x:\sigma_{X,B}(a),\sigma_{X,B}(\Gamma))&\text{if}\;x\neq x_i\\
\sigma_{X,B}(\Gamma)&\text{otherwise}
\end{cases}
\end{eqnarray*}
If $x\neq x_i$ we write $\sigma_{X,B}\gsub{x}{b}$ for $\sigma_{(x_1,\ldots,x_n,x),(b_1,\ldots,b_n,b)}$.
\end{defi}
\begin{defi}[Norm-matching substitution]
A substitution $\sigma_{X,B}$  where $X=(x_1,\ldots,x_n)$ and $B=(b_1,\ldots,b_n)$ is called  \emph{norm matching w.r.t.~}$\Gamma$ iff $\Gamma=(\Gamma_0,x_1:a_1,\Gamma_1\ldots x_n:a_n,\Gamma_n)$, for some $\Gamma_0$ and $a_i$, $\Gamma_i$ where $1\leq i\leq n$ and furthermore
for all these $i$ we have, with $\sigma$ abbreviating $\sigma_{X,B}$:
\[
\sigma(\Gamma)\sngv \sigma(a_i),\quad\sigma(\Gamma)\sngv\sigma(b_i)\quad\text{and}\quad\nrm{\sigma(\Gamma)}{\sigma(a_i)}=\nrm{\sigma(\Gamma)}{\sigma(b_i)}
\]
\end{defi}
\noindent
Norm-matching substitutions indeed preserve norms:
\begin{lem}[Norm preservation of norm-matching substitutions]%
\label{sub.nrm}
Let $\sigma_{X,B}$ be norm-matching w.r.t.~$\Gamma$, then for all $\Gamma$ and $a$ we have $\nrm{\sigma_{X,B}(\Gamma)}{\sigma_{X,B}(a)}=\nrm{\Gamma}{a}$.
\end{lem}
\begin{proof}
Proof is by induction on the number $k$ of variables in $X$ and follows directly from the definition of norm-matching substitution.
\end{proof}
\noindent
As a consequence of Lemma~\ref{sub.nrm}, for any norm-matching substitution $\sigma_{X,B}$ we have $\Gamma\sngv a$ if and only if $\sigma_{X,B}(\Gamma)\sngv\sigma_{X,B}$.
We now come to the last piece missing to show strong normalisation.
\begin{lem}[Normability implies computability]%
\label{nrm.ce}
For all $\Gamma$ and $a$: If $\Gamma\sngv a$ then for any norm-matching substitution $\sigma_{X,B}$ w.r.t.~$\Gamma$
with $\sigma_{X,B}(x_i)\in\ce_{\sigma_{X,B}(\Gamma)}(\nrm{\Gamma}{b_i})$, for $1\leq i\leq n$, we have $\sigma_{X,B}(a)\in\ce_{\sigma_{X,B}(\Gamma)}(\nrm{\Gamma}{a})$. As a consequence normability implies computability.
\end{lem}
\begin{proof}
Let $\bar{a}=\nrm{\Gamma}{a}$ and
let $\sigma=\sigma_{X,B}$ be a norm matching substitution w.r.t.~$\Gamma$ with $\sigma(x_i)\in\ce_{\sigma(\Gamma)}(\bar{b}_i)$ where $\bar{b}_i=\nrm{\Gamma}{b_i}$, for $1\leq i\leq n$.
The proof is by induction on $a$:
\begin{itemize}
\item
$a=\prim$:
Obviously $\sigma(\prim)=\prim$ and $\prim\in\ce_{\sigma_{X,B}(\Gamma)}(\prim)$.
\item
$a=x$:
We have $\Gamma\sngv x$ and $\bar{a}=\nrm{\Gamma}{x}$.
There are two cases:
\begin{itemize}
\item
$x=x_i$, for some $i$.
Obviously $\sigma(x)=\sigma(b_i)$.
We know that $\sigma(b_i)\in\ce_{\sigma(\Gamma)}(\bar{b}_i)$.
\begin{eqnarray*}
\bar{a}&=&\nrm{\Gamma}{x}\\
&=&\quad\text{(Lemma~\ref{sub.nrm})}\\
&&\nrm{\sigma(\Gamma)}{\sigma(x)}\\
&=&\quad\text{(property of substitution)}\\
&&\nrm{\sigma(\Gamma)}{\sigma(b_i)}\\
&=&\quad\text{(Lemma~\ref{sub.nrm})}\\
&&\nrm{\Gamma}{b_i}\\
&=&\quad\text{(definition of $\bar{b}_i$)}\\
&&\bar{b}_i
\end{eqnarray*}
Hence $\sigma(x)=\sigma(b_i)\in\ce_{\sigma(\Gamma)}(\bar{a})$.
\item
If $x\neq x_i$ then $\sigma(x)=x$. Obviously $x\in\sn{}$. From $\Gamma\sgv x$,  by Lemma~\ref{val.nrm} we obtain $\Gamma\sngv x$.
The computability conditions are trivially satisfied therefore we have $\sigma(x)\in\ce_{\sigma(\Gamma)}(\bar{a})$.
\end{itemize}
\item
$a=\binbop{x}{b}{c}$:
From $\Gamma\sngv\binbop{x}{b}{c}$ by definition of norming we obtain $\Gamma\sngv b$ and $(\Gamma,x:b)\sngv c$.
Let $\bar{b}=\nrm{\Gamma}{b}$ and $\bar{c}=\nrm{\Gamma,x:b}{c}$.
Applying the inductive hypothesis with the empty substitution $\sigma_{(),()}$, which is trivially norm-matching w.r.t.~any context, we obtain
$b\in\ce_{\Gamma}(\bar{b})$ and $c\in\ce_{\Gamma,x:b}(\bar{c})$.
By Lemma~\ref{sub.nrm} we know that $\sigma(\nrm{\sigma(\Gamma)}{b})=\nrm{\Gamma}{b}=\bar{b}$.
By inductive hypothesis $\sigma(b)\in\ce_{\sigma(\Gamma)}(\bar{b})$.

Consider an expression $d$ where $d\in\ce_{\Gamma}(\bar{b})$.
In order to apply Lemma~\ref{ce.abs} we have to show that $\sigma(c)\gsub{x}{d}\in\ce_{\sigma(\Gamma)}(\bar{c})$.
If we define $X'=(x_1,\ldots,x_n,x)$, $B'=(b_1,\ldots,b_n,d)$, and $\sigma'=\sigma\gsub{x}{d}$
then obviously $\sigma'$ is norm-matching w.r.t~$(\Gamma,x:a)$ and substitutes to computable expressions.
Therefore by inductive hypothesis for $c$ we know that $\sigma'(c)=\sigma(c)\gsub{x}{d}\in\ce_{\sigma(\Gamma)}(\bar{c})$ (obviously $\sigma'(\Gamma)=\sigma(\Gamma)$).
Hence by Lemma~\ref{ce.abs} it follows that $\binbop{x}{\sigma(b)}{\sigma(c)}\in\ce_{\sigma(\Gamma)}([\bar{b},\bar{c}])$ which by definition of substitution is equivalent to $\sigma(\binbop{x}{b}{c})\in\ce_{\sigma(\Gamma)}([\bar{b},\bar{c}])$.
\item
$a=\prdef{x}{b}{c}{d}$:
From $\Gamma\sngv\prdef{x}{a_1}{a_2}{a_3}$ by definition of norming we obtain $\Gamma\sngv a_1,a_2$, $(\Gamma,x:a_1)\sngv a_3$, and $\nrm{\Gamma}{a_2}=\nrm{\Gamma,x:a_1}{a_3}$.
Let $\bar{b}=\nrm{\Gamma}{b}$, $\bar{c}=\nrm{\Gamma}{c}$.
By inductive hypothesis with $\sigma$ we know that $\sigma_{X,B}(b)\in\ce_{\sigma(\Gamma)}(\bar{b})$, $\sigma(c)\in\ce_{\sigma(\Gamma)}(\bar{c})$, and
$\sigma_{X,B}(d)\in\ce_{\sigma(\Gamma,x:b)}(\bar{c})=\ce_{(\sigma(\Gamma),x:\sigma(b))}(\bar{c})$.
By Lemma~\ref{ce.mon}(1) we then obtain  $\prdef{x}{\sigma(b)}{\sigma(c)}{\sigma(d)}\in\ce_{\sigma(\Gamma)}([\bar{b},\bar{c}])$. % chktex 36
By definition of substitution this is equivalent to $\sigma(\prdef{x}{b}{c}{d})\in\ce_{\sigma(\Gamma)}([\bar{b},\bar{c}])$.
\item
$a=\binop{a_1}{\ldots,a_n}$:
We have $\Gamma\sngv a$ and $\Gamma\sngv a_i$.
Let $\bar{a}=\nrm{\Gamma}{a}$, $\bar{a}_i=\nrm{\Gamma}{a_i}$.
By inductive hypothesis $\sigma(a_i)\in\ce_{\sigma(\Gamma)}(\bar{a}_i)$.
By Lemmas~\ref{ce.mon.neg},~\ref{ce.mon.appl}, and the various cases of Lemma~\ref{ce.mon} we obtain $\binop{\sigma(a_1)}{\ldots,\sigma(a_n)}\in\ce_{\sigma(\Gamma)}(\bar{a})$.
By definition of substitution this is equivalent to $\sigma(\binop{a_1}{\ldots,a_n})\in\ce_{\sigma(\Gamma)}(\bar{a})$.
\end{itemize}
For the consequence, we just take the empty substitution $\sigma_{(),()}$ which is trivially norm-matching.
\end{proof}
\noindent
Hence we know that normability implies computability.
\noindent
Lemmas~\ref{val.nrm} and~\ref{nrm.ce} together yield strong normalisation.
\begin{thm}[Strong normalization of valid expressions]%
\label{sn.valid}
For all $\Gamma$ and expressions $a$: $\Gamma\sgv a$ implies $a\in\sn{}$.
\end{thm}
\begin{proof}
Assume $\Gamma\sgv a$.
By Lemma~\ref{val.nrm} this implies $\Gamma\sngv a$.
By Lemma~\ref{nrm.ce} this implies $a\in\ce_{\Gamma}(\nrm{\Gamma}{a})$.
Obviously this implies $a\in\sn{}$.
\end{proof}
\begin{defi}[Normal form]
If $a\in\sn{}$ then $\nf(a)$ denotes the unique expression to which $a$ is maximally reducible.
Note that this definition is well-founded due to confluence of $\srd$ (Theorem~\ref{rd.confl}).
Furthermore, due to strong normalisation (Theorem~\ref{sn.valid}), $\Gamma\sgv a$ implies that $\nf(a)$ exists.
\end{defi}
\noindent
We note an easy consequence of confluence and strong normalization.
\begin{cor}[Decidability of the type relation]%
\label{decide.dtyp}
For any expression $a$ and context $\Gamma$ there is a terminating algorithm such that $\Gamma\sgv a$ iff the algorithm is not failing but computing an expression $b$ with $\Gamma\sgv a:b$.
\end{cor}
\begin{proof}
The algorithm to attempt to compute a type $b$ of $a$ is recursive on $a$ under the context $\Gamma$.
It basically works by checking the type conditions on unique normal forms.
\end{proof}
\subsection{Consistency}%
\label{regular}
Due to confluence and the strong normalization result for valid expressions it is often sufficient to consider the normal form $\nf(a)$ instead of the expression $a$ itself when proving properties about expressions of \dcalc.
In this section, we study this more rigorously and use a characterization of valid normal forms to show consistency of \dcalc.
\begin{defi}[Valid normal forms]%
\label{vnf}
The set of \emph{valid normal forms} is a subset of $\dexp$ and denoted by $\nf$.
The recursive characterization of valid normal forms in Table~\ref{vnf.set} also uses the auxiliary set of \emph{dead ends} denoted by $\de$.
\begin{table}[!htb]
\fbox{
\begin{minipage}{0.96\textwidth}
\begin{eqnarray*}
\\[-9mm]
\nf&=&\{\prim\}\;\union\;\{[x:a]b,[x!a]b,\prdef{x}{a}{b}{c}\sth a,b,c\in\nf\}\\
&&\;\union\;\{[a,b],[a+b],\injl{a}{b},\injr{a}{b},\case{a}{b}\sth a,b\in\nf\}\;\union\;\de\\
\de&=&\{x\sth x\in\dvar\}\;\union\;\{(a\,b),\pleft a, \pright a,(\case{b}{c}\,a)\sth a\in\de,b,c\in\nf\}\\
&&\;\union\;\{\myneg a\sth a\in\de,a\;\text{is not a negation}\}
\end{eqnarray*}
\end{minipage}
}
\caption{Valid normal forms\label{vnf.set}}
\end{table}
\end{defi}
\noindent
$\nf$ indeed characterizes the normal forms of valid expressions.
\begin{lem}[Normal forms of valid expressions]%
\label{vnf.basic}
For all $a$ where $\sgv a$ we have that $a\in\nf$ iff $\nf(a)=a$
\end{lem}
\begin{proof}
Obviously, by construction, all elements of $\nf$ are irreducible.
For the reverse, we prove the more general property that for all $a$ and $\Gamma$ with $\Gamma\sgv a$ we have that $a\in\nf$ implies $\nf(a)=a$.
The proof is by induction on $a$.
\end{proof}
\noindent
We need a couple of easy lemmas for the consistency proof.
\begin{lem}[Valid normal forms of universal abstraction type]%
\label{nf.regular}
For all $x$, $a$, $b$, and $c$: If $a\in\nf$ and $\sgv a:[x:b]c$ then there is some $d\in\nf$ such that $a=\binbop{x}{b}{d}$ and $x:b\sgv d:c$.
\end{lem}
\begin{proof}
Since $\free(a)=\emptyset$, by definition of $\de$ we know that $a\notin\de$.
We need to check the following remaining cases
\begin{enumerate}
\item
$a$ is $\prim$, a protected definition, an injection, a sum, or a product, and $\sgv a:[x:b]c$:
By definition of typing and properties of reduction this cannot be the case.
\item
$a=\binbop{x}{a_1}{a_2}$.
%By Law~\ref{val.decomp}(2) we know that $\sgv a_1$ and $x:a_1\sgv a_2$.
From  $\sgv\binbop{x}{a_1}{a_2}:[x:b]c$,  by Lemma~\ref{type.decomp}(1) we have $x:a_1\sgv a_2:d'$ for some $d'$ where $[x:a_1]d'\eqv[x:b]c$. % chktex 36
Hence, by Lemma~\ref{congr.basic}(2) we know that $a_1\eqv b$ and $d'\eqv c$. % chktex 36
From $\sgv a:[x:b]c$, by Lemmas~\ref{valid.type} and~\ref{val.decomp}(2) we know that $\sgv b$ and $x:b\sgv c$. % chktex 36
Therefore, since  $x:a_1\sgv a_2:d'$ and $a_1\eqv b$, by Lemma~\ref{eqv.env} we know that $x:b\sgv a_2:d'$.
Similarly since $d'\eqv c$, by type rule \myconv~we know that $x:b\sgv a_2:c$.
Hence we have shown the property with $d=a_2$.
\qedhere
\end{enumerate}
\end{proof}
\noindent
We also need the following obvious variation of Lemma~\ref{type.decomp}(1). % chktex 36
\begin{lem}[Abstraction property]%
\label{abs.decomp}
For all $\Gamma$, $x$, $a$, $b$, and $c$: If $\Gamma\sgv\binbop{x}{a}{b}:[x:a]c$ then $\Gamma,x:a\sgv b:c$.
\end{lem}
\begin{proof}
By Lemma~\ref{type.decomp}(1) we know that $\Gamma,x:a\sgv b:d$ for some $d$ with $[x:a]d\eqv [x:a]c$. % chktex 36
By Lemma~\ref{congr.basic}(2) we know that $d\eqv c$. % chktex 36
By Lemma~\ref{valid.type} we know that $\sgv [x:a]c$.
By Lemma~\ref{val.decomp}(2) we know that $x:a\sgv c$. % chktex 36
Hence by rule \myconv~we can infer that $\Gamma,x:a\sgv b:c$.
\end{proof}
\noindent
Finally, we need a lemma for the case of a declaration $x:\prim$.
\begin{lem}[$\prim$-Declaration property]%
\label{primdec}
For all $a$ and $b$: $x:\prim\sgv a:b$ implies that $b\neqv x$.
\end{lem}
\begin{proof}
Due to Lemma~\ref{rd.type} (subject reduction) we may assume that $a\in\nf$.
The various cases of Lemma~\ref{type.decomp} imply that $b\neqv x$ in case $a$ is an abstraction, a sum, a product, a protected definition, a case distinction, or an injection.
Hence by definition of $\nf$, it remains to look at the case $a\in\de$.
By definition of $\de$, since $x:\prim\sgv a$, $a$ can only be a variable $x$ or a negated variable $\myneg x$ and therefore obviously $b\eqv\prim$.
The property follows since $\prim\neqv x$.
\end{proof}
\begin{thm}[Consistency of \dcalc]%
\label{cons}
There is no expression $a$ such that $\sgv a:[x:\prim]x$.
\end{thm}
\begin{proof}
Assume that there is an expression $a$ with $\sgv a:[x:\prim]x$.
By Theorem~\ref{sn.valid} (strong normalization) we know that there is a normal form $a':=\nf(a)$ with $a\rd a'$.
By Lemma~\ref{valid.rd} we know that $\sgv a'$.
By Theorem~\ref{rd.type} (subject reduction) we know that $\sgv a':[x:\prim]x$.

By Lemma~\ref{vnf.basic} we know that $a'\in\nf$, hence, by Lemma~\ref{nf.regular} there is an expression $c$ where $\sgv\binbop{x}{\prim}{c}:[x:\prim]x$.
By Lemma~\ref{abs.decomp} this implies $x:\prim\sgv c:x$.
By Lemma~\ref{primdec} this implies that $x\neqv x$.
Thus we have inferred a contradiction and therefore the proposition is true.
\end{proof}
\begin{rem}[Limitations of the consistency result]%
\label{cons.negation}
Lemma~\ref{cons} shows that there is no inherent flaw in the type mechanism of \dcalc\ by which one could prove anything from nothing.
Note that the consistency result is limited to empty contexts, hence it does not cover the use of negation or casting axioms (see Section~\ref{typing}).
Furthermore, it remains an open issue if we can generalize the empty type $[x:\prim]x$ to $[x:a]x$ with $\sgv a$.
\end{rem}
\section{Comparison to other systems and possible variations and extensions}%
\label{related}
Due to the use of $\lambda$-structured types, \dcalc\ falls outside the scope of PTS (see \eg~\cite{Bar:93}).
In Section~\ref{overview} we have indicated the differences between the core of \dcalc\ and PTS\@.
Due to its origins from $\lambda^{\lambda}$ and its use of a reflexive type axiom \dcalc\ does not use the concept of dependent product and it does not use a type relation that can intuitively be interpreted as set membership.
Instead \dcalc\ introduces a number of operators which can be functionally interpreted by untyped $\lambda$-expressions, \eg~by stripping of the type tags and negations, by interpreting sums, products, and protected definitions as binary pairs and both universal and existential abstraction as $\lambda$-abstraction. Of course other interpretations are possible which preserve more semantic detail.
In any case, the type rules of \dcalc~would then induce a relation between untyped $\lambda$-expressions.
% which can be interpreted as the relation between proof and the property that is proven.
Furthermore, \dcalc\ has computationally-irrelevant proofs, i.e.~it is not possible to extract for example primitive recursive functions from valid expressions.
%This constitutes a major difference to well-known type systems.
In this section we discuss the use of \dcalc\ as a logic and then discuss several extensions of \dcalc.
\subsection{Encoding of logical operators}%
\label{opencoding}
In $\lambda^{\lambda}$, common encodings of logical operators can be used (see Table~\ref{encoding}, where $[a\fun b]$ abbreviates $[x:a]b$ if $x$ does not occur free in $b$).
\begin{table}[!htb]
\fbox{
\begin{minipage}{0.96\textwidth}
\begin{eqnarray*}
\\[-8mm]
\mathrm{false}&:=&[x:\prim]x \\
\mathrm{true}&:=&[x:\prim][y:x]y\\
\mathrm{implies}&:=&[x:\prim][y:\prim][x\fun y]\\
\mathrm{not}&:=&[x:\prim][x\fun\mathrm{false}]\\
\mathrm{and}&:=&[x:\prim][y:\prim][z:\prim][[x\fun[y\fun z]]\fun z]\\
\mathrm{or}&:=&[x:\prim][y:\prim][z:\prim][[x\fun z]\fun[[y\fun z]\fun z]]\\
\mathrm{forall}&:=&[x:\prim][y:[x\fun\prim]][z:x](y\,z)\\
\mathrm{exists}&:=&[x:\prim][y:[x\fun\prim]][[z:x][(y\,z)\fun x]\fun x]
\end{eqnarray*}
\end{minipage}
}
\caption{Encoding of logical operators\label{encoding}}
\end{table}
Hence one could argue that no further logical operators (apart from the law of the excluded middle) seem necessary.
We do not follow this argument in \dcalc\ because
%, as we have already argued in Section~\ref{overview}, due to the typing rules of the core of \dcalc\ declarations such as $x:\prim$ cannot be instantiated to functions $[x:a]b$.
%This limitation of expressive power is a drawback of such encodings.
%To overcome this limitation, we could introduce abbreviations for expression schemas, \eg~
%\begin{eqnarray*}
%\mathrm{and}_P(a,b)&:=&[z:P][[a\fun[b\fun z]]\fun z]
%\end{eqnarray*}
%However, regardless of the use of schemas we do not adopt the encoding approach in \dcalc\ because
of properties such as
\[
\frac{c:a\quad d:b}{[z:\prim][x:[a\fun[b\fun z]]]((x\,c)\,d):\mathrm{and}(a,b)}
\]
which we consider less intuitive for deductions involving conjunction  as compared to the approach in \dcalc:
\[
\frac{c:a\quad d:b}{[c,d]:[a,b]}
\]
\subsection{Logical interpretation}%
\label{related.hol}
\dcalc\ is treating proofs and formulas uniformly as typed $\lambda$-expressions, and allows each of its operators to be used on both sides of the type relation.
A subset of the operators of \dcalc, if used as types, can be associated with common logical predicates and connectors:
\begin{eqnarray*}
\prim,\;(x\,a_1\ldots a_n)&\simeq&\text{atomic formulas}\\{}
[x:a]b&\simeq&\text{universal quantification}\\{}
[x!a]b&\simeq&\text{existential quantification}\\{}
[a,b]&\simeq&\text{conjunction}\\{}
[a+b]&\simeq&\text{disjunction}\\{}
\myneg a&\simeq&\text{negation}
\end{eqnarray*}
In Sections~\ref{overview} and~\ref{examples.logic} we have shown that on the basis of the type system of \dcalc\ many logical properties of these connectors can be derived without further assumptions.
Furthermore, on the basis of a strong normalization result (\ref{cons}) we have shown  that \dcalc\ is consistent in the sense that the type $[x:\prim]x$ is empty in \dcalc\ under the empty context.
In this sense, \dcalc\ can be seen as a logic where typing can be interpreted as the relation between a deduction and the proposition it has shown~\cite{Howard69}.

However, in order to have the complete properties of negation, additional axioms have to be assumed and
we did not show consistency of the type system under these axioms by means of strong normalization.
Similarly, formalization of mathematical structures in  \dcalc\ must be done axiomatically.
In this sense the expressive power of \dcalc\ is limited and each axiomatization has to be checked for consistency.
Furthermore, there are two important pragmatic issues which differ from common approaches:
\begin{itemize}
\item
First, inference systems for higher-order logic on the basis of typed-$\lambda$-calculus such as~\cite{church40,Paulson1989} typically make a distinction between the type of propositions
and one or more types of \emph{individuals}. In \dcalc, one the one hand there is no such distinction, all such types must either be $\prim$ itself or declared using $\prim$.
On the other hand, due to the restricted formation rules which serve to ensure consistency as well as uniqueness of types, in \dcalc, $\prim$ does not allow to quantify over all propositions of \dcalc\ and additional axioms schemes must be used when reasoning with formulas of complex structure.
\item
Second, \dcalc\ has several operators which are not common logical connectors:
protected definitions $\prdef{x}{a}{b}{c}$, projections $\pleft{a}$, $\pright{a}$, case distinction $\case{a}{b}$, as well as left and right injection
$\injl{a}{b}$, and $\injr{a}{b}$.
However, these operators have meaningful type-roles for defining functions over propositions.
\end{itemize}
Note also that there is a strong relation between left projection $\pleft{a}$ on a deduction $a$ and Hilbert's $\epsilon$-operator $\epsilon x.P$ on a formula $P$ as sometimes used in higher-order logic~\cite{church40,Paulson1989} with a law like:
\[
\forall x.(P \Rightarrow P\gsub{x}{\epsilon x.P})
\]
In a classical logical setting this is obviously implied by
\[
(\exists x.P) \Rightarrow P\gsub{x}{\epsilon x.P}
\]
The latter property can be approximated in \dcalc\ by
\[
[y:[x!\prim](P\,x)](P\,\pleft{y})
\]
and actually is a law since
\[
[y:[x!\prim](P\,x)]\pright{y}\;:\;[y:[x!\prim](P\,x)](P\,\pleft{y})
\]
This illustrates again how existential abstraction and the projection operators together embody a strong axiom of choice.

Finally there is no equality operator in \dcalc, a notion of equality is defined indirectly only through congruence of expressions.
\subsection{Negation}%
\label{related.negation}
In \dcalc\, rather than defining negation by implication to falsehood,
negation is incorporated by defining a subset of the equivalence laws of negation as equalities ($\eqv$).
The purpose is to have unique formal forms with respect to negation in order to simplify deductions.
A direct isomorphism between $\myneg\myneg a$ and $a$ has been advocated in~\cite{Munch2014}.
De Morgan laws for propositional operators have been used to define an involutive negation in a type language~\cite{Barbanera:1996}.

As shown in Section~\ref{examples.logic}, additional axioms schemes must be assumed to have the full set of properties of logical negation.
While this is adequate for deductive reasoning where proofs are not computationally relevant, it leaves open the issue of consistency of the axiomatic extensions, i.e.~the question if typing with additional negation axioms is consistent.

Several approaches have been proposed to internalize classical reasoning into $\lambda$-calculus.
$\lambda\mu$-calculus~\cite{Parigot2000,DBLP:conf/icfp/CurienH00} is an extension of $\lambda$-calculus formalizing inference in classical natural deduction by additional operators that give explicit control over the context in which specified subexpressions are evaluated.
Classical natural deduction has apparently not yet been studied in the context of lambda-typed systems.
In order to prove consistency for \dcalc\ with negation axioms it would seem an interesting step to try to extend its reduction through appropriate control operators.
\subsection{Uniqueness of types}%
\label{related.unique}
Note that uniqueness of types (\ref{type.confl}) was not needed in the strong normalisation proof but only the weaker property~\ref{val.nrm}.

Protected definitions $\prdef{x}{a}{c}{d}$ carry a type tag $d$ allowed to use $x$ in order to ensure uniqueness of types.
Law~\ref{val.nrm} would be retained if we remove the type tag and the binding of $x$ from protected definitions:
\[
\frac{\Gamma\sgv a: b\quad\Gamma\sgv c:d\gsub{x}{a}\quad\Gamma,x:b\sgv d:e}{\Gamma\sgv[\_\!\doteq\!a,c]:[x!b]d}
\]
However, this may lead to a significant number of type variants of a protected definition as in the following example where the four possible types are separated by comma:
\[
\frac{\Gamma\sgv a:b\quad\Gamma\sgv c:[d\fun d]\quad\Gamma,x:b\sgv d:e}
{\Gamma\sgv [\_\!\doteq\!a,c]:[x!b][d\fun d],\;[x!b][x\fun d],\;[x!b][d\fun x],\;[x!b][x\fun x]}
\]
In case of universal abstractions the situation is fundamentally different:
Adding the following type rule for universal abstractions
\[
\frac{\Gamma,x:a\sgv b:\prim}{\Gamma\sgv [x:a]b:\prim}
\]
would violate both~\ref{type.confl} and~\ref{val.nrm} and together with $\sgv\prim:\prim$ result in a paradoxical system~\cite{dcalculus}.
\subsection{Abbreviation systems}%
\label{related.definition}
Complex systems of abbreviations spanning over several abstraction levels play a major conceptual role in mathematical work.
A multitude of proposals for incorporating definitions into typed $\lambda$-calculi have been made, \eg~\cite{deBruijn94}~\cite{Guidi09}~\cite{KAMAREDDINE1999}.
Support for definitional extensions for systems closely related to Automath's $\Lambda$ have been investigated in~\cite{PdG91}.
Note that definitional extensions would allow for a relaxation of a premise of the type rule for existential abstractions.
Recall the introduction rule for existential abstractions (\pdef) which does not allow free occurrences of $x$ in $c$.
\[
\pdefm\qquad\frac{\Gamma\sgv a: b\quad\Gamma\sgv c:d\gsub{x}{a}\quad\Gamma,x:b\sgv d:e}{\Gamma\sgv\prdef{x}{a}{c}{d}:[x!b]d}
\]
Since $c$ may use $a$ in its type, it could also use an abbreviation $x$ of $a$. Hence it seems plausible that also $c$ itself may use an abbreviation $x$ of $a$.
In a calculus that includes definitional extensions this dependency could be modelled. In \dcalc\ this is not possible and maximally unfolded expressions must be used.

While support for definitional extensions is undoubtedly important, in our setting, they have not been necessary to formulate \dcalc.
Note that in other settings this might be different and abbreviation systems become indispensable, \eg~\cite{Luo2003}.
%In our case, these concepts must of course play a major role in any practically useful approach for formal deductions on the basis of \dcalc.
\section{Acknowledgements}
My thanks go to the anonymous reviewers of this paper for their detailed and constructive comments.
%

%%
%% Bibliography
%%

%% Either use bibtex (recommended),
%\bibliography{deductica}

%% .. or use the thebibliography environment explicitely

\end{document}